\documentclass[a4paper,12pt]{amsart}
\usepackage{amsaddr}
\usepackage{latexsym}
\usepackage{enumitem}
\usepackage{amsmath, amssymb, graphics}
\usepackage{amsfonts}
\usepackage{mathtools}
\usepackage{setspace} 
\usepackage{capt-of}
\usepackage{natbib}
\usepackage{booktabs}  
\usepackage{hyperref}
\usepackage{array}
\usepackage{leftidx}
\usepackage{psfrag}
\usepackage{float}
\usepackage{color}
\usepackage{ulem}

\newcommand{\E}{\mathbb{E}}
\newcommand{\N}{\mathbb{N}}
\newcommand{\Z}{\mathbb{Z}_{+}}
\newcommand{\R}{\mathbb{R}_{+}}
\newcommand{\iid}{i.i.d.}
\newcommand{\ffrac}[2]{\ensuremath{\frac{\displaystyle #1}{\displaystyle #2}}}

\usepackage{marginnote} 
   
\newcommand\preq{\stackrel{\scriptstyle \mathclap{\mbox{\tiny d}}}{=}}

\newcommand\mn[1]{
}

\newcommand\omn[1]{
}

\newcommand\bibitemarg[1]{

}

\newtheorem{Theorem}{Theorem}[section]
\newtheorem{Proposition}{Proposition}[section]
\newtheorem{Corollary}{Corollary}[section]

\theoremstyle{definition}

\theoremstyle{remark} 
\newtheorem{Remark}{Remark}[section]
\newtheorem{Example}{Example}[section]

\begin{document}

\title[Discrete distributions and 
statistical mechanics of small systems]{Discrete probability distributions and 
statistical mechanics of small systems}
\author[L.B. Klebanov]{Lev B. Klebanov}
\address{Department of Probability and Mathematical Statistics, Faculty of Mathematics and Physics, Charles University, Sokolovská 49/83, 186 75 Prague, Czech Republic}
\email{Lev.Klebanov@mff.cuni.cz}
\author[M. \v{S}umbera]{Michal \v{S}umbera}
\address{Nuclear Physics Institute, Czech Academy of Sciences, 25068 \v{R}e\v{z}, Czech Republic}
\email{sumbera@ujf.cas.cz}

\date{\today}

\begin{abstract}
We study connections between discrete probability distributions and the statistical mechanics of small systems. Using probability generating functions, we develop the theory of power series, infinitely divisible, and scalable distributions of non-negative integer-valued random variables, and introduce the class of Markovian distributions that arise naturally in stationary solutions of birth-death processes and in scalable infinitely divisible distributions. 

These results are applied to the grand canonical ensemble description in statistical mechanics: the infinite divisibility leads to a quasiparticle picture of an interacting gas, and the virial expansion is linked to the combinants of the distribution. A kinetic model of the liquid-vapor phase transition is presented, in which the particle-number distribution at the critical point converges to the Discrete Stable distribution -- the fixed point of a renormalization semi-group transformation. We also show that the scalability of the particle-number distribution is preserved in Tsallis non-extensive thermodynamics, even though infinite divisibility fails. Moreover, deformed discrete distributions, including the negative binomial as a q-deformation of the Poisson, arise naturally in this setting.
\end{abstract}

\maketitle
{\bf Keywords:} power series distributions; infinite divisibility; scalability; combinants; Markovian distributions; birth-death processes; grand canonical ensemble; virial expansion; critical point; discrete stable distribution; Tsallis thermodynamics; non-extensive entropy; small systems.
\newpage

\tableofcontents

\newpage
\begin{flushright}
{\it Motto: I wanted to learn about elementary particles \\
by studying boiling water.
(A. M. Polyakov)}  
\end{flushright}

\section{Introduction}

Classical statistical mechanics typically deals with the macroscopic behavior of physical systems with an enormous number of degrees of freedom ( see, e.g., \cite{Hill, Lanlif, BlanTh}). To extend its validity to small (non-macroscopic) systems, one needs to study ensembles containing a countable, finite number of particles and corresponding modifications of the thermodynamic functions and interrelationships, including, in general, variations in the size of the system \cite{Hill1962, Hill_small}. Noticing that the true thermal equilibrium of most systems often requires a subdivision into regions described by the fully-open generalized ensemble - the grand canonical isothermal–isobaric ensemble - one frequently resorts to nanothermodynamics \cite{chamberlin2020}.

The small system thermodynamics \cite{Hill_small} describes, among others, colloidal particles, polymers, or macromolecules; but can also be used for systems formed in collisions of ultra-relativistic nuclei \cite{Hagedorn:1985js, Csernai,  Rafelski:2015xej}. There, the final state particles emerge fully thermalized on microscopic timescales \cite{Baier:2000sb, Xu:2004mz, Pasechnik:2016wkt}. Although such systems are usually described with the grand canonical formalism, the particle-number distribution -- except in a few cases (see, e.g., \cite{Mrowczynski:1984tj, Gorenstein:2008nr}) -- is rarely analyzed in depth. Instead, it is assumed to be Poissonian, which makes it practically impossible to describe the observed multi-particle correlations.
Moreover, ignoring the underlying form of a discrete distribution may be inconsistent with the well-defined behavior of thermodynamic quantities near the phase transition. 

This article aims to fill this gap by applying the recent results on discrete probability distributions to the statistical mechanics of small systems.  Our goal is to construct an equation of state based on the general grand canonical formalism, starting from a model assumption concerning a particle-number probability distribution $p_k=\mathbb{P}(N\!=\!k)$, rather than relying on the uncertain dynamics of particle interactions via a microcanonical description.

An important property of both discrete and continuous random variables (r.v.) is their infinite divisibility, which allows their decomposition into a sum of independent identically distributed ones of arbitrary length \cite{Feller1, StvH2004, JKK2005}. Infinitely divisible distributions play a central role in various aspects of probability theory and its applications. Yet, with a few exceptions - including the classical monograph \cite{Hill} - they are not discussed in standard textbooks on statistical mechanics. 

An important subclass of infinitely divisible discrete r.v. consists of those that can be split into two independent parts, one being the fraction of the original, with both parts remaining well-defined r.v. Such random variables are called self-decomposable \cite{StvH2004}. Requiring self-decomposability into a sum of two fractions uniquely leads to the {\it Discrete stable distribution}, which will be shown to represent a fixed point of the evolution in a process with many equally important scales, as arises near the critical point of the liquid-vapor phase diagram. Moreover, certain self-decomposable variables can be scaled to have sizes larger than those of the original r.v. \cite{Klebanov2023-xj, KlSum2026}. 

The article is organized as follows. Section \ref{cons} introduces the basics of the probability theory of discrete non-negative r.v. Using the apparatus of probability generating functions, we study two special classes of distributions - power series and infinitely divisible distributions - and their interrelations. To make contact with virial expansion in statistical mechanics and to demonstrate novel applications of infinitely divisible distributions, we also discuss the combinants \cite{Kauffmann:1978vw}. We then introduce the concepts of self-decomposability and scalability.  

In Section \ref{kinmodels}, we briefly introduce kinetic models of particle-number distributions and their stationary solutions.

In Section \ref{statmech}, all the above results are applied to several topics in statistical mechanics of small systems. Using the grand canonical ensemble formalism, we establish the power-series character of the particle-number distributions and derive conditions for their infinite divisibility. We discuss the equation of state, virial expansion, and their relations to kinetic models. We also address the problem of particle-number distribution at the critical point of the vapor-liquid phase diagram. 

Section \ref{CPsec} studies the particle-number distribution in the vicinity of the critical point.

Finally, in Section \ref{TT}, we present results on particle-number distributions in Tsallis non-extensive thermodynamics \cite{Tsallis:2009zex}. We show that, contrary to infinite divisibility, the scalability of the particle-number distribution holds even in the non-extensive case. We also discuss deformed discrete distributions and demonstrate that, e.g., the negative binomial distribution arises from a deformation of the Poisson distribution. 

The properties discussed are illustrated with many examples that may prove useful in the natural and social sciences.

\newpage
\section{Constructions and Properties of Some Families of Discrete Distributions of Non-Negative Integers} \label{cons}

\subsection{Probability Mass Functions, Generating Functions and Moments} Let $N$ be a random variable (r.v.) taking non-negative integer values, $N\in \Z $ with probability $p_k=\mathbb{P}(N\!=\!k)$ for $k=0,1,2,\ldots$. The {\it probability mass function} (p.m.f.) $p_k$ is thus a sequence of non-negative numbers normalized to one, $\sum_k p_k=1$, which can be classified according to its shape. The p.m.f. is {\it log-convex} (LX), or {\it log-concave} (LC), respectively, if
\begin{equation} \label{eq:cvx}
p^2_n\leq p_{n-1} p_{n+1} \;,\;\; p^2_n\geq p_{n-1} p_{n+1}\;.
\end{equation}
Note that an LC p.m.f. is {\it unimodal}: if it were not, there would exist $ p_{n-1} > p_n < p_{n+1}$ for some $n$, which would contradict the assumed log-concavity. It is also worth mentioning that an LX p.m.f. can, in principle, have several maxima. The {\it ultra log-concave} (ULC) and {\it ultra log-convex} (ULX) distributions are also introduced in some instances (see, e.g., \cite{Johnson_2013}):
\begin{equation} \label{eq: ulc}
n p^2_n\geq (n+1) p_{n-1} p_{n+1}\;,\;\;n p^2_n\leq (n+1) p_{n-1} p_{n+1}  \;, n\geq 0\;. 
\end{equation}
Every ULX (resp. ULC) p.m.f. is LX (resp. LC) but not vice versa.

Using p.m.f. $p_k$ we define the {\it probability generating function} (p.g.f.)\footnote{Sign $:=$ denotes definition.}
\begin{equation}\label{epgf}
G(w):=\mathbb E w^N= \sum_{k=0}^{\infty}p_{k}w^{k};\;\;\; p_k=\frac{G^{(k)}(w =0)}{k!}\;.
\end{equation}
With $p_k \geq 0$ and $G(1)=\sum_{k=0}^{\infty}p_k =1$ the series in (\ref{epgf}) converges for all complex $|w|\leq 1$.  

\begin{Example}
    Consider the r.v. $N\in\Z$ with p.m.f. and p.g.f.
\begin{equation}\label{jj1}
p_n=\frac{\theta^n}{(n+2)(n+1)}\cdot\frac{{1}}{Z(\theta)},\;\; G(w;\theta)=\frac{Z(\theta w)}{Z(\theta)},\; Z(\theta)=-
\frac{\theta + (1-\theta) \log (1-\theta)}{\theta ^2}\;.   
\end{equation}
This distribution has many notable properties, which will be discussed in detail in the following Subsections. It belongs to the class of {\it Power series distributions} (see Subsection \ref{psdsub}), is defined for $\theta \in (0,1]$, and has finite mean and higher moments only for $\theta<1$. The p.m.f. (\ref{jj1}) is LX and therefore is {\it infinitely divisible}, see Subsection \ref{infdsub}. Moreover, the r.v. $N$ is {\it scalable} in the sense that it can be multiplied by any non-negative real number without destroying its discrete character (see Subsection \ref{scalsub}).

\end{Example}

Consider the sum $S = N_1 + N_2$\footnote{More correct way would be to write  $S \preq N_1\!+\!N_2$, where $\preq$ means the equality in probability.} of two independent r.v. $N_{1,2}$ with p.g.f. $G_{1,2}(w)$ and p.m.f.{}s  $p_{n1,2}$. The p.m.f. and p.g.f. of the r.v. $S$ read   
\begin{eqnarray}
\mathbb{P}(N_1\!+\!N_2\!=\!S) =\sum_{n_1=0}^S \mathbb{P}(N_1\!=\!n_1)\mathbb{P}(N_2\!=\!S\!-\!n_1) 
\\
 G(w) = \mathbb{E}[w^{N_1+N_2}] = \mathbb{E}[w^{N_1}] \mathbb{E}[w^{N_2}] =G_1(w)G_2(w)\;. 
 \label{eq: prodpgf}
\end{eqnarray}
\begin{Proposition} \label{PGSwk}
If $G(w)$ is the p.g.f. of a r.v. $N\geq k>0$, then $w^k G(w)$ is the p.g.f. of the r.v. 
$ M=N+k$, and $w^{-k} Q(w)$ is the p.g.f. of the r.v. $M=N-k$.
\end{Proposition}

\begin{Example}
The Bernoulli distribution has p.m.f. $p_n(a)=(1-a)\delta_{n0}+a\delta_{n1}$ and p.g.f. $\mathcal{B}(w; a)=1-a +aw$. The sum of $m$ {\it independent identically distributed random variables} (\iid{}r.v.)  with common p.m.f. $p_n(a)$ has the binomial distribution with p.g.f.  and p.m.f., respectively:
\begin{equation}\label{Binom}
G(w)=(\mathcal{B} (w;a))^m\;,\; \;\; p_j(m,a)=\binom{m}{j}a^j(1-a)^{m-j}\;,\; m\in \N\;.   
\end{equation}
The p.g.f. $w^k\mathcal{B}(w;a)$ has the p.m.f. $p_n(a)=(1-a)\delta_{n,k}+a\delta_{n,k+1}$ for any integer $k$.
\end{Example}

The analyticity of the p.g.f. $G(w)$ at $w=1$ thallows us to define the $k^{\rm th}$ {\it factorial moment} 
\begin{equation} \label{facmom}
\mu_{[k]} := G^{(k)}(w=1)   =\E(N)_k \;;\; (N)_k:=N(N\!-\!1)\cdots (N\!-\!k\!+\!1)\;.
\end{equation}
The analyticity of the {\it moment generating function} $\mathcal{M}_N(t)$ in the region in $|t|\leq 1$ 
\begin{equation}\label{mgf}
\mathcal{M}_N(t):= \sum_{k=0}^{\infty}\frac{\mu_k}{k!}t^k=\sum_{k=0}^{\infty}\frac{\mathbb{E}N^k}{k!}t^k=\mathbb{E}e^{tN}=G(e^t)  
\end{equation}
implies the finiteness of  the ordinary moments $\mu_k:=\E N^k$ with $k\in\Z $. For a p.g.f. $G(w)$ that is analytical in the closed circle $|w|\leq r$, where  $r\geq e$, this is equivalent to the condition $G(e)<\infty$.

Using the Stirling numbers of the first $s(k,j)$ and second $S(k,j)$ kind, respectively, defined by
\[
(x)_k=\sum_{j=0}^n s(n,j)x^j\;,\;\; x^n=\sum_{k=0}^n S(n,k)(x)_k\;.
\]
we obtain the connection between factorial $\mu_{[k]}$ and ordinary $\mu_k$ moments \cite{JKK2005}:
\begin{equation}\label{frac2norm}
 \mu_{[k]} = \sum_{j=0}^k s(k,j)\mu_j\;,\;\; \mu_k=\sum_{j=0}^k S(k,j)\mu_{[j]}\;.
\end{equation}

If $\E N<\infty$ then $\E N^r <\infty$ $\forall r$ such that $0<r<a$ and $a\in (0,1)$. In the opposite case, the r.v. $N$ is called {\it heavy-tailed}; its $r$-th moments exist for $0 < r < a < 1$ provided the p.g.f. $G(w)$ satisfies \cite{Klebanov2023-xj, KS2208}
\begin{equation}\label{momheavy}
 \int_0^{1}\frac{1-G(w)}{(-\log w)^{r+1}w}dw<\infty \;. 
\end{equation}

One may naively expect that for heavy-tailed distributions, the Boltzmann-Gibbs-Shannon (BGS) entropy $S^{BGS}:=-\sum_n p_n \log p_n$ diverges. However, this need not be the case: $S^{BGS}$ is finite if either $\mathbb{E} \log N  = \sum_{n=1}^{\infty}p_n \log n < \infty $, or if $\exists ~r> 0$ such that $\mathbb{E} N^{r} = \sum_{n=1}^{\infty}p_n n^{r}< \infty$ \cite{Baccetti_2013}.

\subsection{Power series distributions} \label{psdsub}
The following construction gives a simple way to transform an arbitrary r.v. $N\in \Z$ into a new one with all moments finite. Let $Z(w)=\sum_{k=0}^{\infty}p_n w^n$ be the p.g.f. of r.v. $N$, and define the family of {\it power series distributions} (PSD) \cite{JKK2005} with p.m.f. $\{r_n(\theta)\}$ and p.g.f. $\{G(w;\theta)\}$ associated with $Z(w)$ 
\begin{equation} \label{PSDdef}
r_n(\theta):=\frac{p_n \theta^n}{Z(\theta)},\;\;\;
G(w;\theta) =\frac{Z(\theta w)}{Z(\theta)} =\frac{1}{Z(\theta)}\sum_{k=0}^{\infty}p_k\theta^kw^k, \quad \theta \in (0,1] \;.
\end{equation}
The function $G(w;\theta)$ is analytic inside the disk $|w|\leq 1/\theta$. Since a power series is infinitely differentiable inside its disk of convergence, a r.v. $N$ with p.g.f. $G(w;\theta)$ has finite moments of all orders for $\theta <1$.

Suppose $Z(w)$ is a p.g.f. Does there exist a family $\{G(w,\tau), \; \tau \in (0,1]\}$ such that $Z(w) = G(w,\tau_o)$ for some given $\tau_o \in (0,1)$? The answer is simple: such a family exists if and only if (iff) $Z(w)$ admits analytical continuation into a closed disk $|w| \leq 1/\tau_o$. In this case $G(w;\tau) =Z(\tau/\tau_o\cdot w)/Z(\tau/\tau_o)$. Thus, the possibility of analytic continuation into a larger disk is what allows a p.g.f. to belong to a wider family of PSDs as an interior point. 

To illustrate that property, consider two cases:
\begin{enumerate}[label=(\roman*), leftmargin=*]
\item Let $P(w; 1) = \exp\{(w-1)\}$ be p.g.f. of the Poisson distribution with intensity $\lambda =1$. The corresponding power series family 
\begin{equation} \label{poispgf}
G(w;\theta) =\frac{P(\theta w;1)}{P(\theta;1)}=\exp\{\theta (w-1)\}  =P(w;\theta) 
\end{equation}
is the Poisson distribution with intensity $\lambda=\theta$. Since $P(w;1)$ is an entire function (i.e., analytic on the whole complex plane), the value of $\theta>0$ is arbitrary.

\item\label{ExSib}
Let $Z(w)=\mathcal{S}(w;\gamma) := 1-(1-w)^{\gamma}\;,\gamma \in (0,1]$, be Sibuya p.g.f..
Then
\begin{equation}\label{exsbd_pgf}
G(w;\theta,\gamma)= \frac{\mathcal{S}(\theta w;\gamma)}{\mathcal{S}(\theta;\gamma)}=
\frac{1-(1-\theta w)^{\gamma}}{1-(1-\theta)^{\gamma}}
\end{equation}
is the p.g.f. of the {\it Extended Sibuya distribution} \cite{Klebanov2023-xj}.
Here $\theta \in (0,1]$ and this family cannot be embedded in a wider one because $\mathcal{S}(w;\gamma)$ does not allow an analytic continuation to $|w|\leq r$ for $r>1$. 

\end{enumerate}

The following two propositions describe elementary properties of the PSD family.
\begin{Proposition}
Assume that the p.g.f. of r.v. $N\in \Z $ belongs to the PSD family $\mathcal{A}=\{G(w;\theta), \; \theta \in (0,1]\}$.
\begin{enumerate}[leftmargin=1.2cm] \label{propPSD}

    \item If $n\in \N$ then $(G(w;\theta))^n\in \mathcal{A}$.
    \item If $H(w)$ is a p.g.f. of r.v. $M\in \Z $ such that $M+k=N$ with  $\mathbb{P}(N=j)=0,\; j\leq k$ and $\mathbb{P}(N=j)>0,\; j> k,\; k \in \Z $ then $H(w;\theta)=G(w;\theta,k) w^{-k}\in \mathcal{A}$. 
    \end{enumerate}
\end{Proposition}
Property $(II)$ follows from Proposition \ref{PGSwk} and the identity $\theta^k(\theta w)^{-k}=w^{-k}$. 

\medskip

The following proposition connects the class of PSDs with the geometric distribution.
\begin{Proposition} \label{G2PSD}
Let $N \in \Z$  be the r.v. with p.m.f. $p_n$ and p.g.f. $Z(w)$, and $M$ be a r.v., independent of $N$, with p.m.f. $q_n=(1-\theta)\theta^n$ of the geometric distribution. Then the r.v. $K = N \wedge M$, which takes the value $k$ if both $N=k$ and $M=k$, has a PSD and vice versa: Every r.v. $K$ from the PSD family is the product of two r.v. $K = N \wedge M$,  where $M\in \Z$ has the geometric distribution and $N \in \Z$ is some arbitrary random variable.  
\end{Proposition}
\begin{proof}
Let $r_k$ be the p.m.f. of the r.v. $K$ with the p.g.f. $G(w;\theta)$. Then:
\begin{equation} \label{PSD2G}
G(w;\theta)=\sum_{k}r_k w^k=\frac{(1\!-\!\theta) \sum_{k}p_k (\theta w)^k}{(1\!-\!\theta) \sum_k p_k \theta^k} 
=\frac{\sum_{k}p_k (\theta w)^k}{\sum_{k}p_k\theta^k } =
\frac{Z(\theta w)}{Z(\theta)}\;. 
\end{equation}
The proof of the opposite conjecture follows by multiplying the ratio $Z(\theta w)/Z(\theta)$ by $(1-\theta)/(1-\theta)$ and expanding it into powers of its arguments. 
\end{proof}

\begin{Example} \label{sibgeom}
Consider a multiplicative process represented by the r.v. $M$, in which a particle of type {\bf A} undergoes binary fission with probability $\theta$. Each offspring particle independently either remains in a stable (non-splitting) state with probability $1-\theta$ or undergoes another fission event with probability $\theta$. The probability of ultimately observing $j$ particles of type {\bf A} is given by
\[
\mathbb{P}(M = j) = (1-\theta)\,\theta^{j-1}, \quad j = 1, 2, \dots,
\]
which corresponds to a geometric distribution starting at $j=1$.
In the concurrent process characterized by the r.v. $N$, the probability that a particle of type {\bf B} undergoes binary fission at step $i$ is given by $1-\gamma/i$. The p.m.f. for observing $j$ particles of type {\bf B} is described by the Sibuya distribution:
\begin{equation}\label{Sbdpmf}
\mathbb{P}(N=j)=p_j=\frac{\gamma}{j}\prod_{i=1}^{j-1}\left(1-\frac{\gamma}{i}\right), \qquad j=1,2,\ldots
\end{equation}
The random variable $K$ characterizes the event in which both processes generate an identical number of particles, that is, $K = N \wedge M$. In this case, $K$ follows an Extended Sibuya distribution with probability generating function given by (\ref{exsbd_pgf}). Notably, even though the first and higher-order moments of $M$ do not exist, all moments of $K$ are finite.
\end{Example}

\begin{Theorem}\label{PSDcum}{\rm \cite{Khatri1959}} 
Let $N\in\Z$ be a random variable with probability generating function $G(w;\theta)=Z(\theta w)/Z(\theta)$, where $\log Z(\theta)$ admits a local power-series representation in $\theta$ with strictly positive coefficients. Then the cumulants $\kappa_{\ell}$ of $N$ satisfy the recurrence relation
\begin{equation} \label{reccum1}
\kappa_{\ell}=\theta \frac{d\kappa_{\ell-1}}{d\theta}\;,\;\; \ell=1,2,\ldots\;.
\end{equation}
\end{Theorem}
\begin{proof}
The $\ell^{\rm th}$ cumulant $\kappa_{\ell}$ of $N$ is the coefficient of $t^{\ell}/\ell!$ in the power-series expansion of the {\it cumulant generating function} $\mathcal{K}_{N}(t)$ in the variable $t=\log w$:
\begin{equation}\label{cumgf}
\mathcal{K}_{N}(t):= \log \E e^{tN} =\log G(e^t;\theta) 
= \sum_{\ell \geq 1}\kappa_{\ell}\frac{ t^{\ell}}{\ell!};\;\;\; \kappa_{\ell}=\mathcal{K}_N^{(\ell)}(t=0) \;.
\end{equation}  
For the PSD, $\kappa_{\ell}$ is obtained from the $\ell^{\rm th}$ derivative of the function $\log Z(\theta e^t)$ evaluated at $t=0$. In particular,
\begin{equation}\label{reccum_m}
\kappa_{1}  =\frac{dG(e^t;\theta)}{dt}\Big|_{t=0}= \frac{d\log Z(\theta e^t)}{dt}\Big|_{t=0} = \theta \frac{d\log Z(\theta)}{d\theta}\;.
\end{equation}
The analyticity of $\log Z(\theta)$ implies that the cumulant $\kappa_{\ell}$ can be obtained by differentiating $\ell$ times the function $\log Z(\theta)$ with respect to $\theta$ or, equivalently, with respect to $\tau=\log \theta$:
\begin{equation}\label{reccum}
    \kappa_{\ell} =\theta \frac{d \kappa_{\ell-1}}{d\theta} = \frac{d\kappa_{\ell-1}}{d\tau}=\frac{d^{\ell}\log Z(\theta)}{d\tau^{\ell}} =
    \theta^{\ell}\frac{d^{\ell}\log Z(\theta)}{d\theta^{\ell}}\;.
    \vspace{-.5cm}
\end{equation}
\end{proof}
From (\ref{reccum}) it follows that, if $\mu=\E N = \kappa_1$ is known as a function of $\theta$, then all higher-order cumulants are determined by this single function.

\begin{Corollary}\label{addcum}
Let $N_{1,2}\in \Z $ be two independent r.v. with the cumulants $\kappa_{\ell,1}$ and $\kappa_{\ell,2}$. Then the r.v. $M=N_1+N_2$ has the cumulants $\kappa_{\ell}=\kappa_{\ell,1}+\kappa_{\ell,2}$.
\end{Corollary}
\begin{proof}
The p.g.f. $G(w)$ of a r.v. $M$  equals the product $G(w)=G_1(w)G_2(w)$ of the p.g.f. of r.v. $N_1$ and $N_2$. Using  (\ref{cumgf}) we  have $\mathcal{K}_{M}(t) = \mathcal{K}_{N_1}(t) +\mathcal{K}_{N_2}(t)$, and comparing coefficients of $t^{\ell}$ gives $\kappa_{\ell}=\kappa_{r,1}+\kappa_{r,2}$.
\end{proof}

\begin{Example}~~~~
\begin{enumerate}[label=(\roman*), leftmargin=*]
\item The Poisson p.g.f. (\ref{poispgf}) with 
$Z(\theta)=e^{\theta}$ has $\kappa_{\ell}=\theta,\;\forall \ell \in \N$.  
\item The sum of independent r.v. $M=N_1+N_2$, where $N_1$ is concentrated on $n=0,1,2,\ldots$ and $N_2$ on $n=0,2,4,\ldots$, each Poisson distributed, has {\it Hermite distribution} \cite{JKK2005} with the p.g.f. $Q(w)=e^{\theta_1(w-1)}e^{\theta_2(w^2-1)}$. Its cumulant generating function and cumulants  are: 
\begin{equation} \label{Hercum}
    \mathcal{K}_N(t)= \theta_1(e^{t} - 1) + \theta_2(e^{2t} - 1)\;,\;\; \kappa_{\ell} = \kappa_{\ell}(\theta_1)+\kappa_{\ell}(\theta_2) =\theta_1+ 2^{\ell}\theta_2\;.
\end{equation}
The knowledge of $\log Z =\theta_1  + \theta_2$ alone is insufficient to determine $\kappa_1=\theta_1+2\theta_2$. Note that for the Poisson distribution concentrated on $n=0,j,2j,3j, j\in \N$ with the p.g.f. $\mathcal{P}(\theta w^j)/\mathcal{P}(\theta)$ we obtain $\kappa_{\ell}=j^{\ell} \theta$. 
\end{enumerate}
\end{Example}
\begin{Example}
Eq. \ref{reccum} yields the cumulants of the {\it Negative binomial distribution} (NBD) with parameters $k>0$\footnote{The geometric distribution on $\Z$ represents the case with $k=1$.} and $0<\theta<1$: 
\begin{eqnarray}\label{NBD}
 \frac{Z(\theta w)}{Z(\theta)}= \left(\frac{1-\theta }{1-\theta w}\right)^k\;,\;\;
 \log Z(\theta)  =-k \log (1\!-\!\theta)
 =k\log \left(1\!+\!\frac{\kappa_1}{k}\right)
\end{eqnarray}
 
\begin{equation} 
\label{NBD_cum}
\kappa_{1}= k\frac{\theta }{1-\theta }\;,\;\; \kappa_2=k \left(\frac{\theta ^2}{(1-\theta )^2}+\frac{\theta}{1-\theta }\right)\;, \ldots\;, \;\;\kappa_{\ell}= k\cdot f(\theta,r)\;,
\end{equation}
where $f(\theta,\ell)$ is a polynomial of degree $\ell^{\rm th}$  in variable $\kappa_1/k=\theta/(1-\theta)$.  
Therefore, at fixed $\theta$, the ratio of any two cumulants $\kappa_i/\kappa_j$ is independent of $k$. 
\end{Example}

\begin{Corollary}\label{Khatri} {\rm \cite{Khatri1959}} 
Every PSD is completely determined by the functional dependence of its first two moments (cumulants) on some parameter $\omega$. 
\end{Corollary}
\begin{proof}
Let $\kappa_{1}\!=\!y_{1}(\omega)$ and $\kappa_{2}\!=\!y_{2}(\omega)$ be such functions. Then from (\ref{reccum1}) it follows that
\[
\kappa_2=y_2(\omega)=\frac{d\kappa_1}{d\log \theta}= 
\frac{d y_1(\omega)}{d\omega}\frac{d\omega}{d\log\theta}\;,\;\;\; 
\frac{d\log \theta}{d \omega}=\frac{1}{y_2(\omega)}\frac{dy_1(\omega)}{d\omega}
\]
and from (\ref{reccum_m}) that
\[
\kappa_1=y_1(\omega)=\frac{d\log Z(\theta)}{d \log \theta}\;,\;\;\;
\frac{d \log Z(\omega)}{d\omega}=\frac{d\log Z(\theta)}{d\log \theta}\frac{d\log \theta}{d \omega}=
\frac{y_1(\omega)}{y_2(\omega)}\frac{d y_1(\omega)}{d\omega}\;.
\]
Consequently, the dependence of \(Z(\omega)\) on \(\omega\) is entirely characterized by the two functions \(y_{1}(\omega)\) and \(y_{2}(\omega)\). Integrating \(\log Z(\omega)\) with respect to \(\omega\) and subsequently exponentiating the resulting expression yields a function of the form \(c_{1} Z(\omega)\), where \(c_{1}\) is a constant of integration. When we construct the probability generating function \(Z(\theta w)/Z(\theta)\), this multiplicative constant cancels, thereby yielding the asserted result and completing the proof.\end{proof}
\begin{Example}
    For the NBD p.g.f. with the cumulants (\ref{NBD_cum}) and $\omega=\theta$ we obtain:
    \[
    \frac{d \log Z(\theta)}{d\theta}=\frac{\kappa_1}{\kappa_2}\frac{d \kappa_1}{d\theta}=
    \frac{\kappa_1}{\theta}=\frac{k}{1-\theta};\;\;\;  Z(\theta)=\frac{c_1}{(1-\theta)^k};\;\;\;\frac{Z(\theta w)}{Z(\theta)}=\left(\frac{1-\theta}{1-\theta w}\right)^k\;.
    \]
\end{Example}
\subsection{Sum of a Random Number of Random Variables: Compound Distributions}
Consider the random sum \(S_K = N_1 + N_2 + \ldots + N_K\), where \(\{N_i\}_{i=1}^K\) are independent and identically distributed integer-valued random variables, \(N_i \in \mathbb{Z}\), each with probability mass function \(\mathbb{P}(N=j)=q_j\) and corresponding probability generating function (p.g.f.) 
\[
Q(w) = \sum_{j} q_j w^j.
\]
If \(K\) is fixed (nonrandom), the p.g.f. of the random variable \(S_K\) is, by (\ref{eq: prodpgf}), given by
\[
H(w) = [Q(w)]^K.
\]

Now assume that the number of summands \(K\) is itself a random variable with probability mass function \(r_k = \mathbb{P}(K = k)\) and p.g.f.
\[
R(w) = \sum_{k} r_k w^k.
\]
Then the p.g.f. of \(S_K\) is
\begin{equation}
\label{eq: comp}
H(w) = \sum_{k} r_k [Q(w)]^k = R(Q(w)) := R \circ Q \;.
\end{equation}
The operator \(\circ\) in (\ref{eq: comp}) denotes the composition (compounding) of two probability generating functions. In this setting, the random variable \(S_K\) is said to have a {\it compound distribution}. The generating function \(R(w)\) is referred to as the p.g.f. of the {\it compounding} (or {\it primary}) distribution, while \(Q(w)\) is the p.g.f. of the {\it compounded} (or {\it secondary}) distribution.

Let us observe that, under the assumptions $\E N<\infty$ and $\E K<\infty$, Eq.~\ref{eq: comp} yields the simple identity $\E S_K = (\E K)\cdot(\E N)$. This relation shows that, in expectation, the compounding of two independent random variables $K$ and $N$ corresponds to the product of their respective mean values.

An important case of Eq. \ref{eq: comp} arises when $H(w)=Q_{a_1}\circ Q_{a_2}=Q_{a_2}\circ Q_{a_1}$, where $a_{1,2}$ are parameters, and $Q_{a_i}\circ Q_{a_j}:=Q(Q(w;a_i);a_j)$ \footnote{In the following we write, e.g., $P_{\lambda}\circ\mathcal{S}_{\gamma}$ as shorthand of $P(w;\lambda)\circ \mathcal{S}(w;\gamma)$.}.  Examples of p.g.f.{}s forming an additive semi-group with $Q_{a_1}\circ Q_{a_2}= Q_{a_1 +a_2}$ are  listed in Table \ref{tab:additive}.
\begin{table}[hbt]
    \centering
    \begin{tabular}{|c|c|c|c|c|}
    \hline
       {\bf Distribution} & {\bf p.g.f.} & {\bf additive parameter} & {\bf ref.} \\[1.ex]
        \hline  \hline  
        Bernoulli & $\mathcal{B}(w; a)=1-a +a w$ & $a$ &  - \\
        \hline \hline
      Geometric & $\mathcal{G}(w;\theta)=(1-\theta)/(1-\theta w)$ & $\log(1-\theta)$ & \cite{JKK2005}  \\
        \hline \hline
      Sibuya & $\mathcal{S}(w;\gamma)=1-(1- w)^{\gamma}$ & $\log \gamma$ & \cite{art:Sibuya}  \\
         \hline \hline
         scaled Sibuya  & $\mathcal{S}(w;\gamma,\lambda)=1-\lambda(1- w)^{\gamma}$ & $\log \gamma$ & \cite{Christoph2000,Klebanov2023-xj} \\
         \hline
    \end{tabular}
    \vspace{.2cm}\caption{p.g.f.{}s forming an additive semi-group under compounding operation.}
    \label{tab:additive}
\end{table}

  Consider a r.v. $M \in \Z $ whose p.g.f. can be $\forall k \in\N$ written as  $G(w) = Q_1\circ Q_2\circ \ldots \circ Q_k$, where $Q_i(w)=Q(w)$ are the p.g.f. of \iid{}r.v. $N_i \in \Z, i=1,\ldots k$. In this case the non-integer moment $\E M^r,~ r=1/k$   represents the geometric mean
\begin{equation}\label{momfrac}
\E M^{1/k}=\E (\sqrt[k]{N_1\cdot N_2\cdot \ldots \cdot N_k})\,. 
\end{equation}
Eq. \ref{momfrac} holds regardless of whether $\E N$ is finite. 
\subsection{Infinitely divisible distributions} \label{infdsub}
The r.v. $N\in \Z$ is $n$ {\it divisible} if it can be expanded into a sum of $n$ \iid{}r.v. $N_1+N_2+\ldots N_n,\; N_i \in \Z$, each with the same p.g.f. $G(w)$, so that the  p.g.f. $H(w)$ of $N$ satisfies $H(w)=[G(w)]^n$.

By generalizing this concept, we say that an r.v. with the p.g.f. $H(w)$ is {\it infinitely divisible}  if $\forall n\in \N$  its $n^{\rm th}$ root $(H(w))^{1/n}$ is also the p.g.f.. This is equivalent to saying that $H(w)$ is the p.g.f. of a {\it compound Poisson distribution} \cite{Feller1} with a secondary p.g.f. $G(w)$: 
\begin{equation}
\label{eq: cpd}
H(w;\lambda)=P(G(w);\lambda)\,,~~ P(w;\lambda)=e^{\lambda(w-1)}\,,~
G(w)=\sum_{n>0} g_n w^n=1+\frac{\log H(w)}{\lambda}\,,
\end{equation}
 \begin{Remark}
 Let us note that $G(0)=0$ implies  $\lambda=-\log H(0)$. 
 For secondary p.g.f. with $G(0)=1-a>0,\; a\in(0,1)$,  there always exists a p.g.f. $G_0(w)$ such that 
\begin{equation}\label{zero-ext-compP}
P(G(w);\lambda)=P(G_0(w);a\lambda))\;.    
\end{equation}
\end{Remark}

\begin{Proposition} {\rm \cite{StvH2004}} \label{infd_nec}
Necessary conditions for the infinite divisibility of a r.v.  $N\in \Z $  with p.m.f. $p_n$ and p.g.f. $H(w)$ are
\begin{enumerate} 
\item $p_0 > 0$.
\item  $\forall n \in \N,\; p_n < 1/e $. 
\item $\log H(w)$ is analytic for $|w|\leq 1$ and has non-negative power series expansion coefficients at $w=0$. 
\end{enumerate}
\end{Proposition}
While conditions $(I)$ and  $(III)$ are immediate,  condition $(II)$ requires more insight\cite{StvH2004}. For the Poisson distribution, the following inequality is true 
\[ \forall \lambda>0 \;\;\; \mathbb{P}(N=n) = \frac{\lambda^n}{n!}e^{-\lambda}\leq \frac{n^n}{n!}e^{-n}=:\pi_n\;.
\]
Since $(1+1/n)^n \uparrow e$ as $n\to\infty$, the sequence $\pi_n$ is non-increasing, so $\mathbb{P}(N=n)\leq \pi_1=1/e$. Therefore, the p.m.f. $p_k,\; k\in \N$, of the random sum $S_N=\sum_{i=1}^N M_i$ of $N$ \iid{}r.v.  $M_i\in \Z$, where $N$ is Poisson distributed satisfies
\[
p_k=\sum_{n=1}^k \mathbb{P}(N=n)\mathbb{P}(S_n=k)\leq \frac{1}{e}\sum_{n=1}^k\mathbb{P}(S_n=k)\leq \frac{1}{e}\;.
\]
Let us add that if $R(w)$ is an infinitely divisible p.g.f. and $G(w)$ is an arbitrary p.g.f. on $\Z $, then $H(w)=R(G(w))$ is also an infinitely divisible p.g.f.

Some examples of infinitely divisible p.g.f. are given in Table \ref{tab:comppois}.

\begin{table}[h]
    \centering
    \begin{tabular}{|c|c|c|c|c|}
    \hline
         {\bf Name} & {\bf Secondary p.g.f.} & {\bf Compound p.g.f.} & {\bf ref.} \\
        \hline  \hline  
       Logarithmic & $\mathcal{L}(w;\theta)=\log(1-\theta w)/\log(1-\theta)$ & Negative Binomial & \cite{JKK2005}  \\
        \hline \hline
       Sibuya & $\mathcal{S}(w;\gamma)=1-(1- w)^{\gamma}$ & Discrete stable & \cite{StvH2004}  \\
         \hline \hline
       Shifted Bernoulli  & $w\mathcal{B}(w;a)=(1-a) w+ a w^2$ & Hermite distribution & \cite{JKK2005} \\
         \hline
    \end{tabular}
   \vspace{.2cm} \caption{Secondary p.g.f. of some compound Poisson distribution.}
    \label{tab:comppois}
\end{table}

\begin{Proposition} {\rm \cite{Feller1}}
\label{Pmultiplets}
Let $H(w;\lambda)=\exp\{\lambda [G(w)\!-\!1]\}$ be the p.g.f. of  infinitely divisible r.v. $M\in \mathbb{Z}_+$, and $G(w)=\sum_{j>0}^{J}g_j w^j$ its secondary p.g.f. Then $M$ can be written as the sum $M \!= \!N_1\!+ \! N_2 \!+\! \ldots \!+\! N_J$ of 
Poisson-distributed singlets, doublets, $\ldots, J$-tuplets. 
\end{Proposition}
\begin{proof}
If $P(w;\lambda_j)$ is the p.g.f. of Poisson distribution then $P(w^j;\lambda_j)$ is the p.g.f. of Poisson distribution concentrated at $j,2j,3j,\ldots$. Their product is:
\begin{equation}  \label{pois2comb}
 H(w;\lambda)=\exp\left[\lambda \left(\sum_{j>0}^{J}g_j w^j\!-\!1\right)\right] \!=\!\prod_{j>0}^{J}e^{\lambda g_j (w^j \!-\!1)}\! =\! 
 \prod_{j>0}^{J}P(w^j,\lambda_j) \;. 
\vspace{-.6cm}
\end{equation}
\end{proof}

Note that the coefficients $\lambda_j=\lambda g_j$ in (\ref{pois2comb}) represent the mean number of clusters $\langle n_j \rangle$, each consisting exactly of  $j$ particles. These can be identified with the combinants to be introduced in Section \ref{comsec}.  

The following proposition establishes the connection between the cumulants of an infinitely divisible distribution and the moments of its secondary distribution.

\begin{Proposition}\label{CPDcum}
Let $\langle j^r \rangle=\E(J^r)$ be the moments of r.v. $J\in \Z $ with the p.m.f. $g_j$ and p.g.f. $G(w)$. Then the r.v. $N$ with the p.g.f. $H(w;\lambda)=\exp\{\lambda[G(w)-1]\}$ has the cumulants $\kappa_r=\lambda \langle j^r \rangle$. Moreover, if $\langle j^r \rangle<\infty$ then $0\leq \kappa_r<\infty$.
\end{Proposition}
 \begin{proof}
Expanding the cumulant generating function (\ref{cumgf}) gives:
\begin{equation}
 \kappa_r=\frac{d^r \log H(w;\lambda)}{dt^r}\Big |_{t=0} =\frac{d^r (\lambda \sum_{j>0}g_j e^{jt})}{dt^r}\Big |_{t=0}=  \lambda \sum_{j>0}j^r g_j e^{jt}\Big |_{t=0}
 = \lambda \langle j^r \rangle \,.
\vspace{-.8cm}
\label{infdcum}
\end{equation}        
\end{proof}

Expressing the moments of $J$ as $\langle j^r \rangle = \sum_{j=0}^r S(r,j)\mu_{[j]}(1)$ \cite{JKK2005}, where $S(r, i)$ are the Stirling numbers of the second kind  (see Eq. \ref{frac2norm} and $\mu_{[i]} = G^{(i)}(1)$ are the factorial moments of the secondary distribution, we obtain
\begin{equation}\label{cum2fac}
\kappa_r= \lambda \langle j^r \rangle = \lambda \sum_{i=0}^r S(r,i)\mu_{[i]}\;.    
\end{equation}
By exploiting the connection between the moments $\mu_r=\E N^r$ and the cumulants $\kappa_r$ of infinitely divisible distribution (Proposition \ref{CPDcum}), we obtain 

\begin{equation}\label{mom2cum}
 \mu_{r+1}=\sum_{i=0}^r \binom{r}{i}\mu_i\kappa_{r-i+1} =\lambda\sum_{j=i}^r \binom{r}{i}\mu_i \langle j^{r-i+1} \rangle\;.
\end{equation}   

Expressing the p.g.f. as a compound Poisson distribution can sometimes be a difficult task. Another way to prove infinite divisibility is  to rewrite (\ref{eq: cpd}) as 
\begin{equation}
\label{eq: cpdr}
H(w)=\exp\left[-\sum_{n=0}^{\infty}\frac{r_n}{n+1}(1-w^{n+1}) \right]\;,\;\; r_n:=\lambda(n+1)g_{n+1}\;;\; k\in\Z;,
\end{equation}
and construct the generating function of the {\it canonical sequence}   $\{r_n\}$ defined as:  
\begin{equation}\label{Rw}
 R(w)\!=\! \sum_{n=0}^{\infty}r_n w^n \!=\!\frac{d}{dw}\log H(w)\!=\!\frac{H^{(1)}(w)}{H(w)}\;.
\end{equation}
Solving the differential equation in (\ref{Rw})  and using the fact that a real-valued function $R(w) \in [ 0,1)$ is completely monotone iff there exists a non-negative sequence $\{r_k\}, k \in \Z$, for which $R(w)$ is the generating function, one can formulate the following condition for infinite divisibility.

\begin{Theorem}{\rm \cite{StvH2004}}\label{HexpR}
The p.g.f. $H(w)\!=\! \exp \Big [\! - \! \int_w^1 R(x) dx  \Big ]$ is infinitely divisible iff the generating function $R(x)$ (\ref{Rw}) is absolutely monotone. 
\end{Theorem}
\begin{proof}
\[
\log H(w)=-\int_w^1 R(x) dx =\sum_{n=0}^{\infty}r_n \int_w^1 x^n= -\sum_{n=0}^{\infty}\frac{r_n}{n+1}(1-w^{n+1})
\]
\end{proof}
Accordingly, the logarithm of the infinitely divisible p.g.f. $\log H(w)$  must be an analytical function with positive power expansion coefficients  $\forall n \in \N$.

Another approach is to prove the infinite divisibility by equating the coefficients on both sides in the power series expansions of $H^{(1)}(w)=R(w)H(w)$. 

 \begin{Theorem}{\rm \cite{StvH2004}} \label{infd_suf}
The necessary and sufficient condition for the infinite divisibility of the r.v.  $N\in \Z $  with p.m.f. $p_n$  and p.g.f. $H(w)$ is the existence of a positive sequence $\{r_n\}$ such that
\begin{equation}\label{r_k}
\sum_{n=0}^{\infty}\frac{r_n}{n\!+\!1}\!=\!-\log(p_0) < \infty,\;\;  (n\!+\!1)p_{n+1}\! =\!\sum_{k=0}^{n}p_k r_{n-k} 
\end{equation}
\end{Theorem}

The following proposition, presented here without the proof, gives an easily verifiable sufficient condition for infinite divisibility. 

\begin{Proposition} {\rm \cite{StvH2004}}\label{LXpmf}
A p.m.f. $p_n, n\in\mathbb{Z}_+$, satisfying the log-convexity (LX) condition of (\ref{eq:cvx}), i.e., with the ratios
$p_n/p_{n-1}$ non-decreasing with $n$,  is infinitely divisible.    
\end{Proposition}

\begin{Example}\label{ULXCex}
~~~~~
\begin{enumerate}[label=(\roman*), leftmargin=*]
 \item The Shifted Logarithmic distribution with p.m.f. and p.g.f. 
    \begin{equation}\label{shlog}
 p_n =-\frac{1}{\log(1-\theta)}\frac{\theta^{n+1}}{n+1}\;,\;\; \; H(w)=\frac{\log (1-\theta  w)}{w \log (1-\theta )}
\end{equation}
is ULX and hence also LX, and is therefore infinitely divisible. 
\item The Poisson distribution satisfies both ULX and ULC conditions. 
\end{enumerate} 
\end{Example}
Additional examples of LX condition can be found in Appendix \ref{ExLX}.

\subsection{Comparisons for different families of p.g.f.} 
We begin with a simple but useful observation.
\begin{Proposition}{\rm \cite{StvH2004}}\label{infdthin}
    Suppose that $H(w;\lambda)$ is an infinitely divisible p.g.f. Then for $\theta \in (0,1]$ each of the families 
\[ \Bigl\{\frac{H(\theta w;\lambda)}{H(\theta;\lambda)}\Bigr\}, \quad \Bigl\{\frac{H(\theta;\lambda)H(w;\lambda)}{H(\theta w;\lambda)}\Bigr\}\quad\text{and} \quad \Bigl\{H(1-\theta+ \theta w;\lambda)\Bigr\} \]
consists of an infinitely divisible p.g.f.{}s.
\end{Proposition}
In particular, for every infinitely divisible p.g.f. $H(w;\lambda)=\exp\left[\lambda(G(w)-1)) \right]$, there always exists a PSD p.g.f. $R(w;\theta,\lambda)=H(\theta w;\lambda)/H(\theta;\lambda)$ representing a compound Poisson distribution with parameter $\nu$ depending on $\theta$ and $\lambda$:
\begin{equation}\label{PSDfromCPD}
 R(w;\theta,\lambda)
 =\frac{e^{\lambda (G(\theta  w)-1)}}{e^{\lambda (G(\theta )-1)}}=e^{\nu ( Q(w)-1)};\;\;
 \nu=\lambda G(\theta),\; Q(w)=\frac{G(\theta w)}{G(\theta)}\;.
\end{equation}

The following theorem addresses the case in which a p.g.f. belongs simultaneously to both
families.
\begin{Theorem} \label{infd2PSD}
Suppose the infinitely divisible p.g.f. $H(w;\lambda)=\exp\{\lambda(G(w)-1)\}$ is at the same time a PSD, $H(w;\lambda) = Z(\theta w)/Z(\theta)$, with $\theta$ depending only on $\lambda$. Then $H(w;\lambda)$ has exactly one nonzero combinant $\lambda_j=\lambda a$, so that
\begin{equation}\label{eq: PSDINFD}
H(w;\lambda) = \exp\{\lambda a(w^j -1)\} 
\end{equation}
for some $j\in \N$ and $a>0$.
\end{Theorem}
\begin{proof}
Because $H(w;\lambda)$ is a p.g.f., $G(w)$ is analytic in $|w|\leq 1$, and the same holds for the function $h(w):=\log Z(w)$. Therefore 
$\lambda (G(w)-1) = h(\theta w) - h(\theta)$ 
for all $|w|\leq 1$ and all $\lambda >0$. Let $j$ be the smallest integer $m >0$ for which $G^{(m)}(0) \neq 0$. Differentiating $j$ times with respect to $w$ we gives
$ \lambda G^{(j)}(w) = \theta^{j}h^{(j)}(\theta w)$.
Setting $w=0$ and using $G^{(j)}(0) \neq 0$ yields
\[ \lambda = \frac{h^{(j)}(0)}{G^{(j)}(0)}\theta^{j};\;\; G^{(j)}(w) =\frac{G^{(j)}(0)}{h^{(j)}(0)} h^{(j)}(\theta w).\]
Because $\lambda$ is arbitrary positive then $G^{(j)}(w)=const$. However, $G^{(n)}(0)=0$ for $0<n<j$. Therefore, $ \lambda G(w)= \lambda(a w^j +b)$. The condition $G(1)=1$ leads to $b=1-a$ and hence to $\lambda (G(w)-1)= \lambda a (w^j-1)$.
\end{proof}

The following theorem addresses the converse situation.
\begin{Theorem} \label{PSDisCPD}
Suppose that the p.g.f. belongs to the PSD family, $H(w;\lambda) = Z(\theta w)/Z(\theta)$. If $Z(0)=a_0>0$ and the function $\log Z(\theta w)$ is analytic in $|\theta w|\leq 1$ with non-negative power series coefficients at $w=0$, then $H(w;\lambda)$ is infinitely divisible with parameter $\lambda=-\log H(0)$ dependending on $\theta$.
\end{Theorem}
\begin{proof}
 Using Eqs.\ref{epgf} and \ref{eq: cpd} we must show that 
\begin{equation} \label{CPD_G}
 G(w) =1+\frac{1}{\lambda}\log \left[ \frac{Z(\theta w)}{Z(\theta )}\right] \;
\end{equation}
is a secondary p.g.f. of the compound Poisson distribution $P(G(w);\lambda)$. Since $G(1)=1$, we only need to verify that $\forall i>0,~ G^{(i)}(0)\geq 0$ and $G(0)=0$. The first condition follows from the assumed analyticity and non-negativity of coefficients of $\log Z(\theta w)$. The second condition, $G(0)=0$, implies  $\lambda=-\log (Z(0)/Z(\theta))$.
\end{proof}

\begin{Example}
The negative binomial distribution is a PSD. The function $\log Z(\theta w)=-k\log(1- \theta w)=k\sum_{n=1}(\theta w)^n/n$ has positive expansion coefficients $\theta^n/n$, and is therefore infinitely divisible.
\end{Example}

\begin{Example}
The r.v $N$ with the p.g.f. (\ref{exsbd_pgf}) of Extended Sibuya distribution is a PSD with $Z(0)=0$ and therefore is not infinitely divisible.  On the other hand, the r.v. $M=N-1$  with the p.g.f. of the {\it Shifted Extended Sibuya distribution}:
\begin{equation} \label{shexsbd}
H_1(w;\theta)=\frac{Z_1(\theta w)}{Z_1(\theta)}\;,\;\; 
 Z_1(\theta)=\frac{1-(1-\theta)^{\gamma}}{\theta}
\end{equation}
with polynomial expansion 
\begin{equation}
H_1(w;\theta)=
\frac{1}{Z_1(\theta)}\sum_{j=0} \binom{\gamma }{j+1}(-\theta )^j w^j\;.
\end{equation}
This function is analytic on $|\theta w| \leq 1$.  With $\lim_{w\to 0^+}H_1(w;\theta)=\theta\gamma/Z(\theta)>0$ and $H_1(1;\theta)=1$, the function $H_1(w;\theta)$ is a valid p.g.f. on $\Z$. Its infinite divisibility follows from the positivity of all coefficients of the power series of $\log Z_1(w)$ in the variable $\theta w$. An alternative proof is given in \cite{Klebanov2023-xj}. 
\end{Example}
\subsection{Combinants}\label{comsec}
Let $p_n$ be the p.m.f. and $H(w)$ the p.g.f. of the r.v. $N\in\Z$, and assume that $\log H(w)$ is an analytic in  $w$ \footnote{Classical example of analytic p.g.f. with nonanalytic logarithm is $H(w)=w^m, m\in \N$, i.e., with the p.m.f. $p_n=\delta_{nm}$.}. In this case, $\log H(w)$ can be expanded directly in powers of $w$: 
\begin{equation}\label{combdef}
 \mathcal{C}_N(w):= \log \E w^N= \log H(w)=\sum_{j=0} \lambda_j w^j=\sum_{j>0}\lambda_j( w^j -1);.
\end{equation}
The coefficients $\lambda_j$ in the power series expansion (\ref{combdef}) are called {\it combinants} \cite{Kauffmann:1978vw}. The last equality in (\ref{combdef}) follows from the analyticity of $\log H(w)$ at $w=0$ together with the condition $\log H(1)=0$, which implies that $\lambda_0=\log p_0=-\sum_{j>0}\lambda_j$. Therefore,

\[
H(w)=\exp[\sum_{j>0} \lambda_j(w^j-1)]=\prod_{j>0}e^{\lambda_j (w^j - 1)}
\]
When $\lambda_j$ are positive, $H(w)$ is infinitely divisible, and the combinants represent the p.m.f. of its secondary distribution, $\lambda_j=g_j$.

As with cumulants, combinants are also additive.
\begin{Proposition}\label{addcomb}
Let $N_{1,2}\in \Z $ be two independent r.v. with the combinants $\lambda_{j,1}$ and $\lambda_{j,2}$ and p.g.f. $H_1(w)$ and $H_2(w)$. Then  $M=N_1+N_2$ has the combinants $\lambda_{j}=\lambda_{j,1}+\lambda_{j,2}$.
\end{Proposition}
\begin{proof}
The p.g.f. $H(w)$ of $M$ satisfies $H(w)=H_1(w)H_2(w)$. Using (\ref{combdef})  and comparing the coefficients of $w^{\ell}$ gives $\lambda_{j}=\lambda_{j,1}+\lambda_{j,2}$.
\end{proof}

 We remark in passing that for distributions of the PSD family, there is a close connection between combinants and cumulants. Expanding the logarithm of the PSD p.g.f.  
\[
\log Z(\theta w) -\log Z(\theta)
=\sum_{j\geq 0}\lambda_j\theta^jw^j -\sum_{j\geq0}\lambda_j\theta^j
=\sum_{j> 0}\lambda_j\theta^j(w^j-1) 
\]
 and using Eq. \ref{reccum} the coefficients $\lambda_j$ are expressed as
\begin{equation}\label{PSDcomb}   
\lambda_j 
=\frac{1}{j!}\frac{d^{j} \log Z(\theta )}{d\theta^j}=
\frac{j!}{\theta^{j}}\kappa_{j}\;.
\end{equation}

\begin{Example}\mbox{} \smallskip
\begin{enumerate}[label=(\roman*), leftmargin=*]
\item Consider the p.g.f. $H(w)=p_0+p_1 w+p_2 w^2+p_3 w^3+p_4 w^4$. The coefficients in the expansion of $\log H(w)=\sum \lambda_j w^j$  are:
    \[
    \lambda_0=\log p_0,\;
    \;\lambda_1=\frac{p_1}{p_0},\;
    \lambda_2=\frac{ p_2}{p_0}-\frac{1}{2}\left(\frac{p_1}{p_0}\right)^2,\;
    \lambda_3=\frac{ p_3}{p_0}-\left(\frac{p_1}{p_0}\right) \left(\frac{p_2}{p_0}\right) +\frac{1}{3}\left (\frac{p_1}{p_0}\right)^3,\;  
    \]
    \[
     \lambda_4=\frac{ p_4}{p_0}-\left(\frac{p_1}{p_0}\right) \left(\frac{p_3}{p_0}\right)+ \left(\frac{p_1}{p_0}\right)^2 \left(\frac{p_2}{p_0} \right)-               \frac{1}{2}\left(\frac{p_2}{p_0}\right)^2   -\frac{1}{4}\left (\frac{p_1}{p_0}\right)^4
     \;.
    \]
    Note that combinants with $j>1$ need not always be positive.
 \item The compound Poisson p.g.f. $P\bigl(Q(w;\theta,\gamma),\nu\bigr)$ with secondary  p.g.f. 
\begin{equation}\label{Pois-ExSbtrNBD}
Q(w;\theta,\gamma)=\frac{G(\theta w;\gamma)}{G(\theta,\gamma)}=\frac{1-(1-\theta w)^{\gamma}}{1-(1-\theta)^{\gamma}}\;,\qquad 0<\theta<1,\;\gamma<1.
\end{equation}
belongs to the PSD family, see Eq.\ref{PSDfromCPD}. In terms of the secondary distribution p.m.f. $q_j$ its  combinants $\lambda_j$ read:
\begin{equation}
\lambda_j = \nu\, q_j 
= \nu\,\frac{(-\theta )^j}{(1-\theta )^{\gamma }-1}\binom{\gamma }{j} 
\xrightarrow[\gamma\to 0]{} -\frac{\theta^j}{j\,\log (1-\theta )}\,.
\end{equation}
For $0<\gamma<1$, the secondary p.g.f. $Q(w;\theta,\gamma)$ is associated with the {\it Extended Sibuya} distribution. In particular, when $Q(w;1,\gamma)=\mathcal{S}(w;\gamma)$, the function $P(Q(w;1,\gamma),\nu)$ represents the p.g.f. of the {\it discrete stable distribution} \cite{Christoph1998a}. In the limit $\gamma\to 1^{-}$, this compound distribution reduces to the Poisson distribution with combinants $\lambda_j=\delta_{j,1}\nu$, whereas in the limit $\gamma\to 0$ it converges to the logarithmic distribution. For $\gamma<0$, the function $Q(w;\theta,\gamma)$ corresponds to the zero-truncated negative binomial distribution with parameter $k=-\gamma$ and p.g.f.
\begin{equation} \label{eq:nbdtrunc}
Q(w;\theta,k)
=\frac{H_{NBD}(w;\theta,k) -H_{NBD}(0;\theta,k)}{1-H_{NBD}(0;\theta,k)}\;,\qquad 
H_{NBD}(w;\theta,k):=\left(\frac{1-\theta}{1-\theta w}\right)^k\, .
\end{equation}
\end{enumerate} 
\end{Example}
The central property of the combinants is encapsulated in the following theorem.
\begin{Theorem}
\label{comb2pmf}
Let $N \in \mathbb{Z}$ be a random variable with p.g.f. $H(w)$. If
\[
\log H(w)=\log H(0)+ \log \mathbb{E}[w^N]
\]
is analytic at $w=0$, then the probability $p_n=\mathbb{P}\{N=n\}$ is uniquely determined by the first $n$ combinants $\lambda_1, \ldots, \lambda_n$, and {\it conversely}: for $0<j\leq n$, each combinant $\lambda_j$ depends only on 
$\{p_0,p_1,\ldots,p_n\}$.
\end{Theorem}
\begin{proof}
The analyticity of $\log H(w)$ at $w=0$, together with the normalization condition $\log H(1)=0$, implies the representation
\[
\log H(w)= \sum_{j>0}\lambda_j( w^j -1),
\]
and hence $\lambda_0=\log p_0=-\sum_{j>0}\lambda_j$. Applying Faà di Bruno’s formula to
\[
H(w)=\exp\Big[\sum_{j>0}\lambda_j( w^j -1)\Big]
\]
yields
\begin{equation}\label{FdB}
p_n=\frac{H^{(n)}(0)}{n!} =  \sum \frac{e^{-\lambda_0}}{m_1!m_2!\cdots m_n! \, n!} \prod_{j=1}^{n}\lambda_j^{m_j}\,,
\end{equation}
where the sum ranges over all $n$-tuples $(m_1,m_2,\ldots, m_n)$ of nonnegative integers satisfying the constraint $\sum_{j=1}^{n} j m_j =n$.
\end{proof}
\begin{Proposition}\label{fincumul}
Let $H(w)$ denote the p.g.f. of a random variable $N \in \mathbb{Z}$ that possesses only a finite number $J = J_{+} + J_{-} < \infty$ of non-zero combinants, among which $J_{+}$ are positive and $J_{-}$ are negative. Then $H(w)$ admits the representation
\[
H(w) = \frac{H_{+}(w)}{H_{-}(w)}\,,
\]
where $H_{\pm}(w) = P(G_{\pm}(w);\lambda)$ are compound Poisson p.g.f.'s and $G_{\pm}(w)$ are polynomials of degree $J_{\pm}$.
\end{Proposition}

\begin{proof}
Decompose the sum in Eq.~\ref{combdef} into its positive and negative parts:
\[
\sum_{j>0}^{J} \lambda_j (w^j - 1)
=
\sum_{j_{+}=1}^{J_{+}} \lambda_{j_{+}} (w^{j_{+}} - 1)
+
\sum_{j_{-}=1}^{J_{-}} \lambda_{j_{-}} (w^{j_{-}} - 1)
=
\lambda_{+} (G_{+}(w) - 1) + \lambda_{-} (G_{-}(w) - 1)\,,
\]
where
\[
\lambda_{+} := \sum_{j_{+}=1}^{J_{+}} \lambda_{j_{+}} > 0,
\qquad
\lambda_{-} := \sum_{j_{-}=1}^{J_{-}} \lambda_{j_{-}} < 0.
\]
Exponentiation of both sides yields
\begin{equation}\label{infdratio}
H(w;\lambda)
=
\frac{\exp\bigl[\lambda_{+}(G_{+}(w) - 1)\bigr]}
     {\exp\bigl[\lambda_{-}(G_{-}(w) - 1)\bigr]}
=
\frac{\exp\bigl[\sum_{j_{+}=1}^{J_{+}} \lambda_{j_{+}} (w^{j_{+}} - 1)\bigr]}
     {\exp\bigl[\sum_{j_{-}=1}^{J_{-}} \lambda_{j_{-}} (w^{j_{-}} - 1)\bigr]}
=
\frac{H_{+}(w;\lambda_{+})}{H_{-}(w;\lambda_{-})}\,,
\end{equation}
which establishes the asserted factorization.
\end{proof}  
\begin{Remark} Although each of the p.g.f.s $H_{+}(w;\lambda_+)$ and $H_{-}(w;\lambda_-)$ corresponds to an infinitely divisible random variable on $\mathbb{Z}$, their ratio $H(w)=H_{+}(w)/H_{-}(w)$ does not, in general, define a valid p.g.f. on all of $\mathbb{Z}$. To illustrate this phenomenon, consider 
\[
H_{+}(w;\lambda_+) = e^{\lambda_{+}(w-1)} 
\quad\text{and}\quad 
H_{-}(w;\lambda_-) = e^{\lambda_{-}(w^2-1)}
\]
with $\lambda_{+}\neq\lambda_{-}$. The associated p.m.f. $\{p_n\}_{n\ge 0}$ then satisfies the recurrence relation
\begin{equation}\label{comb12}
(n+2)p_{n+2} \;=\; \lambda_{+}p_{n+1} \,-\, 2 \lambda_{-}p_{n}, 
\qquad 
p_0=e^{-\lambda_{+}+\lambda_{-}}, 
\qquad 
p_1=\lambda_{+}p_{0}.
\end{equation}
Depending on the specific values of $\lambda_{\pm}$, this sequence $\{p_n\}$ defines a valid p.m.f. only up to a finite index $J=n+1$, beyond which $p_{n+2}<0$ and the nonnegativity condition fails. For example, when $\lambda_{+}=10$, $\lambda_{-}=2$ the first negative term occurs $n=4$; and for $\lambda_{+}=100$, $\lambda_{-}=2$ at $n=30$, and so on.
\end{Remark}

\subsection{Thinning, self-decomposability, stability and scalability} \label{sfdss}
\subsubsection{Thinning}
To preserve the discrete character of r.v. $M\in \Z$ under multiplication by a real number $0<a<1$, consider the random sum $S_M=N_1(a)\!+\!N_2(a)\!+\! \ldots\! +\!N_M(a)$, where each of the r.v. $N_i(a)$ has two-valued Bernoulli distribution $p_0=1-a$ and $p_1=a$.  The {\it a-fraction}  of $M$, equivalently the   {\it thinning} operation $\odot$, is defined as \cite{StvH79}: 
\vspace*{-.1cm}
\begin{equation}
a \odot M:=   \sum_{i=1}^M N_i(a)\,,\;\;\forall a \in (0,1)\,.
    \label{thin}
\end{equation}
If $H(w)$ is the p.g.f. of  $M$ and $\mathcal{B}_{a}(w))=1-a + aw$ the p.g.f. of $N_i(a)$, then $K=a\odot M$ has p.g.f. $H(\mathcal{B}_{a}(w)))$.  For $\E M < \infty$, we have $\E K = a \cdot \E M$ and the following proposition holds.

\begin{Proposition}\label{prop:scal}
Let $\mu_{[k]},\;k=1,2,3,\ldots $ be the factorial moments of the r.v. $M \in \mathbb{Z}_+$ with  $\E M=\mu_{[1]}< \infty$. Then the {\it scaled factorial moments}   
\begin{equation}
\label{eq: Fj}
F_k := \frac{\mu_{[k]}}{(\mu_{[1]})^k} 
\end{equation}
are thinning-invariant. 
\end{Proposition}
\begin{proof}
Let  $H(w)$ and $R(w)=H(\mathcal{B}_{a}(w))$ be the p.g.f. of r.v. $M$ and $a\odot M$, respectively.  Then
\[ R^{(k)}(w)=a^k H^{(k)}(1\!-\!a(1\!-\! w))\;,\;\;\;
\frac{R^{(k)}(1)}{(R^{(1)}(1))^k}\!=\!\frac{a^jH^{(k)}(1)}{(a H^{(1)}(1))^k}\!=\! F_k \]   
\end{proof}

\begin{Theorem}\label{PSD_thin}
Suppose that the thinning family $\{H(\mathcal{B}_a(w)),\; a\in (0,1]\}$ is a PSD family $\{Z(\theta w)/Z(\theta),\; \theta \in (0,1]\}$ with $\theta(a)$ depending on $a$. Then, for some suitable constants $A$ and $\gamma$,
\begin{equation}\label{thfam}
H(w) = (1-A(1-w))^{\gamma} \;.
\end{equation}  
\end{Theorem}
\begin{proof}
For simplicity, assume $\theta(1)=1$, so $H(w)=Z(w)$. Then
\begin{equation}\label{thpow1}
H(1-a+aw)=\frac{Z(\theta w)}{Z(\theta)}= \frac{H(\theta(a)w)}{H(\theta(a))}\;.
\end{equation}
Setting in (\ref{thpow1}) $w=0$ gives
$H(\theta(a))H(1-a)=H(0)$,  and so  
\begin{equation}\label{thpow2}
H(1-a+aw)H(0) = H(\theta(a)w)H(1-a).   
\end{equation}
Letting $h(w)=\log H(w)$, equation (\ref{thpow2})  becomes
\begin{equation}\label{thpow3}
h(0) + h(1-a+aw) = h(1-a) + h(\theta(a)w).  
\end{equation}
Differentiating once and twice  in $w$ yields
\begin{equation}\label{thpow4}
a h^{\prime}(1-a+aw) = \theta(a) h^{\prime}(\theta(a)w)\;\; \text {and}\;\; a^2 h^{\prime \prime}(1-a+aw) = \theta^2(a) h^{\prime \prime}(\theta(a)w).  
\end{equation}
Setting $w=0$:
\begin{equation}\label{thpow6}
\frac{h^{\prime \prime}(1-a)}{(h^{\prime}(1-a))^2}=\frac{h^{\prime \prime}(0)}{(h^{\prime}(0))^2}
\end{equation}
Integrating  (\ref{thpow6}) yields the conclusion of the theorem.
\end{proof}

Two special cases of (\ref{thfam}) are worth noting. For $A=a$, $\gamma=m$, this gives the Binomial distributions (\ref{Binom}) with parameters $a$ and $m$. For $A=-\theta/(1-\theta)$ and $\gamma=-k$, we obtain the NBD p.g.f. (\ref{NBD}) with parameters $\theta$ and $k$. 

The following result connects the thinning family with an infinitely divisible family.
\begin{Theorem}\label{eq: infdthin}
Suppose that the thinning family $\{H(\mathcal{B}_a(w)),\; a\in (0,1]\}$ is a family of infinite divisible p.g.f. $\{H(w;\lambda),\; \lambda>0 \}$ with $\lambda(a)$ depending on $a$. Then 
\[H(w;\lambda) = \exp\{\lambda(G(w)-1)\}=
\exp\{-A (1-w)^{\gamma}\} ,\] 
where $A>0$ and $0<\gamma \leq 1$. (We assume the corresponding r.v. is non-degenerate).
\end{Theorem}
\begin{proof}
We have ~$\exp\{\lambda(G(w)-1)\} = H(1\!-\!a(1-\!w))$.~
Taking  logarithms with $h(w)=\log H(w)$ gives
\begin{equation}\label{eqin}
\lambda (G(w)-1) = h(1-a+a w)\;.    
\end{equation}
Since$G(0)=0$ Eq. \ref{eqin}, setting $w=0$ reduces (\ref{eqin}) to  $-\lambda=h(1-a)$, so
\begin{equation}\label{eqhg}
h(1-a) (G(w) -1) = -h(1-a+a w)\;.    
\end{equation}
Differentiating both sides of (\ref{eqhg}) with respect to $w$:
\begin{equation}\label{eqpr}
h(1-a)G^{\prime}(w) = -a h^{\prime}(1-a+a w).
\end{equation}
Denoting $b=1-a$ and $\gamma = G^{\prime}(0)$, setting $w=0$ in (\ref{eqpr}) gives $\gamma h(b)=(b-1)h^{\prime}(b)$
 Integrating over $b$ yields $h(b)= c_1(b-1)^{\gamma}$, so after substitution back into  (\ref{eqhg}):
\[
 \lambda(G(w)-1)= -h(1-a(1- w))=-c_1(-a(1-w))^{\gamma}\;.
 \]
At $w=0$: $\lambda=-c_1(-a)^{\gamma}= c_1(-1)^{\gamma+1}a^{\gamma}$. Taking $A=a^{\gamma}$ and $c_1=\lambda(-1)^{1-\gamma}$ completes the proof. 
\end{proof}

\subsubsection{Self-decomposability and stability}
The thinning operator $\odot$ allows us to define a subclass of infinitely divisible r.v. $N\in \Z $, called {\it self-decomposable} , satisfying \cite{StvH2004}
\begin{equation}\label{seldec}
  N= a \odot N + M_{a}\;,\;\;  H(w)= H(\mathcal{B}_{a}(w))H_a(w)\;. 
\end{equation} 
Here $H(w)$ and $H_a(w)$ are the p.g.f. of $N$ and $M_a$, respectively. We first show that $N$ is indeed infinitely divisible.
\begin{Theorem}{\rm \cite{StvH2004}}\label{sdinfd}
A self-decomposable distribution on $\Z$ is infinitely divisible.
\end{Theorem}
\begin{proof}
Consider limit $\lim_{a\to 1^{-}}H_a(w)$, which must hold also for $w=0$ since $a \in (0,1)$, $w \in (0,1)$:
\[
1-H_a(0) \leq \{1-H_a(w) \} + \{H_a(w)-H_a(0) \} \leq
\{1-H_a(w) \} + \frac{w}{1-w}\;.
\]
It follows $H_a(0) >0$ for all $a$ sufficiently close to one. Applying (\ref{seldec}) for such $a$ shows that  $H(0) >0$ as well. Noting that the $R$-function of  $H(w)$ defined in (\ref{Rw}) can be written as a limit of completely monotone functions,
\[
R(w)=\frac{1}{H(w)}\lim_{a \to 1^{-}}\frac{H(1\!-\!a\!+\!wa)-H(w)}{(1-a)(1-w)}=\lim_{a \to 1^{-}}\frac{1}{1-a}\frac{1-H_a(w)}{1-z}\;,
\]
we conclude that $R$ is absolutely monotone. By Theorem (\ref{HexpR}), $H(w)$ is infinitely divisible.
\end{proof}

Since $H(w)$ is infinitely divisible, so is $H(\mathcal{B}_{a}(w))$ (by Proposition \ref{infdthin}). That $H_a(w)$ is also absolutely monotone and infinitely divisible is less immediate \cite{StvH2004}. Last but not least, the sum of two self-decomposable r.v. $N_1+N_2$ is also self-decomposable.

\begin{Example} Consider a gas of particles with short-range (approximately point) interactions occupying volume $V$. Assuming that $N$ is an infinitely divisible r.v. with finite mean $\E N \equiv \kappa_1<\infty$, the distribution of the particles in a sub-volume $V_1 \subset V$ is described by the r.v.$N_1=a\odot N$,  where $a=V_1/V$.  Clearly, $F_j(V_1)=F_j(V)$. Moreover, the particle distribution in the complementary sub-volume $V_2=V-V_1$ is described by the infinitely divisible r.v. $N_2=N-N_1$, too. 
\end{Example}

Consider the thinning of Sibuya p.g.f.:
\begin{equation}\label{Scaled_sib}
\mathcal{S}(\mathcal{B}_a(w);\gamma)=1-a^{\gamma}(1- w)^{\gamma}:=\mathcal{S}(w;\gamma,\nu)\;;\;\;\; \nu=a^{\gamma} \;.  
\end{equation}
Here $\mathcal{S}(w;\gamma,\nu)$ is the p.g.f. of {\it scaled Sibuya distribution} \cite{Christoph2000, Klebanov2023-xj}. By allowing $\nu \in(0,1 \rangle$, the scaled Sibuya distribution is a zero-extended version of the ordinary Sibuya $\mathcal{S}(w;\gamma,1)=\mathcal{S}(w;\gamma)$, closed under thinning $\mathcal{S}(\mathcal{B}_b(w);\gamma,\nu) = \mathcal{S}(w;\gamma,b^{\gamma}\nu)$. 
 
\begin{Proposition}\label{Sib}
 Let $\mathcal{S}(w;\gamma,\lambda)=1-\lambda(1-w)^{\gamma}$ be a scaled Sibuya p.g.f.
 
 (1) If  $a \neq b$ and $p$ are some constants such that
 $a,b \in (0,1),\; p\in\langle 0,1\rangle$, then 
 \begin{equation}\label{scalsibmix}
 p\cdot \mathcal{S}(w;\gamma,\lambda_1)+(1-p)\cdot \mathcal{S}(w;\gamma,\lambda_2)=\mathcal{S}(w;\gamma, \nu) \;,\;\; \nu=p\lambda_1+(1-p)\lambda_2\;.
 \end{equation}
 
(2) If $\lambda(x)\in (0,1)$ is a function and $f(x)$ is a p.d.f., then
\begin{equation}
 \int_{0}^{\infty}\mathcal{S}(w;\gamma,\lambda(x)) f(x)dx = \mathcal{S}(w;\gamma,\nu)\;,\;\;\;
\nu= \int_{0}^{\infty}\lambda(x) f(x)dx\;.   
\end{equation}
\end{Proposition}
\begin{Remark}
With $N_{1,2}=\mathbb{P}(N_{1,2}>0)$  and $\lambda = \mathbb{P}(N>0)$, Eq. \ref{scalsibmix} yields
\[
\mathbb{P}(N>0)=p\cdot \mathbb{P}(N_{1}>0)+(1-p)\cdot \mathbb{P}(N_{2}>0)
\]
\end{Remark}

Replacing $\mathcal{S}(w;\gamma,\lambda)$ in (\ref{scalsibmix}) by an arbitrary p.g.f. $G(w)$ leads to the following characterization of the scaled Sibuya distribution.

\begin{Theorem}\label{ThSib} Let $G(w)$ be a p.g.f. The relation
\begin{equation}\label{eqSib}
p\cdot G(\mathcal{B}_a(w))+(1-p)\cdot G(\mathcal{B}_b(w)) = G(\mathcal{B}_c(w))
\end{equation}
holds for all $p,a,b \in (0,1)$ and $c=(pa^{\gamma}+(1-p)b^{\gamma})^{1/\gamma}$ for a fixed $\gamma \in (0,1)$ if and only if 
\begin{equation}\label{eqSib1}
G(w)= 1-\lambda (1-w)^{\gamma}=\mathcal{S}(w;\gamma,\lambda)\;.   
\end{equation}
\end{Theorem}
\begin{proof}
1. By substituting (\ref{eqSib1}) into (\ref{eqSib}) verifies that it is indeed a solution. 

2. Suppose $G(w)$ is any solution. Denote $Q(w)=G(w)-(1-(1-w)^{\gamma})$. Clearly, $Q(1)=0$ and 
\[ p Q(1-a+aw)+(1-p)Q(1-b+bw)=Q(1-c+cw). \]
Substituting $u=1-w$, $R(u)=Q(1-u)$ gives
$p R(au)+(1-p)R(bu) = R(cu)$.
As $a \to 0$ we have $c \to (1-p)^{1/\gamma}b$, $R(au) \to 0$, so $ (1-p) R(bu) = R((1-p)^{1/\gamma}bu)$.
Substituting $R(u)=u^{p}K(u)$ yields
$ K(bu)=K((1-p)^{1/\gamma}bu)$; since the right hand side depends on $p$ but the left-hand does not, $K(u)=k=const$. Returning from $K$ to $G(w)$  completes the proof.
\end{proof}

\begin{Corollary}\label{DSD_unique}
    Let $N, N_1, N_2 \in \Z$ be \iid infinitely divisible r.v.  with the common p.g.f. $H(w;\lambda)=P(G(w);\lambda)$. The equation $c\odot N = a \odot N_1 + b\odot N_2$ holds iff the secondary p.g.f. $G(w)$ represents the scaled Sibuya distribution.
\end{Corollary}
\begin{proof}
Writing 
$H(\mathcal{B}_c(w);\lambda)=e^{\lambda(G(\mathcal{B}_c(w))-1)}=e^{\lambda_1(G(\mathcal{B}_a(w))-1)}\cdot e^{\lambda_2(G(\mathcal{B}_b(w))-1)}$ and setting $\lambda_1=p\lambda$ and $\lambda_2=(1-p)\lambda$, we find that the secondary p.g.f. $G(w)$ must satisfy Eq. \ref{eqSib}. By Theorem \ref{ThSib}, the unique solution is the scaled Sibuya p.g.f.
\end{proof}
The compounding of the Poisson p.g.f. $P(w;\lambda)$ with the scaled Sibuya gives a compound Poisson-Sibuya distribution with a different scale parameter (see  Eq.\ref{zero-ext-compP}):
\[
P(\mathcal{S}(w;\gamma,\nu);\lambda) = P(\mathcal{S}(w;\gamma);\nu \lambda)\;.
\]
 The distribution is referred to in the literature as the {\it Discrete Stable} distribution (DSD) \cite{Christoph2000, StvH2004}.

Another route to the {\it stable distribution} is via the following construction.
Suppose $\{ H(\mathcal{B}_a(w); \lambda),\; a\in (0,1]\}$ is the thinning family of infinitely divisible p.g.f.  $H=P_{\lambda}\circ\mathcal{S}_{\gamma}$. Then
\begin{equation}\label{eq: dsdthin}
H(\mathcal{B}_a(w),\lambda)=P_{\lambda}\circ\mathcal{S}_{\gamma} \circ \mathcal{B}_a = P_{\lambda a^{\gamma}}\circ \mathcal{S}_{\gamma}= H(\mathcal{S}(w;\gamma);\lambda a^{\gamma})\;.  
\end{equation}
This leads to the double-thinning formulation of Eq.\ref{seldec}: 
$H(w)= H(\mathcal{B}_{a}(w)) H(\mathcal{B}_{b}(w))$,  
\begin{equation}\label{eq: dsdpgf}
P_{\lambda}\circ \mathcal{S}_{\gamma} = (P_{\lambda}\circ\mathcal{S}_{\gamma} \circ \mathcal{B}_a) \cdot (P_{\lambda}\circ\mathcal{S}_{\gamma} \circ \mathcal{B}_b)\;;\;\;\; a^{\gamma}+b^{\gamma}=1, \;\forall \gamma \in(0,1]\;.
\end{equation}
It is worth noting another unique property of the DSD: by Theorem \ref{eq: infdthin}, among all infinitely divisible p.g.f.{}s $P(G(w);\lambda)$, the DSD is the only one with a scale-dependent parameter $\lambda=a^\gamma$.

\subsubsection{Scalability}\label{scalsub}


There is a wide class of self-decomposable p.g.f. (e.g. Poisson (\ref{poispgf}), Binomial (\ref{Binom}) and NBD (\ref{NBD}) distributions) analytical around the point $w_0=1$, and  thus admit the expansion 
\begin{equation}
 H(w)=\sum_{n\geq 0} p_n w^n= 1+\sum_{n>0} b_n (1-w)^n \, , \;\;\;  \sum_{n>0} b_n > -1\, .
\end{equation}
 This category also includes the {\it Generalized Linnik distribution} (GDLD) \cite{Christoph1998} -- the compound Poisson-Logarithmic-Sibuya distribution  \cite{Christoph1998} -- with p.g.f.
 \begin{equation}\label{GDLDpgf1}
Q(w;\lambda,\beta,\gamma)=
P_{\lambda}\circ\mathcal{L}_{\theta}\circ S_{\gamma}
=(1+\lambda(1-w)^{\gamma}/\beta)^{-\beta}\;;\;\;\theta=1-e^{-\lambda/\beta}\;, \;\beta>0  \end{equation}
as well as its special cases: {\it Discrete Mittag-Leffler distribution} \cite{Pillai1995, Huillet2019} ($\beta=1$) which is a compound geometric-Sibuya distribution with the p.g.f. (\ref{DMLpgf}), the {\it Discrete Linnik distribution} \cite{Huillet2019, Dev1993} ($\lambda/\beta=c$), and DSD (\ref{eq: dsdpgf}) ($\beta=\infty$). 

Apart from the Binomial, all p.g.f. mentioned above are scalable: for any $a>0$, the function $H(a(1-w)$ is again the p.g.f.. The general construction of such p.g.f.{}s rests on the following observation. Let $\varphi (s)$ be the Laplace transform of a positive r.v. $X$. Then $H(w)=\varphi (1-w)$, and the change of variable $w \to 1 - a(1-w)$ gives $H(1-a(1-w))=\varphi (a(1-w))$. In other words, while the thinning operator $w \rightarrow \mathcal{B}_a(w)$ corresponds to a scale change from $1$ to $a$ with $a \in (0,1)$, the function $\varphi(as)$ remains a Laplace transform for any positive $a$. 

 This motivates the following problem: {\it characterize all p.g.f.{}s $H(w)$ for which $H(1-a(1-w))$ is again a p.g.f. for all $a>0$}. 


\begin{Theorem} {\rm \cite{ Klebanov2023-xj, KlSum2026} }\label{thK}
Let $H(w)$ be a p.g.f.. Then $H(1-a(1-w))$ is the p.g.f. for all $a>0$ iff there exists a Laplace transform $\varphi (s)$ of a r.v. $X\in \R$  such that $H(w)=\varphi (1-w)$. 
\end{Theorem} 
The proof exploits the fact that if $\varphi(s)$ is the Laplace transform of a distribution function, then so is $\varphi(as)$ for any $a>0$. Therefore, $\varphi(a(1-w))=H(1-a(1-w))$ is a p.g.f. for any $a>0$. We will call all r.v.s $M \in \Z$ with such a p.g.f. {\it scalable}.

The Theorem \ref{thK} raises the question which p.g.f.{}s are non-scalable, i.e., for which the inverse Laplace transform does not exist. This occurs when $H(w)$ fails to be exponentially restricted: $sup_{w>0}H(w)e^{-bw}<\infty$ for some $b>0$. 

Another statement is the following. Let $X$ be a nondegenerate random variable having p.g.f. $H(w)$ and $p_k=\Pr\{X=k\}$. If $H(w)$ is scalable then $p_k>0$ for all $k=0,1,2,\ldots$. For the proof, it is sufficient to express probabilities $p_k$ through series coefficients of $\varphi(w)$.

\begin{Proposition}\label{nonscal}
An infinitely divisible r.v. $N\in \Z$ with a finite number of combinants $1<j<\infty$ is not scalable. 
\end{Proposition}
\begin{proof}
By Proposition \ref{fincumul},  $N$ has p.g.f. $G(w;\lambda)=e^{-\lambda(1-\sum_j^J q_j w^j)}$, with $\sum_{j>0}^J q_j=1$ and $\lambda=\sum_{j>1}^J\lambda_j$. For $1<J<\infty$, this p.g.f. is no longer of exponential form, so its inverse Laplace transform does not exist.
\end{proof}

\begin{Remark}\label{reKl}
Any r. v. concentrated on the non-negative semi-axis may be approximated by a scalable r.v. on an arithmetic progression with a small step. Let $Y$ be a non-negative r.v. with p.g.f. $H(w)$, and let $\xi(s)=H(e^{-s})$ be the Laplace transform of its distribution. Consider another r.v. $Z_A$, depending on an integer parameter $A>0$, with p.g.f. $\zeta(w,A) = \xi(1-(1-A)-A w^{1/A}))$. By Theorem \ref{thK}, $\zeta(w, A)$ is a p.g.f of a r.v. with values on the arithmetic progression $\{0,1/A,2/A, \ldots \}$. 
\end{Remark}

\begin{Proposition}\label{cor:infd2}
If r.v. $X\in \R$ with the p.d.f. $f(x)$ is infinitely divisible, then the r.v. $N \in \Z$ with the p.g.f. 
 \begin{equation}
\label{eq: lapdens}   
H(w)=\int_0^\infty e^{-(1-w)x}f(x) dx \,,
\end{equation}
It is also infinitely divisible. Conversely, if $N$ is an infinitely divisible r.v., then $X$ is an infinitely divisible r.v.
\end{Proposition}
\begin{proof}
 $X$ is infinite divisible iff $\forall n\in \N$  there exists a Laplace transform $\varphi_n(s)$ such that $\varphi(s)=(\varphi_n(s))^n$. Since $(H(w))^{1/n}=\varphi_n(1-w)$, for every $n$, the p.g.f. $H(w)$ is infinitely divisible. The converse is immediate.
\end{proof}
\begin{Proposition}{\rm \cite{StvH2004}}
An integer-valued, bounded, non-degenerate r.v. $X$ is neither infinitely divisible nor scalable.
\end{Proposition}
\begin{proof}
Assume $X$ is bounded by some $m>0$.
Then its p.g.f. $\mathcal{P}_X(z)$ is a polynomial of degree at most $m$. If $X$ were scalable, $\mathcal{P}_X(1-w)$ would have to be a Laplace transform; but no polynomial of degree greater than zero is a Laplace transform. Furthermore, the characteristic function of any bounded r.v. has zeros, so it cannot be infinitely divisible \cite{Linnik}. 
\end{proof}

\begin{Proposition} \label{cor:g} {\rm \cite{Klebanov2023-xj}} 
Let $N \in \Z$ be a scalable r.v. with p.g.f. (\ref{eq: lapdens}). Then 
\begin{equation}\label{frac2mom}
 \E (N)_n =\mu_{[n]}=H^{(n)}(1) = \int_0^{\infty}x^n f(x) dx=\E X^n  \end{equation}
 and it's p.m.f.  
\begin{equation} \label{eq: lapder}
 p_n=\frac{H^{(n)}(0)}{n!},\;\;\;   H^{(n)}(w)= \varphi^{(n)}(1-w) =\int_0^{\infty } e^{-(1-w) x} x^n f(x) \, dx
\end{equation} 
satisfies the one-step recurrence 
\begin{equation}
\label{eq: gnkl}
p_{n+1}=\frac{p_n}{n+1}g(n);\;\;\;\; 
g(n) = \frac{H^{(n+1)}(0)}{H^{(n)}(0)}=\frac{\int_{0}^{\infty}e^{-x}x^{n+1}f(x) dx}{\int_{0}^{\infty}e^{-x}x^{n} f(x) dx}\,.
\end{equation}
\end{Proposition}

We remark in passing that the invariance of the scaled factorial moments $F_n=\mu_{[n]}/(\mu_{[1]})^n$ under $N\to a\odot N$ is equivalent to the invariance of the ordinary scaled moments $\E X^n/(\E X)^n$ under the scale transformation $X\to aX$.

\begin{Remark}
Expanding $e^{-x}$ in (\ref{eq: gnkl}) together with Eq.(\ref{frac2mom}) allows us to express  (\ref{eq: gnkl}) as a ratio of two sums of factorial moments:
\begin{equation}\label{facrat}
\int_{0}^{\infty}e^{-x}x^{n} f(x) dx=\sum_{j=0}^{\infty}\frac{(-1)^j}{j!}\E X^{n+j}=
\sum_{j=0}^{\infty}\frac{(-1)^j\mu_{[n+j]}}{j!}   
\;.
\end{equation}
If the factorial moments satisfy a recurrence relation $\mu_{[n+1]}=f(\mu_{[n]})$ ( as in the NBD, where $\mu_{[n+1]}=\theta(n+k)/(1-\theta)\mu_{[n]}$), then the infinite sum (\ref{facrat}) reduces to a finite polynomial in $n$. 
\end{Remark}

The examples of scalable infinitely divisible p.g.f.{}s are presented in Appendix \ref{Exscal}

\begin{Remark}

By Proposition \ref{Sib}, every scaled Sibuya p.g.f. is an infinite mixture of scaled Sibuya p.g.f. $\mathcal{S}(w;\gamma,\lambda(x))$ with parameter $\lambda(x)$ depending on $x$  and mixing p.d.f. $f(x)$.  Taking $\lambda(x)=e^{-(1-w) x}$ we obtain 
\[
 \nu(w)=\int_0^{\infty}e^{-(1-w) x}f(x)dx\;,\;\;\; \int_{0}^{\infty}\mathcal{S}(w;\gamma,e^{-(1-w) x}) f(x)dx =     1-\nu(w)(1-w)^{\gamma}\;.
 \] 
 Thus $\nu(w)$ represents a p.g.f. of some scalable distribution on $\Z$. The p.g.f.
 \[
 \mathcal{S}(w;\gamma,e^{-(1-w) x})=1-e^{-(1-w) x}(1-w)^{\gamma})\;,\;\;  \mathcal{S}(0;\gamma,e^{-(1-w) x})=1+\sinh (x)-\cosh (x)
 \]
 corresponds to the p.m.f. with the first four terms 
 \[
 p_0=1-e^{-x}\;,\; 
 p_1=e^{-x}(\gamma-x)\;,\;\;
 p_2=\frac{1}{2} e^{-x} (\gamma-(x-\gamma)^2)\;,\;\; p_3=
 \frac{e^{-x}}{6}(2\gamma-((x-\gamma )^2-3 \gamma) (x-\gamma ))\;.
 \]
     
\end{Remark}

\newpage
\section{Kinetic models of particle-number distributions} \label{kinmodels}
\subsection{Stationary Solutions of Birth-Death Equations}
Consider a continuous-time birth-and-death process (B-D) consisting of random creation and annihilation events (e.g., creations and disintegrations of particles \cite{Feller1, Anderson}). The process is a Markov chain with a discrete state space $n\in \N$, in which only transitions to nearest neighbors $n\to n \pm 1$ are allowed. It is described by stochastic process $\{N(t): t \geq 0\}$
\[\mathbb P(N(s+t)=j)|N(s)=i)= \mathbb P(N(t)=j)|N(0)=i)=P_{i,j}(t); \;\; \forall i,j, \;t>0,s>0
\]
with transition probabilities $P_{i,j}(t)$ satisfying the Chapman-Kolmogorov equation.
\[P_{i,j}(s+t)=\sum_{k}P_{i,k}(s)P_{k,j}(t).
\]
The time evolution of $p_j(t)=\mathbb P (N(t)\!=\!j)$ is governed by the differential-difference equations 
\begin{eqnarray}\label{dd}
\frac{dp_{j}(t)}{dt}&=&-(\lambda_j +\mu_j) p_{j}(t) + \lambda_{j-1} p_{j-1}(t) +\mu_{j+1}p_{j+1}(t), \;\; j > i\,,
\\ \nonumber 
 \frac{dp_i(t)}{dt}&=&-\lambda_i  p_i(t) + \mu_{i+1} p_{i+1}(t)\;, \;\; p_j(0)=\delta_{ij}\;,
\end{eqnarray}
where 
\[    \lambda_j=\frac{dP_{j,j+1}(t)}{dt}\Big |_{t=0^{+}}\,, \;\;\;\;
    \mu_j=\frac{dP_{j,j-1}(t)}{dt}\Big |_{t=0^{+}}\;.
\]

One commonly used model considers the transition $j\to j\pm 1$ to result from several underlying microscopic processes at the particle level, $k \to k\pm 1$, $k=0,\ldots,j$, with elementary transition probabilities $\alpha_k$ and $\beta_k$. A particle enters the system of $j$ particles  at rate $\alpha_0$, is annihilated at rate $\beta_0 j$, each particle splits into two at rate $\alpha_1 j$, and two particles fuse at rate $\beta_1 j(j-1)$ {\it etc.} 
\begin{equation}
\label{albe}
\lambda_j  =\sum_{k=0}^{\ell} \alpha_k (j)_k, \; \;\;\;
\mu_j  =  j \sum_{k=0}^{m} \beta_k (j-1)_k; \; \;\;\; \ell \le j ;\; m \le j ;\;\;\;\;\; (j)_k \!:=\! \prod_{i=0}^{k-1} (j\!-\!i)
\end{equation}

It is well known that for sufficiently large times $t$, the solution of Eq.~ \ref{dd} can be well approximated by the time-independent limiting probabilities $p_j=\lim_{t\to\infty}p_j(t)$ \cite{Feller1, Anderson}.  In this case, Eqs.~ (\ref{dd}) become two functional equations for the unknown function $F(j)$
\begin{eqnarray}\label{Fn}
F(j+1) &=& \mu_{j+1} p_{j+1} - \lambda_j p_{j} = \mu_jp_{j}-\lambda_{j-1}p_{j-1} = F(j)\;,\;\;j > i\\ \nonumber  
F(i+1) &= & \mu_{i+1}p_{i+1} - \lambda_i p_{i} \!=0
\end{eqnarray}
 with solution: 
\begin{equation}
\frac{p_{j+1}}{p_j}  = \frac{\lambda_j}{\mu_{j+1}} = \frac{\sum_{k=0}^{\ell} \alpha_k (j)_k}{(j+1)\sum_{k=0}^{m} \beta_k (j)_{k}}\;.
\end{equation}
We have thus proven the following theorem.
\begin{Theorem} {\rm  \cite{Klebanov2023-xj}}
The solution of the stationary B-D equation  (\ref{Fn})  with  $\lambda_j $  and $\mu_j$  given by ~(\ref{albe}) is:
\begin{equation}\label{eq:prodgn}
p_j=p_i\prod_{n=i}^{j-1}\frac{g(n)}{n+1} ; \; \;\;\;\; 
g(n) := \frac{(n\!+\!1)p_{n+1}}{p_n} = \frac{\sum_{k=i}^{\ell} \alpha_k (n)_k}{\sum_{k=i}^{m} \beta_k (n)_{k}}\,, \; \ell \leq j\,, \; m \leq j\;,
\end{equation}
where $p_i$ is the first non-zero value of the p.m.f. $p_j$. 
\end{Theorem}
Starting from $p_i\!>\!0$ where $i\!=\!\min_{p_n>0}(n)$, all subsequent values $p_j$ with $j>i$ are determined by a single  $g(n)$. Since distributions with the p.m.f.{}s satisfying (\ref{eq:prodgn}) represent a discrete Markov chain, we call them {\it Markovian distributions}, and  $g(n)$ their {\it transition function}.
\subsection{Markovian distributions} 
In Subsection \ref{psdsub}, Proposition \ref{G2PSD}, we established that any r.v. $K$ belonging to the PSD family admits a factorization of the form $K = N \& M$, where $N$ and $M$ are independent r.v.s such that $K$ attains the value $k$ if and only if $N = k$ and $M = k$. In this representation, $N$ takes values in $\Z$, while $M$ is geometrically distributed, and, conversely, any such pair $(N, M)$ induces a PSD-distributed random variable $K$. The following proposition elucidates the relationship between the PSD family and the Markovian family.
\begin{Proposition}\label{g_PSD}
Let $g_p(n)$ be the transition function of a Markovian distribution with the p.m.f. $p_n$. Then $g_r(n)=g_p(n)\theta$ is the transition function of the p.m.f. $r_n=p_n \theta^n/Z(\theta)$, where $Z(\theta)=\sum_n p_n \theta^n$.
\end{Proposition}
\begin{proof}
Let $q_n=(1-\theta)\theta^n$ be the p.m.f. of the geometric distribution. Then
\[
r_n=\frac{p_n q_n}{\sum_n p_n q_n}=
\frac{p_n\theta^n}{\sum_n p_n \theta^n}\;,\;\;g_r=\frac{(n+1)r_{n+1}}{r_n}
=\theta g_p(n)
\]
\end{proof}
The class of Markovian distributions encompasses many well-known distributions. Let us consider the simplest case: $g(n)=a n + b$ describing the Katz family of distributions \cite{JKK2005}. Writing    
\begin{equation} \label{gndef}
p_n=\frac{Z_n}{Z}\;,\;\; Z_n(i)=\prod_{j=i}^{n-1}\frac{g(j)}{j+1}\;,\;\;\;Z=\sum_{n\geq 0}Z_n\;,
\end{equation}
one obtains for $i=0$  Poisson, binomial, and negative binomial distributions: 
\begin{equation}\label{gnsimp1}
g_P(n)=\theta,\;\;\; g_N(n)=\frac{a}{1-a}(m-n),\;\;\; g_{NB}(n)=\theta (n+k)\;,
\end{equation}
and for $i=1$ the transition functions of the Logarithmic and Sibuya distributions:
\begin{equation}\label{gnsimp2}
g_L(n)=\theta n,\; g_{SB}(n)=n-\gamma\;.
\end{equation}

More complicated is the case of {\it Generalized Sibuya distribution} with parameters $\nu\geq 0$ and $0< \gamma<\nu +1$ \cite{Kozu2017}. Its transition function and p.g.f., respectively, read:
\begin{equation} \label{gsbd1}
  g_{GS}(n)= \frac{(n+1)(n+\nu-\gamma)}{n+\nu +1} \;,\;\; 
  Q(w;\nu,\gamma)  =\frac{w \gamma}{\nu + 1}\cdot\,
_2F_1(1, \nu-\gamma + 1;  \nu+ 2; w)\,,  
\end{equation} 
where $_2F_1(a_1,a_2;b_1;z)$ is the Gauss hypergeometric function. For $\nu=0$, this reduces to the Sibuya distribution \footnote{Another connection with the classical Sibuya variable $N$ is through the probability of the excess $\mathbb P(N-m=n|N>m)$ which has the generalized Sibuya distribution with $\nu=m$. This extends to the generalized Sibuya variable $N$ for which $\mathbb P(N-m=n|N>m)$ has also the generalized Sibuya distribution with $\nu+m$ as its parameter.}. Its PSD counterpart -- {\it Extended Generalized Sibuya distribution} with $g_{EGS}(n)=g_{GS}(n)\theta$ solves the stationary B-D equations (\ref{Fn}) with coefficients: 
\begin{equation}
\label{eq:abgensbd}
    \frac{\alpha_2}{\beta_1}=\theta,\;\; \frac{\alpha_1}{\beta_1}=(2+\nu-\gamma)\theta,\;\; \frac{\alpha_0}{\beta_1}=(\nu-\gamma)\theta,\;\; \frac{\beta_0}{\beta_1}=\nu+1\,.
\end{equation}

\renewcommand{\arraystretch}{1.5}
\setlength\tabcolsep{1. pt}
\begin{table}[hbt] 
    \centering
    \begin{tabular}{|l|l|c|c|c|}
    \hline
    {\bf No.}  &  {\bf Distribution} & {{\boldmath $Z_n=Z\cdot p_n$}} & $\mathbf {g(n)}$ & {\boldmath $ {\alpha_i,\; \beta_i}$} \\
        \hline  \hline  
      1& Poisson &$\ffrac{\theta^n}{n!}$ & $\theta$ & $\ffrac{\alpha_0}{\beta_0}=\theta$  \\
       \hline \hline
       2& Hyper-Poisson \cite{JKK2005} &$\ffrac{\theta^{n+\nu+1}}{(n+\nu+1)!}$ & $ \ffrac{n+1}{n+\nu}\theta$  & $\ffrac{\alpha_0}{\beta_0}=\ffrac{\alpha_1}{\beta_0}=\ffrac{\theta}{\nu}, \ffrac{\beta_1}{\beta_0}=\ffrac{1}{\nu} $   \\
        \hline \hline
       3& CMP \cite{Shm2005} \;\& $r=2$ &$\ffrac{\theta^n}{((n+1)!)^2}$ & $ \ffrac{\theta}{n+1}$  & $\ffrac{\alpha_0}{\beta_0}=\theta, \ffrac{\beta_1}{\beta_0}=1 $   \\
        \hline \hline
       4& Logarithmic &$\ffrac{\theta^n}{n}$ & $\theta n$  & $\theta=\ffrac{\alpha_1}{\beta_0}$   \\
         \hline \hline
       5& Shifted Logarithmic &$\ffrac{\theta^{n+1}}{n+1}$ & $\ffrac{(n+1)^2}{n+2}\theta $  & $\theta\!=\!\ffrac{2\alpha_2}{\beta_0}\!=\!\ffrac{2\alpha_1}{3\beta_0},\; \ffrac{\alpha_0}{\beta_0}\!=\!\ffrac{\beta_1}{\beta_0}\!=\!\ffrac{1}{2}$   \\
         \hline \hline
    6&  NBD & $(-\theta )^n \binom{-k}{n}$ & $(n+k)\theta$ & $k=\ffrac{\alpha_0}{\alpha_1},\;\theta=\ffrac{\alpha_1}{\beta_0}$ \\
         \hline \hline
    7&    Ext. Sibuya/NBD1& $(-\theta )^{n+1} \binom{\gamma }{n}$ & $(n-\gamma)\theta$  & $\gamma\!=\!-k=1-\ffrac{\alpha_1}{\alpha_2}$,\; $\theta=\ffrac{\alpha_2}{\beta_1}$   \\ 
         \hline \hline
    8&    Shifted Ext. Sibuya & $(-\theta )^{n+1} \binom{\gamma }{n+1}$ & $(n\!+\!1\!-\!\gamma)\ffrac{n+1}{n+2}\theta$  & $\gamma\!=\!\ffrac{\alpha_1\!-\!3\alpha_0}{\alpha_1\!-\!\alpha_0}, \theta
        \!=\!\ffrac{\alpha_2}{\beta_1}, \ffrac{\beta_1}{\beta_0}\!=\!\frac{1}{2}$   \\ 
        \hline \hline
      9&  Hurwitz (\!\ref{HurPSD}\!)\;\&  $s\!=\!1$ & $\ffrac{\theta^n}{a\! +\!n}$ & $\ffrac{(n\!+\!1)(n+a)}{n+a+1} \theta $  & $\theta\!=\!\ffrac{\alpha_2}{\beta_1},\;\ffrac{\alpha_1}{\alpha_0}\!=\!3,\; \ffrac{\beta_0}{\beta_1}\!=\!a\!+\!1$   \\
         \hline \hline
        10& Eq.\ref{jj1} & $\ffrac{\theta^n}{(n+2)(n+1)}$ & $\ffrac{(n+1)^2}{n+3}\theta$  & $\theta\!=\!\ffrac{3\alpha_2}{\beta_0}\!=\!\ffrac{\alpha_1}{\beta_0}\!=\!\ffrac{3\alpha_0}{\beta_0},\; \ffrac{\beta_1}{\beta_0}\!=\!\ffrac{1}{3}$   \\
        \hline 
    \end{tabular}
    \vspace{.2cm}
    \caption{Transition functions $g(n)$ and elementary probabilities $\alpha_i,\beta_i$ of selected PSD p.m.f. $p_n=Z_n/Z$, $\sum_n Z_n=Z$. }\label{tab:mark}
    \end{table}

The Table \ref{tab:mark} collects the results on transition functions $g(n)$ and elementary probabilities $\alpha_i,\beta_i$ of selected PSD.  It is worth mentioning that all shifted distributions shown here are infinitely divisible. The p.m.f. $\# 7$ with $g(n)=(n-\gamma)\theta$ describes the Extended Sibuya distribution when $\alpha_1 < \alpha_2$, and the zero-truncated NBD (NBD1) with $k=-\gamma=\alpha_1/\alpha_2-1$ when $\alpha_1>\alpha_2$. In the limit $\alpha_1 \to \alpha_2$:
\begin{equation}
\lim_{k \to 0}\frac{1 -(1- \theta w)^{-k}}{1- (1 -\theta)^{-k}}=\lim_{\gamma \to 0}\frac{1 -(1- \theta w)^{\gamma}}{1- (1 -\theta)^{\gamma}} = \frac{\log(1-\theta w)}{\log(1-\theta)} 
 \label{logpgf}   
\end{equation}
 it gives the Logarithmic p.g.f., and for $\gamma \to 1^{-}$ ($\alpha_1\to 0$), it collapses to a single-point distribution  $H(w) \to w$.

As a counterexample, we present two infinitely divisible distributions.
\begin{enumerate}[label=(\roman*), leftmargin=*]
\item The {\it Lagrangian Poisson (Poisson Consul}) distribution \cite{JKK2005}  satisfying the functional equation
\begin{equation}\label{lagpoispgf}
Q(w)=e^{\theta(G(w)-1)}\;,\;\;\; G(w)=wQ(w)
\end{equation}
is the compound Poisson, and therefore infinitely divisible. Its p.m.f. and $g(n)$ function read:
\begin{equation}\label{lagpois}
 p_n=\frac{\theta^n (n+1)^{n-1}}{n!}e^{-\theta(n+1)}\;,\;\;\; g(n)=e^{-\theta } \theta (n+1) \Big(\frac{n+2}{n+1}\Big)^n\;,\;0<\theta<1\;.
\end{equation}
This distribution is Markovian, but its function $g(n)$ is not a simple ratio of finite-order polynomials in $n$, and therefore is not a solution of a stationary B-D equations. 

\item  The {\it Hermite distribution} \cite{JKK2005} with p.m.f.
\begin{equation}\label{Hermite_pmf}
  p_n=e^{-\theta_1-\theta_2}\sum_{i=0}^{\left \lfloor{n}\right \rfloor}\sum_{j=0}^{\left \lfloor{i/2}\right \rfloor}\frac{\theta_1^{n-2j}\theta_2^j}{(i-2j)!j!},\;\;\;(n+1)p_{n+1}=\theta_1p_n+2\theta_2p_{n-1}
\end{equation}
is an example of a non-Markovian distribution with a two-step recurrence.
\end{enumerate}

\newpage
\section{Discrete Probability Distributions in Statistical Mechanics} \label{statmech}
Statistical mechanics concerns itself with a collection (an ensemble) of a very large number of systems, each characterized by a set of state variables \(X=\{X_1, X_2,\ldots\}\). Extensive and intensive variables appear in conjugate pairs, such as temperature and entropy, pressure and volume, or chemical potential and particle number. Within each pair, one variable assumes a fixed (sharp) value, denoted \(X_f\), whereas its conjugate partner, \(X_r\), exhibits fluctuations around its mean value \(\bar{X}_r := \mathbb{E} X_r\), which is identified with the ensemble average. This identification is commonly justified by the ergodic hypothesis \cite{Hill}, which postulates that, in equilibrium, the time average of a state variable in a single system coincides with its ensemble average. In what follows, whenever the conjugate variable \(X_f\) is taken to be fixed, we shall frequently omit the bar in \(\bar{X}_r\), implicitly referring to its ensemble-averaged value.

As will become evident below, in direct analogy with the cumulants of the power series distribution (PSD) in Eq.~(\ref{reccum}), all moments \(\mathbb{E}(X_r)^j\) of the fluctuating variable \(X_r\), forming a conjugate pair \((X_f, X_r)\), can be obtained from derivatives of a suitably defined function -- the partition function --  with respect to the fixed variable \(X_f\).

In the subsequent sections, we apply probabilistic concepts and results developed in the first part of this work to problems in statistical mechanics, with particular emphasis on small systems. In contrast to nanothermodynamics, which is typically formulated in terms of the grand canonical isothermal–isobaric ensemble (i.e., at fixed chemical potential, temperature, and pressure; see, for example, \cite{chamberlin2020}, we assume here that the system possesses a well-defined volume \(V\). Furthermore, our definition of the thermodynamic limit replaces the conventional condition \(V \to \infty\) with a limit more suitable for small systems, namely \(\mathcal{N}(V) \to \infty\), where \(\mathcal{N}(V)\) denotes the maximal number of (point-like) particles that can occupy the volume \(V\). This generalized notion of the thermodynamic limit allows us to investigate systems with unconventional properties, such as those exhibiting an infinite mean particle number \(\bar{N}\), which are not usually addressed in standard treatments.

Our objective is to exploit the infinite divisibility and scalability of particle-number distributions \(\mathbb{P}(N = n)\), which naturally emerge in ensembles with interaction range comparable to particle sizes. For systems with long-range interactions, more appropriately described within the framework of Tsallis thermodynamics (see Section~\ref{TT}), only scalability is anticipated to hold. Consequently, our analysis concentrates predominantly on Markovian probability mass functions that satisfy the one-step recurrence relation
\[
(n+1)p_{n+1} = g(n)p_n
\]
[cf. Eq.~(\ref{eq: gnkl})], which also appear as stationary solutions of birth–death (B–D) processes (see Section~\ref{kinmodels}).

Throughout this work, we employ natural units, setting \(\hbar = c = k_B = 1\).

\subsection{Particle-number distributions in grand canonical ensemble}

Consider an open system with fixed inverse temperature $\beta:=1/T$, volume $V$ and fugacity  $\xi=e^{\beta\mu}$, where $\mu$ is the chemical potential. In statistical mechanics (see e.g. \cite{Hill, Lanlif, BlanTh}), a system with fixed $V$, $\beta$, and $\mu$ is completely described by the Grand Canonical Ensemble (GCE). Its partition function $\mathcal{Z}(\beta, V,\xi)$ is a power series in  $\xi$ with coefficients given by the Canonical Ensemble (CE) partition function $\mathcal{Q}(\beta, V, n)$:
\begin{equation} \label{gcz}
 \mathcal{Z}(\beta,V,\xi):= \sum_{n\geq 0}^{\mathcal{N}(V)}\mathcal{Q}(\beta, V,n) \xi^n;\;\;\; \xi:=e^{\beta\mu} \,.
\end{equation}
The upper limit $\mathcal{N}(V)$ in the sum (\ref{gcz}) is the maximum number of particles in a volume $V$. For $\mu\leq \mu_0<0$, the function $\mathcal{Z}(\beta, V,\xi)$ is analytic in closed disk $|\xi|\leq\xi_0 =\exp(\beta\mu_0)<1$. Let us note that, except for  quantum gases obeying Fermi–Dirac statistics, for which the chemical potential $\mu$ may take either positive or negative values, the inequality $\xi < 1$ is always fulfilled in all other cases, including quantum gases obeying Bose–Einstein statistics as well as classical gases \cite{Lanlif}. Consequently, in what follows, we confine our analysis to this parameter regime.

\begin{Remark}
 The special case $\mu = 0$, and consequently $\xi = 1$, corresponds to a regime in which the system can exchange particles with its surroundings without any energetic cost. Imposing $\mu = 0$ removes the particle number $N$ as an independent thermodynamic variable; the thermodynamic potentials are then fully characterized by the inverse temperature $\beta$ and the volume $V$:
\[
\mathcal{Z}(\beta, V, 1) = \mathcal{Q}(\beta, V).
\]
Representative physical realizations include photons, phonons, and other quasiparticles that can be created or annihilated as the temperature or volume varies. For an accessible and pedagogically oriented introduction to this topic, see \cite{Blaschke2026}.
\end{Remark}

In the language of discrete probability distributions, the partition function (\ref{gcz}) with $\xi<1$ defines a {\it power series distribution} (PSD) of the r.v. $N$ -- the number of particles at fixed  $\beta, V$ and $\xi$. Its p.m.f. and p.g.f.{} $H(w;\beta,V,\xi)$ read (cf. (\ref{PSDdef})):
\begin{equation} \label{gcpart} 
 r_n(\beta, V, \xi)\!=\!\frac{\xi^n \mathcal{Q}(\beta,V,n)}{\mathcal{Z}(\beta,V,\xi)},\;\;\;H(w;\beta, V,\xi)= 
     \frac{\mathcal{Z}(\beta, V, w\xi)}{\mathcal{Z}(\beta, V, \xi)}\;.
\end{equation}
The function $H(w;\beta,V,\xi)$ is an analytic inside the disk $|w|\leq 1/\xi$, but the moments $\E N^k, k\in \N$ exist only for $\xi<1$ (see Section \ref{psdsub}). Using the formulas (\ref{reccum}) and (\ref{PSDcomb}) for cumulants $\kappa_{\ell}$ and combinants $\lambda_{\ell}$ of the PSD, we have:
\begin{equation}\label{gccomb}
\kappa_{\ell}=\xi^{\ell}\frac{\partial^{\ell}\log\mathcal{Z}(\beta, V, \xi)}{\partial\xi^{\ell}}=\frac{\lambda_{\ell}\xi^{\ell}}{\ell!}\;.
\end{equation}
In particular, the average number of particles occupying the volume $V$ is $\bar{N}=\kappa_1$. We remark in passing that both the p.g.f. $H(w;\beta,V,\xi)$  and the cumulants $\kappa_{\ell},\;\ell\in \N$  are invariant under the rescaling of partition function $\mathcal{Z}(\beta, V, \xi)\to \delta \cdot \mathcal{Z}(\beta, V, \xi)$, cf. Eqs. (\ref{gcpart}) and (\ref{gccomb}). 

\subsubsection{Lee-Yang zeroes}
At fixed $\beta$, the function $\mathcal{Z}(\beta, V,\xi)$ (\ref{gcpart}) is a polynomial of degree $\mathcal{N}(V)$ with real positive coefficients and zeroes located in the complex $\xi$-plane off the positive real axis. Assuming that particles can be treated to good approximation as a point-like objects, we define the {\it small system thermodynamic limit} as\footnote{For small systems of finite volume, this definition of thermodynamic limit is more appropriate \cite{Hagedorn:1985js} than the standard limit of infinite volume $\lim_{V\to\infty} \log \mathcal{Z}(\beta, V,\xi)/V
$.}:
\begin{equation}\label{thlim}
\lim_{\mathcal{N}(V)\to\infty}\mathcal{Z}(\beta, V,\xi)\;.
\end{equation}
If, in this limit, the sequence of zeroes approaches the positive real axis at some point $\xi_0$, then a {\it phase transition} occurs  \cite{Yang:1952be}. The curve $\xi_0(\beta)>0$ in the $\beta-\xi$ plane separates two phases. 

To proceed further, we need to introduce the {\it single-particle partition function}:  
\begin{equation}\label{sce}
\mathcal{Q}(\beta, V, n=1)=\mathcal{Q}_1(\beta, V):=\sum_{j\in \Omega} e^{-\beta E_j}\;.  \end{equation}
 The sum in (\ref{sce}) runs over the set $\Omega$ of all available micro-states $j$ in which a single particle with energy $E_j$  can be found at temperature $\beta$. The Boltzmann factor $e^{-\beta E_j}$ arises from maximizing the BGS entropy $S^{BGS}:=-\sum_j p(E_j)\log p(E_j)$ of the energy distribution $p(E_j)$ at fixed average energy $\bar{E_j}$, with $\beta$  playing a role of the Lagrange multiplier. 

 \begin{Example}
 Let us consider the classical ideal relativistic gas of particles with mass $m$, momentum ${\bf p}$, and energy $E=\sqrt{p^2+m^2}$, and the single-particle partition function:
\begin{equation}\label{irg}
\mathcal{Q}_1(\beta,V)= V\int \frac{d^3 p}{(2\pi)^3}e^{-\beta E}=V\frac{m^2}{2 \pi^2\beta }K_2\left(m \beta \right)\;,
\end{equation}
where $K_2(x)$ is the modified Bessel function. In the thermodynamic limit (\ref{thlim}), the GCE partition function has no singularities on the positive real $\xi$-axis. 
\begin{equation}\label{GCEidgas}
\mathcal{Z}(\beta,V,\xi)=\sum_{n=0}^{\infty} \frac{(\mathcal{Q}_1)^n}{n!}\xi^n=e^{\xi V z}\;.
\end{equation}
From (\ref{gcpart}) follows that the p.g.f. $e^{\xi V z(w-1)}$ represents the Poisson distribution.
The factorial in the denominator of (\ref{GCEidgas}) reflects the fact that there are $n!$ ways to count $n$ independent, {\it indistinguishable} classical particles. Dropping this assumption, let's replace in (\ref{GCEidgas}) $n! \to (n+1)^{\alpha}$.  This leads to a partition function with a singularity at $\xi_0=1/\mathcal{Q}_1(\beta)$ as $\mathcal{N}(V)\to\infty$: 
\[
\lim_{\mathcal{N}(V)\to\infty} \mathcal{Z}(\beta, \xi,V)=\frac{\text{Li}_{\alpha }(\xi \mathcal{Q}_1(\beta,V))}{\xi \mathcal{Q}_1(\beta,V)}\;;\;\; \text{Li}_{\alpha }(z)=\sum_{k=0}^{\infty}z^k k^{\alpha}
\]
For for $\alpha\!=\!1$ (\ref{gcpart}) yields the p.g.f. of the shifted logarithmic distribution, for  $\alpha=0$ we obtain the geometric p.g.f.
An analogous structure arises in the analysis of the one-dimensional quantum harmonic oscillator, which may be regarded as an illustrative toy model of the Bose gas\footnote{For a comprehensive treatment of the three-dimensional Bose gas, see \cite{Lanlif}}. In this setting, 
\begin{equation}\label{hosc}
\mathcal{Q}_1(\beta)=\sum_{j=0}^{\infty} e^{-\beta E_j} = \sum_{j=0}^{\infty} e^{-\beta (j+\frac{1}{2})\omega_0}
 =\frac{e^{\beta\omega_0/2}}{e^{\beta\omega_0}-1}
 = \frac{1}{2\sinh(\beta\omega_0/2))}\;.     
 \end{equation}
In the thermodynamic limit, the GCE partition function takes the form:
\begin{equation}\label{Zqharm} 
\mathcal{Z}(\beta,\xi)=\sum_{n\geq0}^{\infty}(\mathcal{Q}_1(\beta))^n\xi^n=
(1-\mathcal{Q}_1(\beta)\xi)^{-1}=1-\frac{\xi }{\xi -2 \sinh \left(\frac{\beta \omega_0}{2}\right)}\;,
\end{equation}
with the singularity at 
$\xi_0(\beta)=e^{\beta \mu_0}=1/\mathcal{Q}_1(\beta)$. 
Solving this equation for $\mu_0$ we obtain
\[
\mu_0
=\frac{\omega_0}{2} +\frac{1}{\beta} \log(1 - e^{-\beta \omega_0})\;.
\] 
Note that as $\beta\to \infty$ the chemical potential tends to the value of the oscillator's zero-point energy: $\mu_0\to E_0=\omega/2$. 
The latter is inaccessible due to the restriction $\mu\leq 0$. The largest attainable value where the singularity occurs  is thus $\mu_0=0$, i.e., $\xi_0=1$, with the critical temperature
$T_c=\beta_c^{-1}=\omega_0/(2 \sinh ^{-1}(1/2))\approx 1.04\omega_0$. In this case, an increase of the chemical potential from $\mu<0$ to $\mu=0$ is enforced by decreasing the gas temperature $T$ at constant $\rho$ \cite{Lanlif}.  At $T_c$, the mean number of particles $N$ and all higher moments of particle-number distribution diverge -- the gas transforms into the lowest quantum state -- the Bose-Einstein condensate.  
 
For completeness, consider an interesting solvable model: an ensemble of an infinite number of harmonic oscillators in the infinite heat bath with fixed $\beta$ and $\xi$. The average number of particles $N_j$ emitted by the $j^{\rm th}$ oscillator is:
\[
N_j =\frac{N}{(2j-1)^{2}}\;,\;\; j=1,2,\ldots\;;\;\;\;N=
\xi\frac{\partial\log\mathcal{Z}(\beta, \xi)}{\partial\xi}=
\frac{\xi }{2 \sinh \left(\frac{\beta \omega_0}{2}\right)-\xi}
\]
The distribution of the total number of particles emitted by all oscillators has the p.g.f. (see  Appendix \ref{Exscal}):
\begin{equation}\label{inf_many_osc}
 H(w) = \prod_{j=1}^{\infty} \frac{1}{1+N_j(1-w)}
=\text{sech}\left(\frac{\pi}{2}\sqrt{N(1-w)}\right)\;,   
\end{equation}
The corresponding particle-number distribution is infinitely divisible, scalable, and has the mean  $N\pi^2/8\approx 1.23N$.

Assume now that each emitted particle represents a cluster of many distributed according to the Sibuya distribution. In this case, each harmonic oscillator emits the particles according to the compound geometric Sibuya distribution with the p.g.f. $[1+\lambda(1-w)^{\gamma}]$, (see (\ref{DMLpgf})). The total number of quanta emitted from all oscillators has the p.g.f. (see  Appendix \ref{Exscal}):
\[
H(w) = \text{sech}\left(\frac{\pi}{2}(N(1-w))^{\gamma/2}\right)\;.
\]
The distribution is again infinitely divisible and scalable, but has an infinite mean. In the limit $\gamma\to1^{-}$ we obtain the p.g.f. (\ref{inf_many_osc}) with a finite mean.
 \end{Example}

\subsubsection{The limiting temperature}
Another type of singularity arises when representing the grand-canonical ensemble (GCE) partition function as an integral over the density of states $\sigma(E)$:
\begin{equation}\label{eq:sigma}
    \mathcal{Z}(\beta, V,\xi)= \int_0^{\infty}\sigma(E,V, \xi)\, e^{-\beta E}\, dE \;.
\end{equation}
If 
\[
\lim_{E \to \infty}\frac{S(E)}{E} = \lim_{E \to \infty}\frac{\log\sigma(E)}{E} = 0 \,,
\]
where $S(E) := \log \sigma(E)$ denotes the microcanonical entropy, the integral in Eq.~\eqref{eq:sigma} defines an analytic function in the right half-plane $\mathrm{Re}(\beta) > 0$. If, in addition, the energy spectrum is bounded from above (i.e., $\sigma(E, V, \xi) = 0$ for $E > E_{\max}$,\footnote{For fixed $V$ this corresponds to a maximal attainable energy density $\epsilon_{\max} = E_{\max}/V$.} ), then the partition function $\mathcal{Z}(\beta, V, \xi)$ is an entire function of $\beta$.

In contrast, if the density of states exhibits the asymptotic behavior
\[
\sigma(E, V, \xi) = f(E, V, \xi)\, e^{\beta_0 E}, \qquad \beta_0 > 0,
\]
with $f(E, V, \xi)$ growing at most polynomially in $E$, the partition function $\mathcal{Z}(\beta, V, \xi)$ develops a singularity at $\beta = \beta_0$. The parameter $\beta_0$ then defines the {\it limiting temperature} of the corresponding phase \cite{Hagedorn:1985js, Rumer:1960}. In particular, for 
\[
f(E, V,\xi)\sim (E+E_0)^{-a},
\]
with constants $E_0>0$ and $a>0$, the GCE partition function $\mathcal{Z}(\beta, V, \xi)$ is well-defined only for $\beta > \beta_0$. 

Such an exponential growth of the density of states is realized in hadronic systems, where $\sigma(E)$ exhibits a singularity at the {\it Hagedorn temperature} $T_H = 1/\beta_0 \simeq 160$--170 MeV \cite{Hagedorn:1985js}. At this temperature, hadronic matter becomes thermodynamically unstable: hadrons dissolve into their elementary constituents, initiating a phase transition \cite{Cabibbo:1975ig} to a new state of strongly interacting matter, the quark–gluon plasma \cite{Rafelski:2015xej, Pasechnik:2016wkt}.
\subsubsection{Other Thermodynamic Variables}\label{OtTh}
From  (\ref{eq:sigma}) follows the relation between the GCE partition function and the average energy:
\begin{equation}
    \bar{E}  = -\left [\frac{\partial \log \mathcal{Z}(\beta, V,\xi)}{\partial \beta} \right]_{V,\mu} \;.
\end{equation} 
The latter allows us to introduce the {\it Helmholtz free energy} $F$ as the Legendre transform between dimensionless quantities $\beta  F$ and $\beta  E$: 
\begin{equation}\label{Free}
\beta F=\frac{d S}{d E} E -S =\beta E -S= -\log \mathcal{Z}(\beta, V, \xi),
\end{equation}
as well as the pressure $p$ of the fluid  \footnote{For small systems, we need not define the pressure as the limit $\beta p:=\lim_{V\to \infty}\log \mathcal{Z}(\beta, V, \xi)/V $. }
\begin{equation}\label{pressure}
-\beta F = \beta p V = \log \mathcal{Z}(\beta, V,\xi)\;.
\end{equation}
Let us note in passing that the condition $p\geq0$ implies $\mathcal{Z}(\beta, V,\xi)\geq 1$.

The extensive variable $E$ is homogeneous of degree 1 in $S, V, N$ \footnote{Let us recall that in the GCE the state variables $E, S, N$ represent the ensemble averages.}:
$E=V \varepsilon (S/V, N/V)$, where $ E/V = 
\varepsilon$ is the energy density.
Consequently, by Euler's theorem, we  obtain: 
\begin{equation}\label{flaw}
dE=TdS-pdV+\mu dN    
\end{equation}
which represents the {\it first law of thermodynamics}. Subtracting it from the differential $d(E -T S+ pV - \mu N)$  yields the equations: 
\begin{equation}\label{SVN}
-SdT+Vdp=Nd\mu\;,\;\;\;dp=\frac{N}{V}d\mu + \frac{S}{V}dT=\rho d\mu+sdT\;.    
\end{equation}
The latter allow us to express particle-number density $\rho$ and entropy density $s$ as derivatives of the pressure $p$:
\begin{equation}\label{SfromP}
 \rho=\left( \frac{\partial p}{\partial \mu}\right)_{V,T} \;\; \text{and}\;\; s=\left( \frac{\partial p}{\partial T}\right)_{V,\mu}=-\beta^2\left(\frac{\partial p}{\partial \beta}\right)_{V,\mu}\;.
\end{equation}

Let us now discuss the relation between energy density $\varepsilon$ and $\rho$. By dimensional arguments, the ratio 
\begin{equation}\label{phix}
 \beta\frac{E}{N}=\beta\frac{\varepsilon}{\rho}:=\phi(x) 
\end{equation}
is a function of the dimensionless argument $x=m\beta$, where $m$ is the rest mass of a single particle.  
Let us consider two examples. 
\begin{Example}\label{phi_ex}
 \begin{enumerate}[label=(\roman*), leftmargin=*]
 \item[]
 \item $\phi(x)=\phi_0$ yields  $E/N= \phi_0 T$ which looks similar to classical ideal gas. According to the equipartition theorem\cite{Lanlif}, in non-relativistic systems, each independent degree of freedom contributes an average energy per particle $E/N$ of $T/2$. Thus, for a system with $n_{DoF}$ number of degrees of freedom $\phi_0=n_{DoF}/2$. For ultra-relativistic systems, this value is doubled and $\phi_0=n_{DoF}$.
 
 \item For {\it collision-dominated gas of particles} (or plasma)
 \footnote{A state of matter where the behavior of the gas is governed primarily by frequent particle collisions rather than external forces or boundaries.}, where the energy-momentum conservation takes place at a microscopic level, relativistic kinetic theory gives \cite{Csernai, DeGroot}:  
\begin{equation} \label{eq: collgas}
    \phi(x) =x\left(\frac{K_1 (x)}{K_2(x)} + \frac{3}{x} \right)  \;.
\end{equation}
The high- and low-$x$ expansion of the function $\phi(x)$ yield:

\begin{equation} \label {eq: collgas1}
\varepsilon   \underset{m\beta \ll 1}{\approx}  \left(\frac{3}{\beta }+\frac{\beta  m^2}{2}\right)\rho \rightarrow 3\frac{\rho}{\beta}\;\;\;\; \text{and} \;\;\; 
\varepsilon   \underset{m\beta \gg 1}{\approx} \left(m+\frac{3}{2 \beta}\right)\rho 
\rightarrow m \rho\;.
\end{equation}
These approximations allow us to express the speed of sound squared $c_s^2$ as a function of particle-number density $\rho$:
\begin{equation}\label{c_s}
c_s^2=\frac{dp}{d\varepsilon}=\frac{d\rho}{d\varepsilon} \frac{dp}{d\rho}=\frac{\beta}{\phi(m\beta)}\frac{dp}{d\rho}\;;\;\;\:
c_s^2\underset{m\beta \ll 1}{\approx}\frac{\beta}{3}\frac{dp}{d\rho}\;,\;\; c_s^2 \underset{m\beta \gg 1}{\approx}\frac{2\beta}{m+3}\frac{dp}{d\rho} \;.
\end{equation}
For $m\beta\ll1$ (ultra-relativistic gas), the inequality $c^2_s\leq 1$ leads to $dp/d\rho\leq 3T$.
Using the above variables, we can obtain the specific heat at fixed $V$:
\[
C_V=\frac{d\varepsilon}{dT}=\frac{d \varepsilon}{dp} \frac{dp}{dT} =\frac{s}{c^2_s}
\]
\end{enumerate}
\end{Example}
\subsection{Infinite divisibility and Scalability} \label{infdscal}

General theoretical considerations within the grand canonical ensemble, when applied to systems characterized by short-range interactions, i.e., with a range comparable to particle sizes, imply that the particle-number distribution must satisfy the properties of infinite divisibility and scalability. By contrast, for systems dominated by long-range interactions and formulated within the Tsallis non-extensive thermodynamic framework (see Section \ref{TT}), only the scalability condition remains valid.

Formally, infinite divisibility requires that an arbitrary partition of the total volume $V$ into $j$ equal sub-volumes produces, in each sub-volume, a particle-number distribution that belongs to the same family as the distribution defined on the original volume. Scalability, in turn, denotes invariance of the functional form of the distribution under an arbitrary rescaling of the particle number by a factor $a>0$, thereby encompassing both increases and decreases relative to the original particle count.
\subsubsection{Infinite divisibility of particle-number distribution}
Recall that according to Proposition \ref{infdthin} for every infinitely divisible p.g.f. of the form:
\[
H(w;\lambda)=\exp\bigl\{\lambda\bigl(G(w)-1\bigr)\bigr\},
\]
there exists a corresponding PSD p.g.f.
\[
R(w;\theta,\lambda)=\frac{H(\theta w;\lambda)}{H(\theta;\lambda)},
\]
which represents a compound Poisson distribution with parameter $\nu$, where $\nu$ is a function of $\theta$ and $\lambda$.
\begin{equation}\label{PSDfromCPD1}
 R(w;\theta,\lambda)
 =\frac{e^{\lambda (G(\theta  w)-1)}}{e^{\lambda (G(\theta )-1)}}=e^{\nu ( G_1(w)-1)};\;\;\;
 \nu:=\lambda G(\theta),\;\; G_1(w):=\frac{G(\theta w)}{G(\theta)}\;.
\end{equation}
After substitutions  $\theta\to\xi$, $G(\theta w) \to \log \mathcal{Z}(\xi w)$ (see (\ref{reccum_m}))  and $H(\theta w;\lambda) \to \mathcal{Z}(\cdot,\xi w)$ \footnote{The placeholder $\cdot$ stands for the suppressed arguments  $\beta$ and $V$.}, the logarithm becomes:
\begin{equation}
\log\left(\frac{\mathcal{Z}(\cdot,\xi w)}{\mathcal{Z}(\cdot ,\xi)}\right)=\nu(G_1(w;\xi)-1)
\end{equation}
The condition  $G_1(0,\xi)=0$ implies $\mathcal{Z}(\cdot,0)=1$. For partition functions not satisfying this condition, we apply the transformation $\log \mathcal{Z}(\cdot,\xi)\to \log \mathcal{Z}(\cdot,\xi)-\log \mathcal{Z}(\cdot, 0)$, corresponding to the additive shift  $F(\xi)\to F(\xi)-F(0)$ of free energy (\ref{Free}) at fixed $\beta$, which reflects the familiar non-uniqueness of the energy minimum.

By Theorem \ref{infd2PSD}, if an equilibrium thermodynamic system with the GCE partition function $\mathcal{Z}(\cdot,\xi)$  has an infinitely divisible p.g.f. with fugacity $\xi$ depending only on $\nu=- \log \mathcal{Z}(\cdot,0) /\log \mathcal{Z}(\cdot,\xi)$, then its p.g.f. has only one nonzero combinant $\lambda_j>0$. Such a system therefore describes a single-component ideal gas of $j$ particle clusters with the p.g.f. $e^{\lambda_j(w^j -1)}$.  

On the other hand, by Theorem \ref{PSDisCPD}, the analyticity of $\log \mathcal{Z}(\cdot, \xi w)$ in the disk $|\xi w|\leq 1$ with non-negative power series coefficients at $w=0$ implies an infinitely divisible particle-number distribution with parameter $\nu$ depending on $\xi$. 

We can now frame these observations as a single proposition that establishes the quasiparticle picture of an interacting gas.
\begin{Proposition}\label{cor:mixt}
A gas of interacting particles with $\log \mathcal{Z}(\cdot,\xi w)$ analytic in $|\xi w|\leq 1$ and non-negative power series coefficients at $w=0$ represents a mixture of different species of quasiparticles, each consisting of $j$ particles, in mutual equilibrium. Each component satisfies the {\it Equation of State} (EoS) of a perfect gas.   
\end{Proposition} 

\begin{proof}
The PSD character of the particle-number distribution (Eq.\ref{gcpart}) and the analyticity of $\log \mathcal{Z}(\xi w)$ with non-negative coefficients imply the infinite divisibility by Theorem \ref{PSDisCPD}. Let $G(w)=\sum_{j>0}g_j w^j$ be the p.g.f. of compound Poisson secondary distribution. Then
\[
 H(w;\lambda, \xi)=
 \frac{\mathcal{Z}(\cdot,\xi w)}{\mathcal{Z}(\cdot,\xi)}
 = \exp\left[\lambda \sum_{j>0}g_j \xi^j (w^j-1)\right] =\prod_{j>0}\exp\left[\lambda_j\xi^j(w^j-1)\right]\;,
\]
where $\lambda_j=\lambda g_j$ is the $j^{\rm th}$ combinant. Substituting  $\mathcal{Z}(\cdot,\xi)$ into (\ref{pressure}) and writing the pressure as $p=\sum_{j>0} p_j$, where $p_j$ is the partial pressure due to clusters consisting of exactly $j$ particles, gives
\begin{equation} \label{eq: clust}
\log \mathcal{Z}(\beta, V,\xi) =   \beta pV =  
\sum_{j>0}\beta p_jV  =\sum_{j>0} \lambda_j \xi^j =\sum_{j>0}N_j\;.
\end{equation}
 Thus, each component satisfies the ideal gas EoS: $\beta p_j V = N_j$.
\end{proof}

This result can also be obtained from a physical argument \cite{Hill}: the average number of $j$-particle clusters $N_j$ remains fixed  when their chemical potential $\mu_j$  satisfies the equilibrium condition  $\mu_j=j\mu$, i.e. $\xi^j=\xi_j, ~j\in\N$:
 \begin{eqnarray}
  \beta p = \sum_{j\geq1}b_j(\beta)\xi^j =  \sum_{j\geq1}b_j(\beta)\xi_j = \sum_{j\geq1}\beta p_j\, \nonumber \\ 
 \rho =  \sum_{j\geq1} j b_j(\beta) \xi^j =  \sum_{j\geq1} j b_j(\beta) \xi_j = \sum_{j\geq1}j \rho_j 
 \label{multideal}
 \end{eqnarray}

\subsubsection{Scalability of particle-number distribution.}
Consider a representation of a nonideal gas as a statistical mixture of ideal gases, each component characterized by the equation of state (EoS) $\beta p_j = \lambda_j / V = \rho_j$. Let $r_j$ and $R(w)$ denote, respectively, the p.m.f. and the p.g.f. of the mixing distribution. Then the p.m.f. $p_n$ and p.g.f. $H(w)$ of the resulting mixture distribution are given by:
\[
p_n \;=\; \sum_{j} p_{n\mid j}\cdot r_j ,
\quad\text{and}\quad
H(w) \;=\; R\bigl(G(w)\bigr),
\]
where $p_{n\mid j}$ is the conditional p.m.f. of the number of particles given the $j$-th ideal-gas component, and $G(w)$ is the corresponding conditional p.g.f.\begin{equation}\label{mix}
p_n=\sum_j e^{-\lambda_j} \frac{\lambda_j^n}{n!}\cdot r_j \;,\;\; H(w)=\sum_n p_n w^n=\sum_j  e^{\lambda_j(w-1)} \cdot r_j \;.
\end{equation}
For $\lambda_j=2j$ \footnote{The choice of even values $2j$ prevents the appearance of terms with a negative sign.}, the p.g.f. $H(w)$ can be expressed as a power mixture\cite{StvH2004}  representing a compound distribution with  mean 
$\lambda =  2G^{(1)}(1)$:

\begin{equation}\label{mixpoipgf}
 H(w)=\sum_j e^{-2j(1-w)}\cdot r_j = \sum_j (P(w,2))^j \cdot r_j=R(P(w,2)) \;,\;\; P(w;1)= e^{(w-1)} \;.
\end{equation}
The scale transform $w\to 1-a(1-w)$ of $H(w)$ sends $R(P(w,2)) \to R(P(w,2a))$, which is a valid p.g.f. for both $0<a\leq 1$, (thinning), and $a>1$ (zooming out). Hence $H(w)$ is scalable.

In the continuous case, a p.d.f. $f(\lambda)$ replaces the mixing p.m.f. $r_j$ and  Eq.\ref{mixpoipgf} becomes:
\begin{equation}\label{fluctlamb}
p_n= \int_0^{\infty} e^{-\lambda}\frac{\lambda^n}{n!}f(\lambda) d\lambda\;,\;\; H(w)=\int_0^{\infty}e^{\lambda(w-1)} f(\lambda) d\lambda\;.    
\end{equation}
The p.g.f. $H(w)$ of the mixture is scalable (see Eq.\ref{eq: lapdens}), and its scaled factorial moments $F_j=H^{(j)}(1)/(H^{(1)}(1))^j$ are invariant under $H(w)\to H(1-a(w-1)),\;\forall a\in \R$. 

The integral representation (\ref{fluctlamb}) is valid for a broad class of infinitely divisible p.g.f.s; see, for example, (\ref{eq: lapdens}). In many instances, this representation is a direct consequence of the fact that a scalable p.g.f. is infinitely divisible if and only if the probability density function (p.d.f.) $f(\lambda)$ of the associated mixing distribution is itself infinitely divisible \cite{StvH2004}. In the complementary case -- such as when $f(\lambda)$ is specified by the Tsallis p.d.f. (\ref{Tcontex}) -- a corresponding Laplace-type integral representation can still be constructed; see Eqs. (\ref{Tw}) for a detailed exposition.

The superposition of Poisson distributions characterized by distinct mean values induces fluctuations in the mean particle number, $\langle n \rangle = \lambda$. For a fixed system volume $V$, these fluctuations are equivalently manifested as variations in the particle-number density $\rho$, and, consequently, in the energy density $\varepsilon$.
\begin{Example}
For the NBD, $F_j=(k)_j/k^j$, so all multi-particle correlations are determined by the two-particle correlation function $\langle n(n-1)\rangle/ \langle n \rangle^2 -1=F_2-1 =1/k$. As $k\to \infty$, the correlations weaken, approaching the Poissonian limit $F_2=1$. For an ideal gas of quasiparticles each consisting of $n$ particles, with GCE partition function $\mathcal{Z}(\xi)=\exp(-\lambda \xi^{n})$:
\begin{equation}\label{facmom_n-tup}
 F_2=\frac{(n-1) \xi ^{-n}}{n \lambda}+1\geq \frac{2n-1}{n \lambda}\;,\;\; 
F_3\geq \frac{n (n \lambda  (\lambda +3)+n-3 (\lambda +1))+2}{n^2 \lambda ^2}\;.   
\end{equation}
The Poissonian limit $F_j=1\;, \forall j\in\N$  corresponds to $n=1$.

\end{Example}

The p.m.f. $r_n(\xi)$ of a scalable, infinitely divisible particle-number distribution is Markovian (see  (\ref{eq: gnkl})). Since the same is true for the p.m.f. arising as the stationary solution of B-D processes (\ref{eq:prodgn}), we have in both cases (see (\ref{gcpart})) :
\[
g_r(n)=\frac{(n+1) r_{n+1}}{r_n}=\xi\frac{(n+1)\mathcal{Q}(\beta,n+1)}{\mathcal{Q}(\beta,n)}=\xi g_p(n)\;.
\]
Thus, up to the normalization constant, the CE partition function coincides with the p.m.f. $p_n$ of a Markovian distribution with transition function $g_p(n)=(n+1)p_{n+1}/p_n$. In this case, the CE partition function $\mathcal{Q}(\beta, n)$ is expressible as a product of the single-particle partition functions  $\mathcal{Q}_1(\beta)$ of Eq. \ref{sce} and $g_p(n)$:
\begin{equation}\label{gcpartgn}
\mathcal{Q}(\beta, n) = \mathcal{Q}_1(\beta)\left(\frac{\delta_{n0}}{g_p(0)}+(1-\delta_{n0})\prod_{j=1}^{n-1}\frac{g_p(j)}{j+1}\right)\;.
\end{equation}
\subsection{Virial Equation of State} \label{EOS}
Knowledge of the logarithm of the GCE partition function $\log \mathcal{Z}(\beta, V,\xi)=-\beta F$  and its first derivative -- the particle-number density $\rho(\xi)$
\begin{equation} \label{virexp}
    \beta  pV =\log \mathcal{Z}(\beta, V,\xi),\;\;\;\rho(\xi)=\frac{N}{V}
    =\frac{\xi}{V}\frac{\partial\log \mathcal{Z}(\beta, V,\xi)}{\partial\xi}\;.
\end{equation}
allows us to obtain the EoS of the fluid, $\beta p V=f(\rho)$.

The unknown function \(f(\rho)\) can, in principle, be determined by inverting the relation \(\rho(\xi)\). This inversion can always be performed at least perturbatively, for example by employing the method of series reversion (see, e.g., \cite{AbSteg1972}), which yields the corresponding virial expansion.\begin{equation}\label{virial}
\beta p V = f(\rho)= V\sum_{j\geq 0}^J B_{j}(\beta) \rho^{j}= \sum_{j\geq 0}^{J}\lambda_j(\beta)\xi^j,\;\;\;V\rho(\xi)=\sum_{j\geq 0}^{J}j\lambda_j(\beta)\xi^j \;.
\end{equation}
Note that each power $\rho^j$ contributes one combinant $\lambda_j$. The upper cutoff $J$ in the sum reflects our incomplete knowledge of the GCE partition function $\mathcal{Z}(\beta, V, \xi)$\footnote{Let us note that  (\ref{virial}) with $J<\infty$ and $\lambda_j>0,\;\forall j\in \N$ defines the EoS even for a secondary distribution with a non-existent mean.}. 
Alternatively, $J<\infty$ corresponds to a particle-number distribution with a finite support, or can be interpreted in terms of the characteristic inter-particle distance $r_0$, with $J\approx V/r^3_0$, ensuring that the coefficients $B_j( \beta)$ are independent of $V$ \cite{Hill}.  

To illustrate the last point, consider the EoS\footnote{Case with $k=1$ describes a gas with maximum entropy at fixed mean particle-number $V\rho$.}:
\begin{equation}
 \beta pV =\log \mathcal{Z}(\beta, V,\xi)= \log \left(1\!+\!\frac{V\rho}{k}\right)^k =
V\sum_{j>0}B_j\rho^j;\;\;\; B_j= \frac{1}{j} \left(-\frac{V}{k}\right)^{j-1}
\end{equation}
corresponding to the NBD p.g.f.(\ref{NBD}):
\[
H(w)=\frac{\mathcal{Z}(\xi w)}{\mathcal{Z}(\xi)}\;,\;\; \mathcal{Z}(\xi) = (1-\xi)^{-k}= \left(1\!+\!\frac{V\rho}{k}\right)^k\;,\;\; \xi=\left(1\!+\!\frac{V\rho}{k}\right)^{-1}\;.
\]
The requirement that the virial coefficients $B_j(\beta)$ and the fugacity $\xi$ be independent of the volume $V$ is fulfilled only in the regime where $k \sim V$. In this regime, the conventional thermodynamic limit $V \to \infty$ reproduces the ideal-gas equation of state and leads to a Poissonian particle-number distribution. By contrast, for finite (small) systems, the thermodynamic limit $V \to \infty$ is not applicable, and one must retain $V$-dependent virial coefficients $B_j(V)$. Under these conditions, the two-particle correlation function, characterized by $F_2 - 1 = 1/k$, may exhibit a more general functional dependence on the volume. In particular, for short-range interactions, it can become effectively independent of $V$.

\subsubsection{The vacuum pressure} \label{vacpres}
From the normalization condition of the probability generating function, namely $\log H(1) = 0$, it follows that $\lambda_0 = -\sum_{j>0}^{J} \lambda_j$. Accordingly, the zeroth coefficient $B_0(\beta) < 0$ in the virial expansion can be interpreted as generating a confining pressure, whereas $B_0(\beta) > 0$ corresponds to a vacuum pressure.  


In the standard thermodynamic limit, $V \to \infty$, this contribution becomes negligible and effectively vanishes. For finite systems of point-like particles, however, the absolute term
\[
\frac{1}{V} \lim_{\rho \to 0^+} \log \mathcal{Z}(\beta, V, \xi) 
= \beta p_0 
= B_0(\beta) \neq 0
\]
cannot be neglected, as it yields a nontrivial vacuum (or reference) pressure even in the dilute limit.
Since $\kappa_1 = \xi \partial\log {\mathcal{Z}(\beta,V,\xi)}/\partial \xi >0$, the density $\rho(\xi)$ is a monotonically increasing function of $\xi$. The condition for the appearance of an absolute term in the virial expansion is $\lim_{\rho\to 0^+} \mathcal{Z}(\beta, V,\xi)  \neq 1$, as seen from   (\ref{gcz}).

\subsubsection{The EoS and the particle-number  distribution}\label{virEx}
Retaining in (\ref{virial}) only the term linear  in $\xi$ and noting that $\lambda_0=-\lambda_1$\footnote{Note that  $\rho>0$ implies $\lambda_1>0$.} we obtain:
\[
\beta pV =-\lambda_1+\lambda_1 \xi,\;\; V\rho=\lambda_1 \xi,\;\; \beta p =-\frac{\lambda_1}{V}+\rho\;.
\]
 In the standard thermodynamical limit of $V\to \infty $, this reduces to the ideal gas EoS $\beta p =\rho$. To  next order, $J\!=\!2$, Eq. \ref{virial} yields:
\[  \beta p = -\frac{\lambda_1+\lambda_2}{V}+\rho  
-\frac{\lambda_2 V}{\lambda_1^2}\rho ^2.
\] 
The corresponding partition function $\mathcal{Z}(\beta,V,\xi)\!=\!\exp[\beta p V]
\!=\!\exp\left[\lambda_2\xi^2+\lambda_1\xi\right]$ yields the p.g.f. of Hermite distribution \cite{JKK2005}: 
\begin{equation}\label{Hermite_pgf}
H(w;\lambda_1\xi,\lambda_2\xi)=\frac{\mathcal{Z}(w\xi)}{\mathcal{Z}(\xi)}=
\exp\left[\lambda_2\xi^2(w^2-1)+\lambda_1\xi(w-1)\right]
\end{equation}
with the p.m.f. satisfying the (non-Markovian) recurrence
\[
(n+2)p_{n+2}=\rho \left(p_{n+1}-\frac{4\lambda_2 \rho}{\lambda_1^2} p_n\right);\;\; \; p_0=e^{-\rho+\frac{2\lambda_2}{\lambda_1^2}\rho^2},\; p_1=\rho p_0\;.
\]
For $\lambda_2>0$ the distribution always exists and is infinitely divisible, for $\lambda_2<0$  the inequality $p_0+p_1\leq 1$ limits its validity only to 
\[
\lambda_2\leq \frac{\lambda_1^2 \left(\rho -\log (\rho +1)\right)}{2 \rho ^2}\leq 0\;.
\]

Several illustrative examples of equations of state derived via the virial expansion from selected probability generating functions are presented in Appendix \ref{ExEoS}.

\subsubsection{Reconstructing the particle-number distribution from a given  EoS} 
Given the virial expansion (\ref{virial}) truncated at order $J$, the combinants $\lambda_j$, $j = 1, \ldots, J$, of the particle-number distribution can be determined by substituting
\[
\rho \;=\; \sum_{j=1}^{J} j\,\lambda_j\,\xi^j
\]
into the function $\beta p(\rho, J)$ and re-expanding the result in powers of $\xi$ up to order $J$. By matching the resulting coefficients with those of
\[
\beta p(\xi) \;=\; \sum_{j=1}^{J} \lambda_j\,\xi^j,
\]
one obtains explicit expressions for the combinants $\lambda_j$.  

Next, decompose the set $\{\lambda_j\}_{j=1}^J$ into two sets of $J_{+}$ positive $\{\lambda_j^{+}\}_{j=1}^{J_{+}}$ and $J_{-}$ negative $\{\lambda_j^{-}\}_{j=1}^{J_{-}}$ elements and apply Proposition \ref{fincumul}. This yields the representation
\[
H(w) \;=\; \frac{H_{+}(w)}{H_{-}(w)},
\]
where
\[
H_{\pm}(w) \;=\; P\bigl(G_{\pm}(w); \lambda\bigr), 
\quad
G_{\pm}(w) \;=\; \sum_{j=1}^{J_{\pm}} \lambda_j^{\pm}\,w^j.
\]
In the special case $J_{-}=0$, the particle-number distribution is infinitely divisible. Otherwise, it can be expressed as the distribution of the difference of two infinitely divisible random variables. Several equations of state reconstructed by means of this procedure are presented in Appendix \ref{RecEoS}.

Let us note that in some instances, finding an appropriate p.g.f. can be greatly simplified by the Corollary \ref{Khatri} 
saying that every PSD is completely determined by the functional dependence of its first two moments (cumulants) on some parameter $\omega$.
Thus, for instance, studying the temperature dependence of the first two cumulants $\kappa_{1,2}=y_{1,2}(\beta)$ we can in principle reconstruct the whole GCE partition function $\mathcal{Z}(\beta, V, \xi)$.

As another example, consider the simple EoS given by $\beta p = C\rho$ with $C>0$. We investigate under which conditions this relation defines a p.g.f. Using (\ref{virexp}) and (\ref{virial}), we obtain the differential equation for the function $\log \mathcal{Z}(\xi)$ and its solution:
\begin{equation}\label{bPCr}
\log \mathcal{Z}(\xi)= C \xi\frac{\partial\log \mathcal{Z}(\xi)}{\partial \xi} = V\rho\;,\;\;  \log \mathcal{Z}(\xi) = c_1 \xi^{1/C}\;.
\end{equation}
The function
\begin{equation}\label{PoisC}
H(w)=\frac{\mathcal{Z}(w\xi)}{\mathcal{Z}(\xi)}=\exp\left[c_1 \xi^{1/C}(w^{1/C}-1)\right]
\end{equation}
defines a p.g.f. if and only if the coefficients in its power-series expansion in $w$ are nonnegative. This requirement is satisfied exclusively when $1/C = j \in \mathbb{N}$. Comparison with (\ref{eq: clust}) then shows that, in this case, the EoS (\ref{bPCr}) corresponds to an ideal gas composed of $j$-particle clusters, each obeying $\beta p_j = \rho_j$:
\begin{equation}\label{EoSj}
 \log \mathcal{Z}(\beta, V,\xi) =   
\beta p_jV  = \lambda_j \xi^j = N_j = \frac{N}{j} = \frac{V
\rho}{j} = V\rho_j\;.
\end{equation}
In the last relation, we have used the identity $N_j = N/j$, which expresses that the mean number of $j$-particle clusters, $N_j$, is equal to the mean total particle number $N$ divided by the number of particles per cluster, $j$. Consequently, the corresponding number densities must satisfy $\rho_j = \rho/j$.

\subsubsection{Perfect fluids} 
  From the first law of thermodynamics (\ref{flaw}) follows the expression for the pressure $p$  at fixed $S$ and $N$:
\[
p=-\left(\frac{\partial E}{\partial V}\right)_{S,N}=
-\varepsilon - V\left(\frac{\partial \varepsilon}{\partial V}\right)_{S,N}
\]
Assuming that the energy density depends on $S$ and $N$ through particle number density $\varepsilon=\varepsilon(\rho)$ only \footnote{Note that this assumption, which is generally not true, defines a class of EoS.}, we obtain the identity:
\begin{equation}\label{p_rho_eps}
p=
-\varepsilon - V \frac{\partial \varepsilon}{\partial\rho}\cdot\frac{\partial (N/V)}{\partial V}=
\rho \frac{\partial\varepsilon}{\partial\rho}-\varepsilon\;.
\end{equation}
The last expression, valid independently of the assumption of the constancy of $S$ and $N$, must also be true in GCE.  

Let us consider the model with $\varepsilon(\rho)=a \rho^{\omega +1}$, where $a>0$ is a constant. In this case (\ref{p_rho_eps}) leads to the EoS of a {\it perfect fluid}\footnote{The reverse equation $\varepsilon=\alpha p$ is used in \cite{Blaschke2026} to define a {\it constant gas} as a simple model of gas with zero chemical potential.}.
\begin{equation}\label{perfleos}
 p= \omega \varepsilon ,\;\omega\leq 1      
\end{equation}
Since the EoS parameter $\omega$ is directly related to the speed of sound: $c_s^2=dp/d\varepsilon=\omega$, its upper limit $\omega=1$ corresponds to a medium with the speed of sound equal to the speed of light  -- a {\it stiff fluid} \cite{Zeldovich:1961sbr}. The case of $\omega=-1$ describes the system that cannot exchange heat or perform mechanical work -- an attribute of a uniform vacuum energy in cosmology, often referred to as {\it dark energy}\cite{BlanTh}.

The perfect fluid equation of state (\ref{perfleos}) is widely employed in astrophysics and condensed matter physics to characterize fluids that lack both viscosity and heat conduction. Exemplary realizations of nearly perfect fluids that can be investigated under controlled laboratory conditions include the quark–gluon plasma \cite{Schafer:2009dj}, hadronic matter in the vicinity of the QCD crossover \cite{Rafelski:2015xej, Schafer:2009dj}, ultracold Fermi gases at unitarity \cite{Giorgini:2008zz}, and the electron fluid in graphene \cite{mueller2009}.

Let us give examples.
\begin{enumerate}[label=(\roman*), leftmargin=*]
\item A perfect fluid characterized by an equation-of-state parameter \(\omega = 0\) corresponds, in three spatial dimensions, to an ideal classical gas composed of non-interacting particles, for which the pressure is negligible compared to the energy density $\varepsilon=\frac{3}{2}T\rho$ for the non-relativistic and   $\varepsilon=3T\rho$ for the ultra-relativistic case.
\item The EoS parameter $\omega=1/3$ describes the photon gas \cite{Lanlif, BlanTh}.
\item  The case $\omega = 1$ with $\varepsilon \sim \rho^2$ corresponds to a system characterized by pairwise interactions among the particles. In this regime, each particle experiences an effective mean-field potential that is proportional to the density $\rho$, generated collectively by the remaining $N=\rho V$ particles.
 \end{enumerate}

Power-law $\varepsilon\sim \rho^{\omega+1}$ arises naturally when virial expansion (\ref{virial}) is dominated  by a single power in $\rho$. For $\omega=0,1$ analyticity of $\log \mathcal{Z}$ is guaranteed. When $\omega$ is a non-negative integer, the function 
\[
\rho^{\omega+1}=e^{(\omega+1)\log \rho}
\]
is not expandable in a power series around $\rho=0$ due to a branch point. This nonanalyticity would manifest as the GCE partition function $\mathcal{Z}(\xi)$  failing to have a power series representation in the fugacity $\xi$ at the boundary of its radius of convergence. In this case, the p.g.f. {} of particle-number distribution (\ref{gcpart}) does not exist. 
This behavior is encountered in the quantum case when treating the  $\mu \to 0$ limit of the Bose-Einstein distribution. The fact that $\log\mathcal{Z} \sim T^3 \sim \rho$
is itself analytic in $T$ but becomes nonanalytic when re-expressed in terms of $\rho$ -- the inversion $T\sim \rho^{1/3}$ is the source of all fractional powers.

Given their physical relevance, it is natural to investigate the possibility of representing perfect fluids as an ideal gas composed of $j$-particle clusters. More precisely, we ask under which conditions the equations of state (\ref{perfleos}) and (\ref{EoSj}) are mutually consistent. Using (\ref{phix}), we multiply (\ref{perfleos}) by $\beta$ and obtain
\begin{equation}\label{perfl}
 \beta p  = \beta \omega \varepsilon = \omega \phi(\beta m)\rho = \frac{\rho}{j}\;.
\end{equation}
Eliminating $\rho$ from the last two expressions on the right-hand side yields
\begin{equation}\label{perfl1}
\omega = \frac{1}{j \,\phi(\beta m)},\;\;\; j\in\mathbb{N}\;.
\end{equation}
The relations $\varepsilon = \phi(\beta m)\rho/\beta = \varepsilon(\rho,\beta)$ and $\varepsilon=\varepsilon(\rho)$ can be made compatible if $\phi(\beta m)=\phi_0 \cdot \beta m$. Under this condition, the EoS parameter $\omega$ becomes a linear function of the temperature:
\begin{equation}\label{omT}
 \omega(T) = \frac{T}{j \phi_0 m}\;.
\end{equation}
Let us now consider the model specified by (\ref{omT}). For $T \leq \phi_0 \equiv T_c$, the ideal gas of particles behaves as a perfect fluid and reaches the stiff-matter regime at $T = T_c$. In the temperature interval $T_c < T \leq 2T_c$, the system can be interpreted as a condensate of particle doublets ($j=2$), again approaching stiff matter at $T = 2T_c$. At higher temperatures, depending on the singularity structure of the partition function, the system may be realized as an ideal gas of triplets, quadruplets, and, in general, $j$-particle clusters.

\newpage
\section{Particle Distribution in the Vicinity of the Critical Point}\label{CPsec}
\subsection{Kinetics of a first-order phase transition}
We begin by quoting from K.G. Wilson's Nobel lecture, in which he describes liquid-vapor behavior near the endpoint of the coexistence curve \cite{Wilson:1983xri}. 
{\it  As  the pressure $P$ and temperature $T$  approach their critical values at the critical point (CP) 
of the phase diagram, the difference in density between water and steam approaches zero. 
At  $(P_{CP}, T_{CP})$ bubbles of steam and drops of water are intermixed at all size scales.    
Away from that point, surface tension makes small drops or bubbles unstable. Still, as water and steam become indistinguishable at the critical point, the surface tension between the two phases vanishes. 
}

The results presented in the preceding sections allow us to formulate the following kinetic model of the first-order phase transition. The assumption of infinite divisibility implies that the particle-number distribution is a compound Poisson one. We can therefore focus on the evolution of separate Poisson-distributed, and hence {\it uncorrelated}, single-particle states - the seeds of particle production. For example, instead of studying the kinetics of the process leading to the NBD (i.e., a compound Poisson Logarithmic distribution), we study only the evolution of its Logarithmic component. 

Starting from the initial Poisson-distributed seeds with the p.g.f. $P(w;\nu)$ representing density fluctuations arising independently in different parts of the volume $V$, we follow their evolution with increasing temperature $T$ in small discrete time steps $\tau$. Assuming $\tau$ to be sufficiently small that equilibrium is reestablished rapidly after each step, the creation and annihilation of new bubbles can be described in terms of a stationary birth-death process. 

In our model, we exploit three elementary transition probabilities $\alpha_1, \beta_1, \alpha_2$, describing splitting ($1\to2$), fusion ($2\to1$), and the creation of a third particle in a two-particle collision($2\to3$), respectively. The transition function of the distribution of the particles originating from a single seed is:
\begin{equation}\label{gn_bub}
g(n)=\frac{(n+1)p_{n+1}}{p_n}=\frac{\alpha_1+\alpha_2(n-1)}{\beta_1}=(n-\gamma)\xi \;;\;\; \gamma=1-\frac{\alpha_1}{\alpha_2}\;,\;
\xi=\frac{\alpha_2}{\beta_1}\;,
\end{equation}
with the corresponding p.g.f.: 
\begin{equation}\label{pgf_bub}
G_1(w;\xi,\gamma)=\frac{1-(1-\xi w)^{\gamma}}{1-(1-\xi)^{\gamma}}= \frac{G(\xi w;\gamma)}{G(\xi;\gamma)}\;;\;\; 
\gamma\leq 1\;,\; \xi\leq 1\;.
\end{equation}

Initially, the system exists in the region $k = \alpha_1/\alpha_2 - 1 = -\gamma\gg 0$ and $\xi= \alpha_2/\beta_1\ll1$, where the process $2\to3$ is suppressed due to low thermal energy and low particle number density, and the p.g.f. (\ref{pgf_bub}) with $\gamma<0$ describes the zero-truncated negative binomial distribution. As $\alpha_2\to \alpha_1$, this regime persists up to temperatures when  $\alpha_2 = \alpha_1$ -- the particle distribution morphs into the Logarithmic one. A further rise in $T$ leads to the prevalence of inelastic collisions over fragmentation,  $\alpha_2>\alpha_1$, and to the emergence of the Extended Sibuya distribution with $0<\gamma=1-\alpha_1/\alpha_2<1$. Finally, as $\alpha_2\to \beta_1$, and hence $\xi\to 1^{-}$, the particle number turns into the Sibuya distribution. 
 
 In this way, initially Poisson-distributed single-particle seeds - the uncorrelated precursors of a new phase within the old one - evolve into a stochastic ensemble of bubbles described by the compound Poisson-Sibuya, i.e., the Discrete Stable distribution. Assuming that the limit $\xi=1$ can be reached at the finite temperature ($\beta\neq 0$), the DSD then describes the state of matter with $\mu=0$ where the particle number $N$ is no longer an independent thermodynamic variable. Although the power-series character of its probability distribution is broken, and its moments diverge, its p.g.f. is well defined. 

 We now turn to the analysis of the grand-canonical ensemble (GCE) partition function. Using Eqs.~(\ref{PSDfromCPD}) and (\ref{pgf_bub}), the probability generating function (p.g.f.) of the total particle-number distribution can be expressed as
\begin{equation}\label{ZZ}
P(G_1(w);\xi, \gamma, \nu)=
 e^{\nu \left( G_1(w;\xi,\gamma)-1\right)}=
 \frac{e^{\lambda \left(G(\xi w,\gamma)-1\right)}}{e^{\lambda \left(G(\xi,\gamma)-1\right)}}=
 \frac{\mathcal{Z}(\xi w; \gamma, \lambda)}{\mathcal{Z}(\xi, \gamma, \lambda)}\;,\qquad \nu=\lambda G(\xi,\gamma)\;.    
\end{equation}
 
The behavior in the vicinity of the critical point (CP) can be obtained by computing the cumulants via Eq.~\ref{gccomb}:
\begin{equation}\label{cumDSD}
  \kappa_{\ell}
  =\xi^{\ell}\frac{\partial^{\ell}\log \mathcal{Z}(\lambda, \xi, \gamma)}{\partial\xi^{\ell}}
  = \frac{\gamma \lambda \xi^{\ell}}{(1-\xi )^{\ell-\gamma}}
 \prod_{i=1}^{\ell -1}(i-\gamma)
  \;.
\end{equation}
In particular, as $\xi \to 1^{-}$, the first three cumulants exhibit power-law divergences:
\begin{equation}
\kappa_1= \frac{\gamma \lambda \xi}{ (1-\xi )^{1-\gamma }}\;,\qquad 
\kappa_2=\frac{\gamma\lambda (1-\gamma)   \xi ^2}{  (1-\xi )^{2-\gamma}}\;,\qquad 
\kappa_3=\frac{\gamma\lambda (1-\gamma) (2-\gamma)   \xi ^3}{ (1-\xi )^{3-\gamma}}\;.
\end{equation}

Concurrently, the negative pressure exerted on the emerging phase by the surrounding medium is progressively reduced:
\[
\beta p V =\log \mathcal{Z}(\xi, \gamma,\lambda)
=\lambda \bigl(G(\xi,\gamma)-1\bigr) =
-\lambda(1-\xi)^{\gamma} \xrightarrow[\xi \to 1^{-}]{} 0^{-}\;,
\]
where the parameter $\gamma$ serves as the critical exponent characterizing the phase transition. This behavior becomes particularly evident in the equation of state in the vicinity of the critical point. By employing Eqs. (\ref{PSDcomb}) and (\ref{cumDSD}), we derive, up to fourth order in $\rho$, the following expression:
\begin{equation}
\beta p \approx 
-\frac{\lambda}{V} +\rho -\frac{(1\!-\!\gamma) \rho^2 V}{2 \gamma  \lambda }
+\frac{(1\!-\!\gamma) (1\!-\!2 \gamma) \rho^3 V^2}{3 \gamma ^2 \lambda ^2}
-\frac{(1\!-\!\gamma) (2\!-\!3 \gamma ) (1\!-\!3 \gamma) \rho^4 V^3}{8 \gamma ^3 \lambda ^3}\;.
\end{equation}
Independent of the value of $\gamma$, the surrounding medium exerts a constant negative pressure $-\lambda/V$ on the emergent phase. While the effective two-particle interaction is always attractive, the three- and four-bubble interactions become attractive only for $\gamma>1/2$ and $2/3>\gamma>1/3$, respectively. Consequently, up to fourth order in the density $\rho$, for $2/3>\gamma>1/2$, i.e., when $5/3>\alpha_2/\alpha_1>3/2$, all bubbles experience mutual attraction.

 It is instructive to highlight the nontrivial connection between the power-series representation and the infinite-divisibility property of the full compound Poisson particle-number probability generating function  $P(\nu;G(w; \xi,\gamma))$. When the function $G(w;\xi,\gamma)$ is given by (\ref{pgf_bub}), the requirement that the corresponding p.g.f. simultaneously (i) belongs to the PSD class and (ii) is infinitely divisible implies the following relation:
\begin{equation}\label{part_bub}
\nu=\lambda G(\xi,\gamma)= 
\lambda\left (1-\xi\right)^{\gamma}=
\lambda\left (1-\frac{\alpha_2}{\beta_1}\right)^{1-\frac{\alpha_2}{\alpha_1}}
\;,
\end{equation}
where \(\lambda>0\) is a multiplicative constant that characterizes the mean number of seeds at temperature \(T\) in the reference case \(\alpha_2=\alpha_1\), i.e., for \(\gamma=0\).
Consequently, for $\alpha_2<\alpha_1$, i.e., $\gamma<0$, the mean number $\nu(\xi,\gamma)$ increases with $T$, whereas for $\alpha_2>\alpha_1$, i.e., $\gamma>0$, it decreases with $T$. Approaching the critical point results in the complete extinction of the original seeds, in the sense that
\[
\lim_{\xi\to 1^{-}}\nu(\xi,\gamma) = 0.
\]
In the terminology of liquid–vapor critical phenomena, the density contrast between the two coexisting phases vanishes in this limit. Therefore, as $\mu = T \log (\alpha_2/\beta_1)\to 0^{-}$, the presence of seeds ceases to be necessary for the generation of particles.

Another dramatic effect follows from the self-decomposability of the Sibuya p.g.f.: $\mathcal{S}(w;\gamma_1 \cdot \gamma_2) = \mathcal{S}(\mathcal{S}(w;\gamma_2);\gamma_1)$. Once $T=T_{CP}$, the notion of a single bubble loses meaning; each bubble represents an infinite tower of smaller ones. At the CP, the cost of producing new bubbles is zero.

Finally, we consider a (rather idealized) scenario in which bubble production via the $2\to 3$ process increases with temperature more rapidly than their fusion via $2\to 1$, while remaining subdominant to fragmentation via $1\to 2$. This situation corresponds to a stationary stochastic process with parameters $\alpha_1>\alpha_2\geq \beta_1$. As the critical point is approached, the original zero-truncated negative binomial distribution, specified by $k=1 -\alpha_1/\alpha_2>0$, is deformed into a heavy-tailed distribution that exists for $0<k\leq 1$: 
\begin{equation}\label{ztrNBD_lim}
\lim_{\xi\to 1^{-}}G(w; \xi, k)=1-(1-w)^k\;,\;\; p_n=(-1)^{j+1} \binom{k}{j}\;;\; n\geq 1.
\end{equation}
Consequently, the parameters must satisfy the constraint $\alpha_2<\alpha_1\leq 2\alpha_2$.

\subsection{Discrete stable distribution as a fixed point}
We now show how the particle-number distribution at the critical point arises from a self-similar stochastic process. Consider the r.v. $M \in \Z$ with the p.g.f. $H(w,t)$, representing the number of particles (e.g., bubbles of a new phase within the old one) at time $t$. We follow its evolution with increasing temperature $T$ in discrete dimensionless time steps $\tau$. The increase in the number of particles in a single step $\tau$ is described by a r.v. $N$ with p.g.f. $Q(w)$. After two successive steps $t\to t^{\prime} = t+\tau,~ t^{\prime} \to t^{\prime\prime} =t+2\tau$, the p.g.f.{}s  satisfy\footnote{Here we adapt to our discrete case the analysis \cite{Calvo:2010jz} performed in a continuous variable.}:
\begin{equation}
H(w, t+ \tau) = Q(w) H(w,t)\;,\;\;\; H(w, t+ 2\tau) =  Q^2(w) H(w,t)\;.
\label{rwdis}
\end{equation}

We assume the system is self-similar: the evolution of the number of particles in volume $V$ over time steps $\tau$ is the same as in the sub-volume $V^{\prime}= aV, ~a<1$ over time steps $\tau/2$. The evolution in  $V^{\prime}$ is governed by Eq. \ref{rwdis} with the  replacement $H(w,t) \to H^{\prime}(w,t^{\prime}=2t)$ and $Q(w) \to Q_a(w)=Q(\mathcal{B}_a)$: 
\begin{equation}
H^{\prime}(w,t^{\prime}+ \tau)= Q^2_{a}(w) H^{\prime}(w,t^{\prime})= T_{a}Q(w)H^{\prime}(w,t^{\prime})\;, \;\;\; T_{a}Q(w):= Q^2_a(w)\;.
\label{RGdis}
\end{equation}
Here, $T_a$ is an operator transforming the r.v. $N$ into the rescaled sum $a\odot (N_1\!+\!N_2)$ of two \iid copies $N_{1,2}= N$ (see (\ref{thin})). Applying $T_a$ iteratively $j$ times gives:
\begin{equation}
(T_a)^j Q(w)H^{\prime}(w,t^{\prime}) = (Q_{a^j}(w))^j H^{\prime}(w, 2^j t)\;.
\label{rgdism}
\end{equation}
 In the limit $j \to \infty$:
\begin{equation}\label{inflim}
H(x, t+ \tau) = Q_{\infty}(w)H(w,t)\;,\;\;\; Q_{\infty}(w):= \lim_{j \to\infty}(T_a)^{j} Q(w)\;,
\end{equation}
where $Q_{\infty}(w)$  is the fixed point of the functional equation: 
\begin{equation}
\label{fxpt}
T_{a}Q_{\infty}(w)= Q_{\infty}^2(\mathcal{B}_a) = Q_{\infty}(w)\;.
\end{equation}
Solving (\ref{fxpt}) gives  $Q_{\infty}(w)=\exp[-\lambda(1-w)^{\gamma}]$ with $2a^\gamma=1$. This p.g.f. exists for $\lambda>0$ and $\gamma \in (0, 1]$, and belongs to the Discrete Stable distribution (DSD)(\ref{eq: dsdpgf}). As a check, substituting into (\ref{fxpt}):
\begin{equation}
T_a Q_{DSD}(w) = e^{-2\lambda(1- (1-a(1- w)))^{\gamma}} = 
e^{-2\lambda a^{\gamma}(1-w)^{\gamma}}= e^{-\lambda(1-w)^{\gamma}} = Q_{DSD}(w)\;.
\end{equation}

The stability of the DSD is confirmed by Theorem \ref{eq: infdthin}: among all infinitely divisible p.g.f.{}s $H(w;\lambda)=\exp\{ \lambda (G(w)-1)\}$ with $\lambda$ depending on $a$, only the DSD remains infinitely divisible under the thinning for all $a\in (0,1]$.  By Theorem \ref{PSD_thin}, among all GCE partition functions with $\xi(a)$ depending on $a$, only 
$\mathcal{Z}(\beta,V,\xi)=(1-\xi)^{\gamma}$ is form-invariant under the thinning. The corresponding p.g.f. $\mathcal{Z}(\xi w)/\mathcal{Z}(\xi)$ is the NBD with $\theta=\xi$ and $k=-\gamma>0$. Thus, had we prioritized the existence of the GCE partition function throughout the heating process at the expense of infinite divisibility,  we would obtain the NBD.

In writing  (\ref{inflim}), we have exploited the fact that the limit is the same for all $Q(w)$ in the {\it domain of attraction} of the fixed point. To illustrate this property, consider the GDLD p.g.f. $Q(w;\lambda,\beta,\gamma)=(1+\lambda(1-w)^{\gamma}/\beta)^{-\beta}\;,\;\;\beta>0$ (Eq.\ref{GDLDpgf1}).
Since GDLD is a compound NBD-Sibuya distribution and the Sibuya distribution itself is self-similar under the {\it renormalization semi-group} transform (\ref{RGdis}), the same holds for GDLD:
\[
Q_a(w;\lambda,\beta,\gamma ) = Q(w;a^{\gamma}\lambda,\beta, \gamma) = \left[1+a^{\gamma}\lambda(1-w)^{\gamma}/\beta\right]^{-\beta}\;.
\]
Assuming the parameter $a(j)=j^{-1/(j \gamma)}$ decreases over $j$ time steps, the limit  (\ref{inflim}) gives:
\[ \lim_{j \to\infty} (Q(w;a^j\lambda,\gamma))^j=
\lim_{j \to\infty}\left[1+(\lambda/(\beta j^{\gamma})(1-w)^{\gamma}\right]^{-j \beta} = e^{-\lambda(1-w)^{\gamma}}\;.
\]

While the above procedure is valid for $a<1$, it is natural to ask whether it can be extended to $a>1$. In other words, we ask if the DSD p.g.f. is scalable. The answer is given by Theorem \ref{thK}: the function $Q(1-a(1-w))$ is a p.g.f. $\forall a>0$ iff there exists a Laplace transform $\phi(s)$ of a r.v. $X\in \R$  such that $Q(w)=\phi(1-w)$. For $Q_{DSD}(w;\lambda,\gamma)$ this is indeed satisfied:
\begin{equation}
\label{eq: Pollard1}
e^{-\lambda(1-w)^{\gamma}}= \int_0^\infty e^{-(1-w)x}p(x;\lambda,\gamma) dx\;,\;\;\; p(x;\lambda,\gamma)=\lambda^{-1/\gamma}f_{\gamma}(\lambda^{-1/\gamma}x)\, ,
\end{equation}
where 
$$f_{\gamma}(x)=\frac{1}{\pi} \sum^{\infty}_{j=0}\frac{(-1)^{j+1}}{j!x^{1+\gamma j}}\Gamma(1\!+\!\gamma j)\sin(\pi\gamma j)$$
is the p.d.f. of the one-sided continuous stable distribution (see e.g., \cite{Penson2010}). The case $\gamma=1/2$ corresponds to the p.d.f. of the L\'{e}vy distribution.
\begin{equation}\label{levy}
    f_{1/2}(x)= \frac{\lambda e^{-\frac{\lambda^2}{4 x}}}{2\sqrt{\pi } x^{3/2}};\;\;\;
\int_0^\infty e^{-(1-w)x}f_{1/2}(x) dx = e^{-\lambda \sqrt{1-w}}=Q_{DSD}(w;1/2)\;.
\end{equation}

Thus, the evolution of the original process in volume $V$ over time steps $\tau$ is the same as in volume $V'=aV>V,\;a=2^{1/\gamma}$  over steps of size $2\tau$. Consequently, the domain of attraction of the fixed point $Q_{DSD}(w;\lambda,\gamma)$ consists of all scalable compound p.g.f. of the form $H(\mathcal{S}(w;\gamma))$. 

The particle-number distribution at the critical point is obtained by expanding $Q_{DSD}(w;\lambda,\gamma)$ in powers of $w$:
\begin{equation}\label{eq: DSDpmf}
p_n =(-1)^n \sum_{j=0}^{\infty} \binom{\gamma j}{n}\frac{(-\lambda)^j}{j!} 
\xrightarrow[n\to\infty]{}\frac{1}{n}\sum_{j=0}^{\infty} \frac{(-\lambda)^j n^{-\gamma j}}{\Gamma(-\gamma j)j!}\;, 0<\gamma<1\;.
\end{equation} 
This distribution is completely monotone for $\lambda<1/\gamma$; ~for $1/\gamma\leq\lambda<1/\gamma+1$ it has a maximum at $p_1$; and for $1/\gamma+1\leq \lambda<1/(2\gamma) \sqrt{5 \gamma ^2+4}+3/2$ a maximum at $p_2$. The situation for $n>2$ is more complex.  
The parameters $\lambda\!=\!-\log p_0$ and $\gamma\!=\!-p_1/(p_0 \log p_0)$ can be extracted from $p_0$ and $p_1$ alone.

Two finite-sum expansions for $p_n$ are also available. One was found in  \cite{Christoph1998a}:
\begin{equation}
    p_n=(-1)^n e^{-\lambda}\sum_{m=0}^{n}r_m\;,\;\;\; r_m=\sum_{j=0}^{m} \binom{m}{j}\binom{\gamma j}{n} (-1)^j \frac{\lambda^m}{m!}f
\end{equation}
The other is based on the combinants $\lambda_j$: 
\begin{equation}\label{DSDcum}
 \log Q_{DSD}(w;\lambda,\gamma)
 =\sum_{j=1}^{\infty}\lambda_j (w^j-1)
 \;,\;\; \lambda_j=  (-1)^{j} (-\lambda)  \binom{\gamma }{j}\;. 
\end{equation}
Applying Fa\`a di Bruno's formula:
\begin{equation}\label{FaaDSD}
p_n = \frac{Q_{DSD}^{(n)}(0)}{n!} = e^{-\lambda} \sum  \prod_{j=1}^{n}\frac{\lambda_j^{m_j}}{m_j!}=
(-1)^n e^{-\lambda}\sum \prod_{j=1}^{n}\frac{(-\lambda)^{m_j}}{m_j!}\binom{\gamma}{j}\;.
\end{equation}
where the sum runs over all $n$-tuples $\{m_1,m_2,\ldots, m_n\}$ of non-negative integers $m_j$ satisfying $\sum_{j=1}^{n} j m_j =n$.

From the scalability of the DSD and Corollary \ref{cor:g}, it follows that the p.m.f. (\ref{eq: DSDpmf}) is Markovian -  it satisfies the one-step relation $(n+1)p_{n+1}=p_ng(n)$. However, differentiating $Q_{DSD}^{\prime}(w;\lambda,\gamma)=\lambda Q_{DSD}(w;\lambda,\gamma)\mathcal{S}^{\prime}(w;\gamma)$ and expanding in series leads to a more complex recursion \cite{Christoph1998a}:
\begin{equation}\label{DSDrecur}
    (n+1)p_{n+1}=\lambda \sum_{m=0}^{n}p_{n-m}(m+1)(-1)^m \binom{\gamma}{m+1}
\end{equation}
Thus, $g(n)$ is not a simple ratio of finite-order polynomials in $n$ but depends explicitly on $p_n$. 
\[
g(n)=\gamma\lambda+\frac{\sum_{m=1}^{n}p_{n-m}(m+1)(-1)^m \binom{\gamma}{m+1}}{p_n}\;.
\]
This holds even in the simplest case $\gamma=1/2$, where the transition function is a ratio of two modified Bessel functions of the second kind: 
\[ g(n)= \frac{\lambda}{\sqrt{2}} \frac{K_{n+\frac{1}{2}}\left(\sqrt{2} \lambda \right)}{ K_{n-\frac{1}{2}}\left(\sqrt{2} \lambda \right)}.
\]     

We conclude by providing an additional argument that underscores the distinguished role of the DSD, based on Corollary \ref{DSD_unique}. 

Let $N, N_1, N_2 \in \mathbb{Z}$ be \iid{} infinitely divisible random variables with a common probability generating function (p.g.f.) given by $H(w;\lambda) = P(G(w);\lambda)$. Then the distributional identity
\[
c \odot N \;=\; a \odot N_1 \;+\; b \odot N_2
\]
holds if and only if the secondary p.g.f. $G(w)$ corresponds to a scaled Sibuya distribution.

Moreover, every DSD can be represented either as a compound Poisson–Sibuya distribution with parameters $\{\gamma,\lambda\}$ or as a compound scaled Sibuya distribution with parameters $\{\gamma,\nu\lambda\}$. In particular, we have
\begin{equation}
P(\mathcal{S}(w;\gamma), \nu\lambda)\;;\;\;\; P(w;\lambda) = e^{(w-1)\lambda}\;,
\end{equation}
where $P(w;\lambda)$ denotes the p.g.f. of a Poisson distribution with intensity parameter $\lambda$.

\newpage
\section{Tsallis Thermodynamics}\label{TT}
\subsection{Statistical mechanics of systems with long-range interactions }
In this section, which is partially independent of the preceding discussion, we investigate the statistical mechanics of systems characterized by a non-extensive generalization of the Boltzmann–Gibbs–Shannon (BGS) entropy \cite{Tsallis:1987eu}. 
We first develop the formal framework of deformed probability distributions and subsequently examine their implementation within the context of statistical mechanics.
\subsection{Deformed discrete distributions} \label{Escd}
\subsubsection{Escort probability distributions}

Let $q\in \mathbb{R}$ be a parameter and $p_n$ the p.m.f. of the r.v. $N$. The {\it escort} p.m.f. $P_n(q)$ and the {\it q-moment} $\E_q(N^k)$  are defined as:
\begin{equation}\label{Escmom}
  P_n(q):=\frac{p_n^q}{\sum_n p_n^q},\;\; \E_q(N^k):=\sum_n n^k P_n(q)\equiv \langle n^k \rangle _q\;.
\end{equation}
The parameter $q$ acts as a microscope exploring different regions of $p_n$: for $q > 1$, the more singular (large) values are amplified; for $q < 1$, the less singular (small) regions are accentuated. Originally, escort probability distributions with $q>1$ were introduced in \cite{Beck1993} to assign suitable weight to the tails of heavy-tailed distributions. 
\begin{Example}
    Consider the r.v. $N$ with p.m.f. (\ref{jj1}) and its escort distribution:
    \[
    p_n=((n+2)(n+1))^{-1}\;,\;\;P_n(q)=((n+2)(n+1))^{-q}\;.
    \]
    Although the ordinary moments $\E N^k $ exist only for $k<1$, the $q$-moments $\E_q N^k$  are finite for $q>(k+1)/2$. 
    
\end{Example}
Another interpretation of the escort probability (\ref{Escmom}) comes from the thermodynamic analogy:  \begin{equation}\label{analogy}
 P_n(q)=  \frac{p_n^q}{\sum_n p_n^q}= \frac{\exp(-\beta E_n)}{\mathcal{Z}(\beta)}\;, \;\;\; \mathcal{Z}(\beta)=\sum_n \exp(-\beta E_n)
\end{equation}
Here $P_n(q)$ is the probability of finding the system at temperature $T=\beta^{-1}$ with energy $E_n$.

\begin{Example}
Energy levels of a quantum harmonic oscillator, $E_n \!=\! (n\!+\!\tfrac{1}{2})\hbar\omega$, lead to the Bose-Einstein distribution, here ``rediscovered'' via the thermodynamic analogy: 
\begin{equation}
P_n(q)=\frac{\exp(-\beta E_n)}{\mathcal{Q}(\beta)} =
(1\!-\!\theta)\theta^n\;,\;\;\theta = \exp(-\hbar \omega) \;.
\end{equation}
\end{Example}
\subsubsection{$q$-deformation}
To gain further intuition about the parameter $q$,  we follow \cite{Tsallis:2009zex}  and consider the first-order ODE
\[
\frac{dy}{dx}=x^q,\;\; q\in\mathbb{R},\;\; y(0)=1\;.
\]
Its solution $y(x)$ and inverse $x(y)$ define the {\it q-exponential} and {\it q-logarithm}:
\begin{equation}\label{e_qlog_q}
    y(x)=[1+(1-q)x]_{+}^{1/(1-q)}:=e_q^{x},\; x(y)=\frac{y^{1-q}-1}{1-q}:=\log_q(y)\;;\;[x]_{+}=\max\{0,x\}
\end{equation}
These are the {\it $q$-deformation} of the ordinary exponential $\exp_1(x)\!=\!\exp(x)$ and logarithm $\log_1(y)\!=\!\log(y)$.
From the first equation in (\ref{e_qlog_q}):
\begin{equation}\label{e_qprod}
e^x_q e^y_q=e^{x+y+(1-q)x y}_q:=e^{x \oplus_q y}_q,\;\;e^x_q e^{-x}_{2-q}= (e^x_q)^qe^{-qx}_{1/q}=1; \;\; \forall q\;.    
\end{equation}
The second one yields the non-additivity of the $q$-deformed logarithm:
\begin{equation}\label{log_qprod}
\log_q(x_A\cdot x_B)=\log_q(x_A)+\log_q(x_B)+(1-q)\log_q(x_A)\cdot \log(x_B)\;.    
\end{equation}
Finally:
\begin{equation}\label{de_qdlog_q}
\frac{d e_q^x}{dx}=(e_q^x)^q\;,\;\;\; \frac{d \log_q(y)}{dy}=y^{2-q}\;,\;\;\; \int e_q^x dx=\frac{e_q^x}{q-2}\;,\;\;\;\int\log_q(y)dy=\frac{\frac{y^{2-q}}{2-q}-y}{1-q} \;.
\end{equation}

\begin{Example}
Using $\exp_q(x)$ and $\log_q(y)$ to generalize some p.g.f.{}s:
\begin{enumerate} [label=(\roman*), leftmargin=*]
\item The $q$-deformed Poisson p.g.f. is:
    \begin{equation}\label{qPoiss}
 e_q^{\lambda(w-1)}=
\left[1+(1-q)\lambda(w-1)\right]_{+}^{\frac{1}{1-q}}\;.   
\end{equation}
For $q>1$, this is the NBD p.g.f. with $\langle n \rangle =\lambda$ and $k=1/(q-1)$ \cite{Tsallis:2009zex}; for $q<1$ with $m=1/(1-q)\in \N$ and $(1-q)\lambda=a<1$, it gives the Binomial distribution with success probability $a$ and the mean $ma$.

\item  The $q$-deformed Logarithmic p.g.f. is:
\begin{equation}
\frac{\log_q(1-\theta w)}{\log_q(1-\theta)}=
\frac{1-((1-\theta ) w)^{q-1}}{1-(1-\theta )^{q-1}}\;.
\end{equation}
For $1<q<2$, this reduces to the Extended Sibuya p.g.f. (\ref{exsbd_pgf}) with $\gamma=q-1$; for $q<1$, to the zero-truncated NBD (\ref{eq:nbdtrunc}).

\item The $q$-deformation of the p.g.f. $H(w) = \text{sech}(\sqrt{1-w})$ is:
\begin{equation}\label{qsechpgf}
H_q(w) = \frac{2}{e_q^{\sqrt{1-w}}+e_q^{-\sqrt{1-w}}}\;;\;\; \E N =\frac{q}{2}\;,\;\;F_2 =
\frac{\langle N (N-1) \rangle}{\langle N \rangle^2}=\frac{1}{24} (2-q) (6 q-1)\;.
\end{equation}
This p.g.f. exists for $0\leq q <2$, with 
$p_0(q)=2(q^{\frac{1}{1-q}}+(2-q)^{\frac{1}{1-q}})^{-1}$
and $p_n(1/2)=4\cdot 5^{-j-1}$.
Note that for $q<1/6$, the second normalized factorial moment $F_2$ becomes negative.
\end{enumerate}
\end{Example}

\subsection{Havrda-Charv\'at-Tsallis non-extensive entropy}
The $q$-deformed logarithm (\ref{e_qlog_q}) allows us to generalize the BGS entropy $S^{BGS}$  in the following way \cite{Tsallis:1987eu, HCh:1967}:
\begin{equation}\label{Tsalent}
S_q=-\sum_n p_n\log_q p_n=\frac{1-\sum_n p_n^q}{q-1},\;\; S_1=S^{BGS}  \;.
\end{equation}

To maximize $S_q$ subject to constraints specifying the values of the $q-$moments $\langle f_i(n) \rangle_q =F_i$, we apply the Lagrange multiplier method:
\[
\hat{S}=S_q - \lambda \sum_n (p_n-1) -  \sum_{i=1}^m \lambda_i \sum_n [f_i(n) P_n (q)  - F_i]\;.
\]
Setting $\delta \hat{S}/\delta p_n=0$ and solving for the unknown function $p_n$ gives:
\begin{equation} \label{meppmfq}
p_n=\frac{e_q^{-\sum_{i=1}^m \lambda_i f_i(n)}}{Z(\lambda_1,\ldots,\lambda_m)},\;\;\;\;
Z(\lambda_1,\ldots,\lambda_m)=  \sum_m e_q^{ -\sum_{i=1}^m \lambda_i f_i(n)}\;.
\end{equation}

For $m=1$ and $\langle f_1(n) \rangle_q =\langle n \rangle_q$, Eqs.(\ref{meppmfq}) and (\ref{e_qlog_q}) give:
\begin{equation}\label{Hurwpmf}
p_n
=\frac{[1+(q-1)\lambda_1 n]_{+}^{1/(1-q)} }{Z(\lambda_1)}
=\frac{(1 +  n/a)^{-s}}{Z(s/a)}
;\;\; s=\frac{1}{q-1},\; a=\frac{s}{\lambda_1}\;.
\end{equation}
For $a>0$  and $s\geq 2$ (i.e. for $q\leq3/2$) with 
$Z(s/a)=a^{-s}\zeta (s,a)$ - where $\zeta (s,a)=\sum _{k=0}^{\infty } (a+k)^{-s}$ is the Hurwitz zeta function - this p.m.f. represents the zero-inflated {\it Hurwitz distribution} \cite{JKK2005} \footnote{The original Hurwitz distribution is defined only for $n\in \N$.}.
This distribution exists on $\Z$ for $1<q\leq 3/2$ with finite moments $\langle n^j \rangle$ for $s>j+1$, i.e., for $1/(q-1)>j+1$. It is heavy-tailed for $q \geq 3/2$ with support on $\Z \backslash~ \infty$.
For $a<0$, the p.m.f. (\ref{Hurwpmf}) exists for arbitrary $s$ and $n=0,1,\ldots, [-a]$. 

Although the Hurwitz p.m.f. (\ref{Hurwpmf}) is Markovian for all $s$, its interpretation as a solution of the stationary B-D processes (Section \ref{kinmodels}) is possible only for  $s\in \N$:
\begin{equation} \label{Hurwgn}
    g(n)=(n+1)\frac{p_{n+1}}{p_n}
    =(n+1)\left(\frac{n+a}{n+a+1}\right)^s\;.
\end{equation}
In particular, for $s=1$ we obtain $\alpha_0/\beta_0=a/(a+1)$, $\alpha_1/\beta_0=(a+2)/(a+1)$ and $\alpha_2/\beta_0=\beta_1/\beta_0=1/(a+1)$. 
For $a=1$, the transition function $g(n)=(n+1-\gamma)(n+1)/(n+2)$ with $\gamma=0$ coincides with that of the Shifted Extended Sibuya  (see Table \ref{tab:mark}); for $\alpha_2 = \beta_1 = 3\alpha_1 = \alpha_0=1/2 \beta_0$, with the Shifted  Sibuya distribution.

The PSD associated with the Hurwitz distribution (\ref{Hurwpmf}) is:
\begin{equation}\label{HurPSD}
 p_n = \frac{(a +  n)^{-s}\theta^n}{\Phi (\theta,s,a)}\;,\;\;g(n)
    =(n+1)\left(\frac{n+a}{n+a+1}\right)^s \theta\;,\;\; s\in \N\;.
\end{equation}
For $s=1$ ($q=2$)  (\ref{HurPSD}) gives a  generalization of the Logarithmic distribution with the p.m.f $(n+a)^{-1}\theta^n/\Phi (\theta, s, a)$; for $s=a=1$, the $\gamma\to 0$ limit  of the Shifted Extended Sibuya distribution  (Table \ref{tab:mark}). In stationary B-D processes, this limit corresponds to $\alpha_1 \to 3\alpha_0$; for $\alpha_2 \to \beta_1$, the process is divergent.

\subsection{Energy distribution and partition function in non-extensive thermodynamics}
We now apply the above formalism to the statistical mechanics of systems with long-range interactions and determine the energy distribution $p_q(E)$ for a single microstate with a fixed average energy. The maximization of the Tsallis entropy
\[
S_q = -\int_{0}^{\infty} p(E)\,\log_q p(E)\, dE
\]
under the constraint of fixed \(q\)-expectation value of the energy,
\[
\langle E \rangle_q = \int E\, P(E,q)\, dE\;,\qquad P(E,q) = \frac{p^q(E)}{\int p^q(E)\, dE}\,,
\]
with $\beta$ as the Lagrange multiplier, yields the desired distribution:
\begin{equation}\label{qBoltz}
p_q(E)=\beta e_q^{-\beta E}
=\beta(q-2)\left[1+(1-q)\beta E\right]_{+}^{\frac{1}{1-q}}\;.
\end{equation}
This p.d.f. with $q>1$ has found applications in various fields, including particle physics \cite{Cleymans:2011in}. Note that for $q<3/2$, the first and higher moments of $E$ are non-existent. It is worth noting that the p.d.f. (\ref{qBoltz}) can be equivalently obtained from the Boltzmann factor $e_1^{-\beta E}$  by assuming that $\beta$ is a random variable which fluctuates according to a gamma distribution with mean $\langle \beta \rangle = \beta_0$ and variance $\langle \beta^2 \rangle - \langle \beta \rangle^2 = (q-1)\beta_0$ \cite{Wilk:1999dr}.

In terms of the energy density $\sigma(E)$, the GCE partition function (\ref{eq:sigma}) of a system with non-additive entropy reads: 
\begin{equation}\label{eq:sigmaq}
\mathcal{Z}_q(\beta)= \int_0^{\infty}\sigma(E) e_q^{-\beta E} dE = \int_0^{\infty}\sigma(E)  (2-q)\left[1+(1-q)\beta E\right]_{+}^{\frac{1}{1-q}} dE\;.
\end{equation}
For energy density
$\sigma(E, V, \xi)=f(E) e^{\beta_0 E};\; \beta_0>0$
with $f(E)\sim E^{\alpha}$ polynomialy bounded, the integral (\ref{eq:sigmaq}) has a singularity only for $q=1$ (the additive case) or $\alpha \leq -1$.

Consider the GCE partition function of a system composed of two independent parts $\mathcal{Z}(\beta) = \mathcal{Z}_1(\beta)\mathcal{Z}_2(\beta)$. 
In the additive case, the total free energy $F=F_1+F_2$ and the pressure $P=P_1+P_2$ are both additive, as in Eq. \ref{pressure}.
For a system with non-additive entropy, neither quantity  is additive (see Eq.\ref{log_qprod}:
\begin{equation}\label{eq:sigmaq1}
-\beta F(q) = \log_q\mathcal{Z}(\beta)= 
-\beta \left[(F_1+F_2) -(q-1)\beta (F_1\cdot F_2) \right]\;.
\end{equation}
In particular, for $F_1=F_2$ we have $F=2F_1-(q-1)\beta F_1^2$. Hence, for $(q-1)\beta=2$, i.e., at the temperature $\beta=2/(q-1)$ the total free energy becomes zero: $F=0$.

\subsection{Scalability in non-extensive thermodynamics}
Rewriting the p.d.f. (\ref{qBoltz}) in new variables $s$ and $\nu$:
\begin{equation}\label{Tcontex}
 p(E;s,\nu)= \frac{s-1}{\nu}\left[1+\frac{E}{\nu}\right]^{-s}\;;\;\;\; s=\frac{1}{q-1}>1\;,\;\; \nu=\frac{s}{\beta}>0
\end{equation}
we calculate the moments $\langle E^r \rangle $:
\begin{equation}\label{momTcontex}
\langle E^r \rangle =\frac{(s-1) \nu ^r \Gamma (r+1) \Gamma (-r+s-1)}{\Gamma (s)}\;,\;\;\; \langle E \rangle =\frac{\nu}{s-2}\;,\;\;
\langle E^2 \rangle =\frac{2 \nu ^2}{(s-2)(s-3)}\;.
\end{equation}
They are finite for $s>r+1$. 

Since $\nu$ represents the energy scale, performing the Laplace integral (\ref{eq: lapdens}) gives a new scalable p.g.f. of non-negative r.v. $N$:
\[
T(w;s,\nu)=\int_0^{\infty}e^{-(1-w)E}p(E;s,\nu)dE
\]

The factorial moments $\langle (N)_r\rangle$ of the r.v. $N$ with the p.g.f. $T(w;s,\nu)$ are equal to ordinary moments of the continuous r.v. $E$ given by (\ref{momTcontex}) (see Proposition \ref{cor:g}). For $1<s \leq 2$, the distribution is heavy-tailed. The expansion $T(w;s,\nu)=\sum_n t_n w^n$ gives the PSD p.m.f. $t_n$ satisfying the two-step recurrence:
\begin{eqnarray}\label{Tsalpmf}
(n+2)t_{n+2}=(s+\nu-n-2)t_{n+1}-\nu t_n\;,\\
t_1=(s-1)e^{\nu } \nu  (E_{s-1}(\nu )-E_s(\nu ))\;,\;\; t_0=(s-1)e^{\nu} E_{s}(\nu) \nonumber \;.
\end{eqnarray}
However, owing to its scalability, $t_n$ also satisfies the one-step relation (\ref{eq: gnkl}) with transition function:
\[
g(n)=\frac{(n+1) U(s,s-n-1,\nu )}{U(s,s-n,\nu )}\;;\;\;U(a,b,z)=\frac{1}{\Gamma (a)}\int _0^{\infty }t^{a-1} (t+1)^{-a+b-1} e^{-zt}d t \;,
\]
where $U(a,b,z)$ is the Tricomi confluent hypergeometric function. For $\nu \to \infty$,  $g(n)\approx (n+1)(1-s/\nu)=(n+1)(1-\beta)$.

 The p.g.f. $T(w;s,\nu)$ can be expressed as the ratio of two p.g.f.{}s, so that $N=N_1-N_2$ is the difference of two non-negative r.v.:
\begin{equation}\label{Tw}
T(w;s,\nu)=\frac{G(w;s,\nu)}{P(w;\nu)}
=\frac{(s-1)E_{s}(\nu(1-w))}{e^{\nu (w-1)} }\;;\;\;\; E_s(\nu)\!=\!\int_1^{\infty} e^{-\nu t}t^{-s}dt\;.
\end{equation}
Here $N_2$ is Poisson-distributed with the parameter $\nu$, and $N_1$ has p.g.f. $G(w;s,\nu)$ with p.m.f. and  function:
\begin{equation}\label{numTsal}
p_n=\frac{(s-1) \nu^n  E_{s-n}(\nu )}{n!}\;,\;\;
g(n)=\nu \frac{E_{s-n-1}(\nu )}{E_{s-n}(\nu )}
\approx \nu +1+\frac{n-s}{\nu }+\mathcal{O}\left(\nu^{-2} \right)
\end{equation}
Unlike $N_2$, the r.v. $N_1$ is not is infinitely divisible\footnote{$p_n$ is strongly peaked at $p_0=(s-1)E_s(\nu)$  where it violates the necessary condition $p_n<1/e$ of infinite divisibility, see Proposition \ref{infd_nec}.}, but for $\nu\to \infty$  its distribution converges  to the NBD with parameters $\theta=1/\nu$ and $k=\nu(\nu+1)-s$ (see Table \ref{tab:mark}), which is infinitely divisible. Both $N$  and $N_2$ are heavy-tailed for $1<s \leq 2$. 
 
To prove the scalability of the r.v. $N_1$, we express  $G(w;s,\nu)$ as a Laplace transform. Using the integral representation (\ref{Tw}) for $E_s(\nu)$, the p.d.f. $f(E)$ corresponds to a Pareto distribution with scale  parameter $\nu$ and shape parameter $s$:
\begin{equation}\label{Pareto2pgf}
(s-1)E_{s}(\nu(1-w))=\int_0^{\infty}f(E)e^{-(1-w)\nu E}dE\;,\;
\;f(E)=\frac{s-1}{\nu}\left(\frac{E}{\nu}\right)^{-s}\cdot \theta(E-\nu)\;.
\end{equation}
Expressing  $P(w;\nu)=e^{\nu (w-1)}=\int_0^{\infty}e^{ (w-1)E}\delta(E-\nu)dE$ ((Eq.\ref{Poislapl})) we observe
that convolving the Tsallis p.d.f. (\ref{Tcontex}) with
$\delta(E-\nu)$ yields the Pareto distribution:
\[
\frac{s-1}{\nu}\left[\frac{E}{\nu}\right]^{-s}=\int_0^{\infty}\frac{s-1}{\nu}\left[1+\frac{x}{\nu}\right]^{-s}\delta(x-(E-\nu))dx\;.
\]

\newpage
\section{Conclusions}
In this article, we investigate the connections -- relatively underexplored in the existing literature -- between discrete random variables and statistical mechanics. Our principal motivation is the formulation of an equation of state for small systems within a grand-canonical framework.

In the first part, we present both classical and novel results on power series and infinitely divisible distributions of discrete, non-negative random variables, and analyze their self-decomposability and scalability properties. We introduce a significant new class of Markovian distributions that naturally emerge in the study of stationary birth–death processes and scalable non-negative random variables.

In the second part, these probabilistic results are applied to problems in statistical mechanics. We demonstrate that, under the assumption of infinite divisibility of the particle-number distribution in the grand canonical ensemble, the cumulants can be related to the moments of an associated secondary distribution. Furthermore, we examine the relationship between the virial expansion and combinants.

A central focus of the work is the characterization of the particle-number distribution near the critical point of a first-order liquid–vapor phase transition. We formulate a simple kinetic model of this transition, in which a discrete stable distribution describes the number of particles (bubbles) at the critical point. We show that this distribution constitutes a fixed point of a renormalization semigroup transformation.

Finally, we explore particle-number distributions in non-extensive systems. We demonstrate that, in contrast to infinite divisibility, the scalability property of the particle-number distribution can persist even in the non-extensive regime. In addition, we discuss deformed discrete distributions and illustrate how, for example, the negative binomial distribution can be obtained as a deformation of the Poisson distribution.

\section*{Acknowledgements}
  One of the authors, M.\v{S}, would like to thank professors Petr Jizba and Ji\v{r}\'i Kolafa for reading the manuscript and for interesting discussions.
\newpage

\appendix
\section{Log-convex infinitely divisible distributions}\label{ExLX}
In this section, we provide an explicit illustration of the application of the log-convexity/log-concavity condition given in equation~(\ref{eq:cvx}).\begin{enumerate}[label=(\roman*), leftmargin=*]
\item 
\label{ex:zextSbd}
{\it The zero-inflated Extended Sibuya distribution} has the p.g.f. 
 \begin{equation}\label{zexsbd_pgf}
 Q(w;\alpha,\theta)=\alpha+(1-\alpha)G(w;\theta)\;,\; \alpha\in (0,1)\;,
 \end{equation}
where $G(w;\theta)$ is given by (\ref{exsbd_pgf}).
Both the Extended Sibuya and its zero-inflated version have for $n>0$ the same transition function 
$g(n)=(n+1)p_{n+1}/p_n=\theta(n-\gamma)$. The latter is non-decreasing; so for $n>0$, the p.m.f. $p_n$ is ULX, see (\ref{eq: ulc}). Applying the condition (\ref{eq:cvx}) with $n=1$ gives:
\[ p_1^2=\left(\frac{(1-\alpha)\theta\gamma}{1-(1-\theta)^{\gamma}}\right)^2
 \leq p_0 p_2 =\alpha  \frac{(1-\alpha)\theta\gamma}{1-(1-\theta)^{\gamma}}\theta(1-\gamma) \]
The p.g.f. (\ref{zexsbd_pgf}) is aninfinitely divisible for $\alpha$ satisfying the inequalities:
\begin{equation}
 1>\alpha\geq \frac{\gamma}{1-(1-\gamma)(1-\theta)^{\gamma}}\;.   
 \end{equation} 
\item \label{ex:zgSbd}
The two-parameter {\it Generalized Sibuya distribution} on $\N$ \cite{Kozu2017} with $p_{n+1}/p_n=(n+\nu-\gamma)/(n+\nu+1)$ and $\nu>\gamma>0$  admits an extension  to {\it zero extended Generalized Sibuya distribution} on $\Z $ via the identity
\[1=\sum_{n=0}^{\infty}p_n= \sum_{n=0}^{\infty}p_0\prod_{i=0}^n \frac{g(i)}{i+1}=p_0\frac{\nu-\gamma}{\gamma}.\] 
The condition $0<p_0=\gamma/(\nu-\gamma)<1$ restricts its existence to $\nu>2\gamma$. Infinite divisibility follows from log-convexity. Looking at its p.g.f.
\[H(w) =\gamma  \Gamma (\nu +1) \, _2\tilde{F}_1(1,-\gamma +\nu +1;\nu +2;w)\;,\]
where $_2\tilde{F}_1(a,b;c;z)$ is the regularized hypergeometric function, it is far from obvious how to write it as a compound Poisson p.g.f. 
\item \label{divPois}
Consider an infinitely divisible p.g.f. $H(w;\lambda_1)$ with the secondary p.g.f. $G(w)=\sum g_nw^n$ and Poisson p.g.f. $P(w;\lambda_2)$. Let us ask: When is the ratio of those p.g.f.{}s $R(w)=H(w,\lambda_1)/P(w;\lambda_2)$ infinitely divisible p.g.f.? Assuming $R(w)=P(Q(w);\lambda)$ with $Q(w)=\sum_{j>0}q_jw^j$, $R(w)$ is compound Poisson and hence infinitely divisible if $Q(w)$ is the secondary p.g.f. Writing
\begin{eqnarray}
\log R(w) =\lambda_1\left[G(w)-1-\frac{\lambda_2}{\lambda_1}(w-1)\right]=
(\lambda_1\!-\!\lambda_2)+\lambda_1\left[G(w)-\frac{\lambda_2}{\lambda_1} w\right]\\\nonumber
=(\lambda_1\!-\!\lambda_2)+\lambda_1\left[g_0 +\left (g_1-\frac{\lambda_2}{\lambda_1}\right) w +\sum_{j>1}g^jw^j\right]=
\lambda \left[\sum_{j>0}q_jw^j -1 \right]\;.
   \end{eqnarray}
Equating the powers of $w$ gives
\[
\lambda=\lambda_1+\lambda_2,\; \; q_j=\frac{\lambda_1}{\lambda}g_j,\; j\neq 1,\; \;
q_1=\frac{\lambda_1}{\lambda}g_1-\frac{\lambda_2}{\lambda}\;.
\]
Thus $Q(w)$ is a valid p.g.f. provided $\lambda_1>\lambda_2$ and $\lambda_1 q_1>\lambda_2$.
\end{enumerate}

\section{Scalable infinitely divisible distributions}\label{Exscal}

The following examples show the scalability of several infinitely divisible p.g.f.{}s, and incidentally uncover new infinitely divisible distributions on $\R$. 

\begin{enumerate}[label=(\roman*), leftmargin=*]
\item 
The Poisson p.g.f. is a Laplace transform of a delta measure
\begin{equation}\label{Poislapl}
P(w;\lambda)=e^{-(1-w)\lambda}=\int_0^{\infty}e^{ (w-1)x}\delta(x-\lambda)dx
\end{equation}
\item \label{NBDLap}
The NBD p.g.f. is a Laplace transform of the gamma p.d.f. $p(x;\alpha,\beta)$.
\begin{equation}\label{nbd2lap}
H_{NBD}(w;\alpha,\beta) = \left[1\!+\! \frac{1\!-\!w}{\beta}  \right]^{-\alpha} = \int_0^\infty e^{-(1-w)x}\frac{\beta^{\alpha}}{\Gamma(\alpha)}x^{\alpha-1}e^{-\beta x}dx\,.
\end{equation}

\item \label{DMLlap}
The p.g.f. of the compound geometric Sibuya distribution, the Discrete Mittag-Leffler distribution (DML) \cite{Pillai1995}
\begin{equation}\label{DMLpgf}
Q(w;\lambda,\gamma)=G[\mathcal{S}(w;\gamma);\lambda]=[1+\lambda(1-w)^{\gamma}]^{-1}
\end{equation}
is the Laplace transform of the heavy-tailed p.d.f. 
\begin{equation}\label{DMLpdf}
p(x;\lambda,\gamma)=
\frac{x^{\gamma-1}}{\lambda}E_{\gamma ,\gamma }\left(-\frac{x^{\gamma}}{\lambda}\right) \xrightarrow[x\to\infty]{}\frac{x^{\gamma -1}}{\lambda  \Gamma (\gamma )}\;,
\end{equation}
where $E_{\alpha ,\beta }(z)= \sum_{k=0}^{\infty}z^k/\Gamma(\alpha k +\beta)$ is the Mittag-Leffler function. Since DML is scalable, it has
a Markovian distribution. Note that this property is not obvious from its p.m.f. recurrence \cite{Pillai1995}:
\begin{equation}\label{DMLpmf}
p_0=\frac{1}{1+\lambda},\;p_1=\frac{\gamma \lambda}{(1+\lambda)^2},\;\;p_n=\frac{\lambda}{1+\lambda}\sum_{j=0}^n (-1)^{j-1}\binom{\gamma}{j}p_{n-j};\; n\geq 2\;.
\end{equation}
\item \label{DMLExtlap} 
The PSD generalization of DML p.g.f. (\ref{DMLpgf}):
\begin{equation}\label{DMLPSDpgf}
    Q(w;\lambda, \gamma, \theta) =
  \frac{1+\lambda(1-\theta)^{\gamma}}{1+\lambda(1-\theta w)^{\gamma}}\;,   
  \end{equation}
is infinitely divisible (see Proposition \ref{infdthin}) and scalable: it represents the Laplace transform of the p.d.f. 
\[
p(x;\lambda,\gamma, \theta)=(1+A)A\cdot \nu^{\gamma}x^{\gamma-1}E_{\gamma ,\gamma }\left(-A(\nu x)^{\gamma}\right);\;\;\nu=\frac{\theta}{1-\theta}\;,\; A=\lambda(1-\theta)^{\gamma}\;.
\]
\item \label{sechdist}
Let us discuss three p.g.f. connected to the hyperbolic secant function. 

(a) Consider the r.v. $N\in \Z$ with the p.g.f. 
\begin{equation}\label{sechpgf}
H(w) = \text{sech}(\sqrt{1-w})\;;\;\; \E N =1/2
\end{equation}
Using the well-known product formula \cite{AbSteg1972}
\begin{equation}\label{sech2prod}
\text{sech}(u) = \frac{1}{\cosh(u)}=\prod_{j=1}^{\infty}\frac{1}{1+u^2/(\pi (j-1/2))^2)}
\end{equation} 
and setting $u^2=(1-w)$ we can express the p.g.f. $H(w)$ as an infinite product of geometric p.g.f.{}s $G(w,\mu_j)$ with the mean $\mu_j$:
\begin{equation}\label{infprodG}
H(w)= \prod_{j=1}^{\infty}\frac{1}{1+(1-w)/(\pi (j-1/2))^2} 
=\prod_{i=1}^{\infty}\frac{1}{1+\mu_j (1-w)}
=\prod_{i=1}^{\infty}G(w,\mu_j)
\;.
\end{equation} 
Since geometric p.g.f. are infinitely divisible and scalable, so is their infinite product $H(w)$. Moreover, since the p.g.f.{}s $G(w;\mu_j)$ are PSDs, their product is also PSD. 

(b) Replacing  $u^2 \to \lambda(1-w)^{\gamma}$ with $0<\gamma< 1$ in (\ref{sech2prod}) gives a new p.g.f. $H(w;\gamma)$ expressible as an infinite product of DML p.g.f.{}s:
\begin{equation}\label{infprod1}
H(w;\lambda,\gamma)= \text{sech}\left( (\lambda(1-w))^{\gamma /2}\right)
=\prod_{i=1}^{\infty}\frac{1}{1+\lambda\mu_j (1-w)^{\gamma}}\;.
\end{equation}
The infinite divisibility and scalability of the DML distribution imply the same properties for $H(w;\gamma)$. 

(c) Infinite product of PSD DML p.g.f. (\ref{DMLPSDpgf}) yields infinitely divisible and scalable p.g.f.:
\begin{equation}\label{infprod2}
H(w;\lambda,\gamma,\theta)= 
\prod_{i=1}^{\infty}\frac{1+\lambda\mu_j (1-\theta)^{\gamma}}{1+\lambda\mu_j (1-\theta w)^{\gamma}} =
\frac{\cosh (z(\theta))}{\cosh (z(\theta w))}\;;\;\; z(\theta)=\sqrt{\lambda  (1-\theta )^{\gamma}}\;.
\end{equation}

\item \label{recipr}
Infinitely divisible PSD p.g.f. of Eq. \ref{jj1} is the Laplace transform of the p.d.f.
\begin{equation}
    p(x;\theta) =-\frac{e^x}{Z(\theta)\theta^2} \left(\theta  e^{-\frac{x}{\theta }}-(\theta +x) \text{Ei}\left(-\frac{x}{\theta }\right)\right)\;, \;\; \text{Ei}(x)=-\int_{-x}^{\infty}e^{-t}t^{-1} dt\;, 
\end{equation}
where Ei$(x)$ is the exponential integral function. 
\item 
The p.d.f. corresponding to the shifted logarithmic p.g.f.  (\ref{shlog}) is:
\begin{equation}\label{shlogLap}
p(x;\theta)=\frac{e^x \text{Ei}\left(-x/\theta\right)}{\log (1-\theta )}\,.   
\end{equation}
\item \label{shexlap}
The p.d.f. corresponding to the shifted Extended Sibuya p.g.f. (\ref{shexsbd}) is:
\begin{equation}
p(x;\theta,\gamma)=-\frac{e^x \Gamma \left(-\gamma ,x/\theta\right)}{\left(1-(1-\theta )^{\gamma }\right) \Gamma (-\gamma )}\;;\;\;     \Gamma(a,z)=\int_z^{\infty}t^{z-1}e^{-t} dt\; .
\end{equation}

\end{enumerate} 
\begin{Remark}
Except for the gamma distribution with $k>1$ and PSD DML, all other obtained p.d.f.s are strictly monotonous. Using the fact that the sum of several \iid r.v. $N_1+N_2+\ldots,\; N_i \in \Z$ each having a maximum at $n=0$, can have a maximum at $n>0$, one can obtain a p.d.f. with a maximum at $x>0$ as an the inverse the Laplace transform of the new p.g.f. $(H(w))^j,\; j>1, j\in \N$ . For example, the Shifted Logarithmic distribution (\ref{shlogLap}) raised to the power $j=3$ has a maximum at $n=1$  for $0.857\gtrsim \theta>2/3$  and at $n=2$ for $0.933\gtrsim \theta\gtrsim 0.857$. 
\end{Remark}

\section{The virial expansion of particle-number p.g.f.{}s}\label{ExEoS}
The following examples show the virial EoS of several p.g.f.{}s.
\begin{enumerate}[label=(\roman*), leftmargin=*]
\item[]
\item For the {\it Hermite p.g.f.} \ref{Hermite_pgf} with $\lambda_1=\lambda(1-\alpha)$ and $\lambda_2=\lambda\alpha$ we have:
\[
\beta pV=\lambda(1-\alpha)\xi +\lambda \alpha \xi^2,\;\; V\rho=\lambda(1-\alpha+2\alpha \xi)\;,\; \xi=\frac{V \rho -\lambda(1-\alpha)}{2\alpha\lambda}\;.
\]
The second-order polynomial expansion in the mean $N = V\rho$ gives the EoS:
\begin{equation}\label{HermEoS}
 \beta p \approx \rho-a\rho^2\;,\;\; a=\frac{\alpha V}{(\alpha-1)^2 \lambda }\;.   
\end{equation} 
 Notably, for $\lambda\sim V$ it reduces to the van der Waals EoS (\ref{vdWEoS}) without the excluded volume correction ($b=0$). 
\item The {\it Conway–Maxwell–Poisson} (CMP) distribution with the partition function 
$\mathcal{Z}(\xi)=(I_0\left(2 \sqrt{\xi }\right)-1)/\xi$ of (see Table (\ref{tab:mark})) has a smaller variance than the Poisson distribution. Its virial EoS up to the 6th order in $\rho$ reads:
\[
\beta p \approx 
\rho  +\frac{\rho ^2 V}{18}+\frac{\rho ^3 V^2}{81}
-\frac{263 \rho ^4 V^3}{72900} 
+\frac{139 \rho ^5 V^4}{164025}
-\frac{3611 \rho ^6 V^5}{9644670}+\ldots
\]
Note that the second and third virial coefficients $B_2$ and $B_3$ are always positive, implying repulsion between two and three particles.
\medskip
\item The  {\it Hyper-Poisson} PSD, or the displaced Poisson distribution (see Table (\ref{tab:mark})) with $\mathcal{Z}(\xi)=\,_1F_1(1,\nu,\xi)$ yields: 
\[ \beta p\approx \rho  
-\frac{(\nu -1) \rho ^2 V}{2 (\nu +1)}+
\frac{(\nu -1) (\nu  (\nu +5)-2) \rho ^3 V^2}{3 (\nu +1)^2 (\nu +2)}+\mathcal{O}(\rho^4V^3)\;.
\]
This distribution is {\it sub-Poisson} for $\nu<1$  and {\it super-Poisson} for $\nu>1$. Therefore, the second virial coefficient $B_2$ is positive/negative for $\nu<1$/$\nu>1$.
\medskip
\item For the p.g.f. of Eq.\ref{jj1} with
\[ V\rho=
\frac{2 \xi -(\xi -2) \log (1-\xi )}{(\xi -1) \log (1-\xi )-\xi } \approx \frac{2 \xi ^2}{9}+\frac{\xi }{3};\;
\;\xi\approx \frac{3}{4} \left(\sqrt{8 V \rho +1}-1\right)
\]
we obtain the following virial expansion:
\[\beta p \approx -\frac{\log (2)}{V} +\rho -\rho ^2 V
+\frac{14 }{15}\rho ^3 V^2-\frac{7 }{8}\rho ^4 V^3
+\frac{146 }{175}\rho ^5 V^4
-\frac{97}{120} \rho ^6 V^5 + \ldots
\]
The positive absolute term due to a non-zero limit $\lim_{\xi\to 0^+} \mathcal{Z}(\xi)=1/2$ describes the {\it vacuum pressure} exerted on a bubble from the surrounding medium.
\medskip
\item The {\it shifted Extended Sibuya distribution} of Eq. \ref{shexsbd} with 
\[\mathcal{Z}(\xi)=\frac{1-(1-\xi)^{\gamma}}{\xi}\;,\;\;
V\rho =\frac{\gamma  \xi  }{(1-(1-\xi )^{\gamma })(1-\xi )^{1-\gamma}}-1\;,\;\; 0<\gamma<1\]
has the virial EoS containing the {\it negative vacuum pressure}:
\[\beta p\approx 
\frac{\log (\gamma )}{V}+\rho  -\frac{(5-\gamma) \rho ^2 V}{6 (1-\gamma)}+
\frac{((\gamma -4) \gamma +7) \rho ^3 V^2}{9 (1-\gamma )^2}
- \ldots;\;\;\]
 
\end{enumerate}
\section{Particle-number distributions derived from a given EoS}\label{RecEoS}
Here, we discuss particle-number distributions of two well-known EoS.
\begin{enumerate}[label=(\roman*), leftmargin=*]
\item[]
\item The {\it Carnahan-Sterling} EoS \cite{Carnahan1969} 
\begin{equation}\label{CS_EoS}
\beta p=\rho \frac{1+\eta+\eta^2-\eta^3}{(1-\eta)^3}=\rho+\sum_{j>1}^{\infty} \left(j^2+j-2\right) \left(\frac{v}{4}\right)^{j-1}\rho^j
\end{equation}
is valid for $\eta=v\rho/4<1$, where $v$ is the packing fraction.
Writing $\lambda_j=a_j(v/4)^{j-1}\lambda^j$ and truncating at $J=6$ gives the coefficients \(a_j\):
\[ a_1=1,\; a_2=-4,\;a_3=27,\;a_4\approx -227,\;a_5
\approx 2164,\;a_6\approx -22178\;.
\]
The particle number \(N = N_1 - N_2\) represents the difference of two infinitely divisible random variables \(N_1\) and \(N_2\) with probability generating functions \(H_1(w)\) and \(H_2(w)\), which describe, respectively, the gas of clusters containing an odd and even number of particles.
\item {\it van der Waals} (vdW) EoS:
\begin{equation}\label{vdWEoS}
\beta (p+ a\rho^2)(1-b \rho)=\rho\;,\;\;
\beta p  
=\frac{\rho}{1-b\rho}-\beta a \rho^2
\end{equation}
is valid for $b\rho<1$.The constants \(a\) and \(b\) account, respectively, for the mutual attractive interaction between particles and for their finite excluded volume. 


Requiring that the solution of the quadratic equation for \(\rho\), obtained from (\ref{vdWEoS}) by setting \(p=0\), possess real roots leads to the lower bound of the applicability of the vdW equation: $T_0>a/(4b)$.
An additional constraint arises from the thermodynamic stability condition 
\(\partial p/\partial\rho > 0\), which is automatically fulfilled in the domain \(\tau = T/T_c > 1\), i.e., above the inflection point of the function \(p(\rho)\):
\[
\rho_c = \frac{1}{3b}\;, \qquad 
p_c = \frac{a}{27 b^2}\;, \qquad
T_c = \frac{8a}{27 b}\;.
\]
Writing the condition $\partial p/\partial \rho=0$  in dimensionless variables $\tau$ and $x=b\rho$:
\begin{equation}\label{xstab}
x(1-x)^2=\frac{4\tau}{27}    
\end{equation}
the condition $\tau=1$ yields the inflection point at $x=1/3$. 
While for $\tau>1$ (\ref{xstab}) has no solution with $\partial p/\partial\rho>0$ everywhere, and therefore no instability,
for $\tau<1$ and $x\in (0, 1)$ it has two roots $x_1<1/3<x_2$ which bound the mechanically unstable region where  $\partial p/\partial\rho<0$ \footnote{ 
 The fluid EoS in this region is described by Maxwell’s construction which replaces the curve $p(\rho)$ by a constant pressure line:  if $T$ is fixed, $p$ is fixed too.}.

Before examining the full vdW equation let us analyze the two limiting cases:
\begin{itemize}
    \item[\(b=0:\)] This regime is relevant for densities \(\rho \ll 1/b\). The corresponding equation of state,
    \[
    \beta p = \rho - \beta a \rho^2,
    \]
    comes from the Hermite probability generating function given in (\ref{HermEoS}). The conditions \(p>0\) and \(\rho>0\) restrict its applicability to
    \[
    1 - \beta a > 0 \quad \Rightarrow \quad \beta < \frac{1}{a}.
    \]

    \item[\(a=0:\)] This case describes the high temperature limit $\beta \to \infty$ of  (\ref{vdWEoS}):
    \[
    \beta p = \frac{\rho}{1-b\rho}.
    \]
    The coefficients \(\alpha_j\) parameterizing the combinants \(\lambda_j = \alpha_j \lambda^j\) up to \(J=5\) are:
    \[
    \alpha_1 = 1,\quad
    \alpha_2 = -b,\quad
    \alpha_3 = \frac{3}{2} b^2,\quad
    \alpha_4 = -\frac{8}{3} b^3,\quad
    \alpha_5 = \frac{125}{24} b^4,\;\ldots
    \]
    The particle number \(N = N_1 - N_2\) represents the difference of two infinitely divisible random variables \(N_1\) and \(N_2\) with probability generating functions \(H_1(w)\) and \(H_2(w)\), which represent, respectively, the gas of clusters comprising an odd and an even number of particles. The total p.g.f.\ can be written as
    \[
    H(w)
    = \frac{\exp\!\Big[\lambda(w-1) + \frac{3}{2} b^2 \lambda^3 (w^3-1) + \frac{125}{24} b^4 \lambda^5 (w^5-1)\Big]}
           {\exp\!\Big[b\lambda^2 (w^2-1) + \frac{8}{3} b^3 \lambda^4 (w^4-1)\Big]}
    = \frac{H_1(w)}{H_2(w)}.
    \]
\end{itemize} 
For the general vdW EoS (\ref{vdWEoS}) we get the following set of the coefficients of $\lambda_j=\alpha_j \lambda^j$:
\[ \alpha_1=1\;,\;
\alpha_2=(a \beta-b ),\; \alpha_3=\frac{1}{2} (b-2 a \beta ) (3 b-2 a \beta )\;, 
\]
\[
\alpha_4=\frac{1}{3}\left(16 a^3 \beta ^3-48 a^2 \beta ^2 b+27 a \beta  b^2-2 b^3\right)\;,
\]
\[
\alpha_5=\frac{1}{24}\left(400 a^4 \beta ^4-1600 a^3 \beta ^3 b+2040 a^2 \beta ^2 b^2-944 a \beta  b^3+125 b^4\right)
\]
The sign of each combinant \(\alpha_j\) for \(j>1\) depends explicitly on the temperature. Although, in general, one cannot expect exact infinite divisibility, this temperature dependence renders an approximate infinite divisibility attainable. In particular, at sufficiently low temperatures, specifically for
\[
T_0 < T \lesssim \frac{4a}{11b},
\]
the first four combinants become positive. This approximate infinite divisibility is further constrained to particle densities satisfying
\[
\rho < \rho_{\max} = \frac{x_1}{b}, \qquad x_1 \in (0,1/3).
\]

Here $x_1$ is the smallest root of equation (\ref{xstab}). Using the Maxwell’s construction this can probably be extended even to densities where the pressure stays constant. The simultaneous constancy of the pressure and temperature means that the evolution of the particle number density is governed by the equation 
\[
p\beta =\frac{\rho }{1-b\rho} - \beta a \rho^2 =const.
\]
where $\beta$ is also some constant.
 
An in-depth mathematical analysis of the van der Waals equation can be found in \cite{Prodanov2022}.

\end{enumerate}

\newpage
\bibliographystyle{unsrt}
\bibliography{main}
\end{document}